\documentclass[11pt,twoside,a4paper]{report}
\pdfoutput=1

\usepackage{amsthm}
\usepackage{amsmath}
\usepackage{amssymb}
\usepackage{graphicx}
\usepackage[linesnumbered,boxed,noline]{algorithm2e}
\usepackage{algorithmic}
\usepackage{tabularx}
\usepackage{aliascnt}
\usepackage[pdftex,
	bookmarks=false,
	colorlinks,
	linkcolor=blue,
	citecolor=magenta,
	pdftitle={PhD Thesis: Approximation Algorithms for Covering and Packing Problems on Paths},
	pdfauthor={Arindam Pal},
	pdfsubject={Algorithms, Mathematics, Theory},
	pdfstartview=FitH,
	pdfstartpage=1,
	pdfdisplaydoctitle=true,
]{hyperref}

\usepackage{url}
\usepackage{verbatim}
\usepackage{setspace}

\usepackage{fancyhdr}
\fancypagestyle{plain}{%
\fancyhead{} 
}

\def \thesis {1}

\ifx \thesis \undefined
\newtheorem{theorem}{Theorem}
\newtheorem{claim}{Claim}
\newtheorem{conjecture}{Conjecture}
\newtheorem{cor}{Corollary}
\newtheorem{prop}{Proposition}
\newaliascnt{lemma}{theorem}
\newtheorem{lemma}[lemma]{Lemma}
\aliascntresetthe{lemma}
\newtheorem{sublemma}{Sublemma}
\newtheorem{corollary}{Corollary}
\newtheorem{proposition}{Proposition}
\newtheorem{fact}{Fact}
\newtheorem{observation}{Observation}
\newtheorem{definition}{Definition}
\newtheorem{openq}{Open Question}
\else
\newtheorem{theorem}{Theorem}[chapter]

\newaliascnt{lemma}{theorem}
\newtheorem{lemma}[lemma]{Lemma}
\aliascntresetthe{lemma}

\newtheorem{observation}{Observation}[chapter]
\newtheorem{definition}{Definition}[chapter]

\fi

\newcommand{\vs}{{\vspace *{0.2 in}}}

\newcommand{\edpmax}{\textsc{Max-EDP}}
\newcommand{\ufpbag}{\textsc{Bag-UFP}}
\newcommand{\ufpmax}{\textsc{Max-UFP}}

\newcommand{\ufpround}{\textsc{Round-UFP}}
\newcommand{\ufprounduniform}{\textsc{Round-UFP-Uniform}}

\newcommand{\oic}{\textsc{Online Interval Coloring}}

\newcommand{\RR}{\mathcal{R}}
\newcommand{\eps}{\epsilon}
\newcommand{\inv}[1]{\frac{1}{#1}}
\newcommand{\col}{{\tt col}}
\newcommand{\D}{{\cal D}}
\newcommand{\cl}{{\tt cl}}
\newcommand{\profit}{{\tt profit}}
\newcommand{\lar}{{\tt l}}
\newcommand{\sm}{{\tt s}}

\newcommand{\nba}{\emph{no-bottleneck assumption}}
\newcommand{\roundufp}{\textsc{Unsplittable Flow Problem with Rounds}}
\newcommand{\maxufp}{\textsc{Unsplittable Flow Problem}}
\newcommand{\bagufp}{\textsc{Unsplittable Flow Problem with Bag Constraints}}
\newcommand{\resjob}{\textsc{Resource Allocation for Job Scheduling Problem}}

\newcommand{\junk}[1]{}

\newcommand{\eat}[1] {}

\newcommand{\alg}{{\textrm{ALG}}}
\newcommand{\opt}{{\textsc{OPT}}}

\newcommand{\J}{{\cal J}}

\newcommand{\I}{{\cal I}}
\newcommand{\E}{{\cal E}}

\renewcommand{\S}{{\cal S}}

\renewcommand{\O}{{\cal O}}

\newcommand{\ResAll} {{\sc ResAll}}
\newcommand{\ZeroOneResAll} {{\sc (0-1)-ResAll}}

\newcommand{\PResAll}{{\sc PartialResAll}}
\newcommand{\PCResAll}{{\sc PrizeCollectingResAll}}

\newcommand{\lspc}{{\sc LSPC}}
\newcommand{\smfc}{{\sc SMFC}}

\newcommand{\calI} {{\cal R}}

\newcommand{\cM}{{\cal M}}

\newcommand{\cC}{{\cal C}}
\newcommand{\cJ}{{\cal J}}

\newcommand{\lptr}{{\em l-ptr}}
\newcommand{\rptr}{{\em r-ptr}}
\newcommand{\ljob}{{\em l-job}}
\newcommand{\rjob}{{\em r-job}}

\newcommand{\sh} {{\rm sh}}
\newcommand{\calL} {{\cal L}}
\newcommand{\calT} {{\cal T}}
\newcommand{\calM} {{\cal M}}
\newcommand{\cA} {{\cal A}}
\newcommand{\calA} {\cA}
\newcommand{\cB} {{\cal B}}
\newcommand{\DP} {{\rm DP}}
\newcommand{\ceil}[1] {\lceil #1 \rceil}
\newcommand{\wh}[1] {\widehat{#1}}

\DeclareMathOperator*{\argmin}{arg\,min}

\begin{document}
\doublespacing
\pagestyle{empty}
\begin{titlepage}

\begin{center}

\Large \textbf{\MakeUppercase{Approximation Algorithms for Covering and Packing Problems on Paths}}\\ 

\vspace{4cm}

\large

\textbf{\MakeUppercase{Arindam Pal}} 

\vspace{8cm}
\hspace{0cm}
\includegraphics[width=1in]{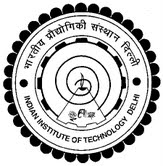}
\vspace{1cm}

\large{\textbf{DEPARTMENT OF COMPUTER SCIENCE AND ENGINEERING}}\\
\large{\textbf{INDIAN INSTITUTE OF TECHNOLOGY DELHI}}\\
\large{\textbf{NOVEMBER 2012}}\\

\end{center}

\end{titlepage}

\newpage \ \newpage
\pagestyle{empty}
\begin{center}
\Large{\textbf{\MakeUppercase{Approximation Algorithms for Covering and Packing Problems on Paths}}} 

\vs \vs \vs \vs \vs
by \\
\large{\textbf{\MakeUppercase{Arindam Pal}}} \\
\large{\textbf{Department of Computer Science and Engineering}}\\
\vs
\vs
\vs \vs \vs
Submitted \\ in fulfillment of the requirements of the degree of \\
{\bf Doctor of Philosophy} \\
\vs \vs \vs \vs \vs
to the \\
\vs
\begin{figure}[h]
\begin{center}
\includegraphics[width=1in]{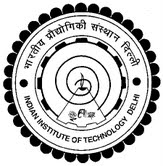}
\end{center}
\end{figure}
\vs
\bf{\Large{Indian Institute of Technology Delhi \\
November 2012}} \\
\end{center}
\pagebreak

\newpage


\newpage \ \newpage
\vspace*{\fill}
\begingroup
\begin{center}
\textsf{\Large{To my family for their love, support and patience}}
\end{center}
\endgroup
\vspace*{\fill}

\newpage \ \newpage
\begin{center}
{\Huge \textbf{Certificate}}
\end{center}
\ \\ \ \\
\ \\ \ \\

\noindent This is to certify that the thesis titled \textbf{Approximation Algorithms for Covering and Packing Problems on Paths} being submitted by Arindam Pal for the award of the degree of
\textsf{Doctor of Philosophy} in Computer Science and Engineering is a record of original bonafide research work carried out by him under our guidance and supervision at the Department of Computer Science and Engineering, Indian Institute of Technology Delhi. The results contained in this thesis have not been submitted in part or full to any other university or institute for the award of any degree or diploma.

\ \\ \ \\  \ \\ \ \\

\begin{table}[h]
\begin{center}
\begin{tabular}{ccc}
\textbf{Naveen Garg} & \hspace{1cm} &\textbf{Amit Kumar}\\
Professor & & Professor \\
Department of & & Department of \\
Computer Science and Engineering & & Computer Science and Engineering \\
Indian Institute of Technology Delhi & & Indian Institute of Technology Delhi \\
\end{tabular}
\end{center}
\end{table}

\newpage \ \newpage
\begin{center}
{\Huge \textbf{Acknowledgments}}
\end{center}
\ \\ \ \\
I am grateful to my doctoral research advisors Professor Amit Kumar and Professor Naveen Garg for being a constant source of inspiration. They gave me the freedom to work on problems that I like. They patiently listened to my ideas, even when some of them were not so great. They also gave valuable suggestions to improve the ideas and gave many ideas of their own. I am indebted to them for their help and advice during my graduate studies. This thesis would not have been possible without their cooperation.

I learned a great deal from my professors -- Sandeep Sen, Amitabha Bagchi and Ragesh Jaiswal. The courses taught by them (along with my advisors) built the foundation of my research on different topics of mathematics and theoretical computer science. In addition, I had the good fortune of working with them on some research problems, which also shaped my thoughts on theoretical research.

I would also like to thank my coauthors Sambuddha, Venkat, Yogish, Prashant and Saurav. Without their help, many of the results in my thesis would not have seen the light of the day. Thank you very much for all those discussions and ideas that we had over the last few years.

I spent a great time at IIT Delhi with my friends Muralidhara, Rudra, Ayesha, Syamantak, Anamitra, Shibashis, Pravesh, Brojeswar, Manoj, Anuj, Chinmay, Swati and Mona. Sorry if I forgot anyone's name. You gave me company during both good and bad times. I will fondly remember and cherish those moments in the years to come.

I thank my good friends Dinesh, Vinay and Gopalda for giving me company. Special thanks to Dinesh for those intensive discussions at IBM Research and at my home. Thanks to Vinay for the good times we had in JNU. Thanks also to Gopalda for being such a good friend and mentor. I will always remember the good moments that I spent with all of you.

My special thanks to Roger Federer for playing such great game of tennis and providing quality entertainment over the last 10 years. You showed me by your own example, that hard work, determination, dedication and discipline can help a man to reach great heights. I love you Roger!

Last but not the least, I would like to thank my family members -- my uncle Dr Debi Prasad Pal (Jatha), late father Bani Prasad (Baba), mother Bani (Ma), wife Sushmita, sister Anindita (Didi), brother-in-law Kanai (Dada), and nephews Arkajyoti and Debajyoti (Bhagna). Without their moral support, encouragement and cooperation, this thesis would not have been possible. I affectionately dedicate this thesis to them.

\begin{center}
{\Huge \textbf{Abstract}}
\end{center}
\ \\ \ \\
Routing and scheduling problems are fundamental problems in combinatorial optimization, and also have many applications. Most variations of these problems are NP-Hard, so we need to use heuristics to solve these problems on large instances, which are fast and yet come close to the optimal value. In this thesis, we study the design and analysis of approximation algorithms for such problems. We focus on two important class of problems. The first is the {\maxufp} and some of its variants and the second is the {\resjob} and some of its variants. The first is a \emph{packing} problem, whereas the second is a \emph{covering} problem.

In the {\maxufp}, we are given a path or a tree, each edge of which has a capacity. We are also given a set of requests, each of which has a start vertex, an end vertex, a demand and a profit. The objective is to select a subset of requests so as to maximize the total profit, subject to the condition that on every edge the total demand of the selected requests is at most it's capacity. We also study variants of this problem such as {\roundufp} and {\bagufp}. We give constant factor approximation algorithms for all of these problem on paths and trees under the {\nba}. We also give a constant factor competitive algorithm for the {\oic} problem.

In the {\resjob}, the timeline is divided into a set of discrete timeslots.. We are given a set of jobs, each of which has a start time, an end time and a demand requirement. We are also given a set of resources, each of which has a start time, an end time, a capacity  and a cost. A feasible solution is a set of resources satisfying the constraint that at any timeslot, the sum of the capacities offered by the resources is at least the demand required by the jobs active at that timeslot, \emph{i.e.}, the selected resources must cover the jobs. The objective is to select a subset of resources of minimum cost, which will cover all the jobs. This is called the \emph{resource allocation problem} ({\ResAll}). We consider the partial covering version ({\PResAll}) and the prize-collecting version ({\PCResAll}) of this problem. We give an $O(\log (n+m))$-approximation algorithm for the {\PResAll} problem, where $n$ is the number of jobs and $m$ is the number of resources respectively. We also give a $4$-approximation algorithm for the {\PCResAll} problem.


\pagestyle{fancy}

\tableofcontents
\newpage \ \newpage
\listoffigures
\newpage \ \newpage
\listoftables


\newpage \ \newpage
\ifx \thesis \undefined
\section {Introduction}
\else
\chapter {Introduction}
\label {chap1}
\fi

In this thesis, we study several important classes of covering and packing problems restricted to paths. In the class of covering problems, each edge (or a consecutive set of edges) of a path has a \emph{demand}, and we would like to allocate resources to meet the demands under various constraints. We broadly call this class of problems \emph{resource allocation for job scheduling}. In the packing scenario, we consider the problems where each edge has a \emph{capacity}, and we would like to route demands under these constraints. We broadly call this class of problems \emph{routing problems in communication networks}.

Many combinatorial optimization problems which are \textsc{NP-Hard} on general graphs remain \textsc{NP-Hard} on paths. A path is a natural setting for modeling many applications, where a limited resource is available and the amount of the resource varies over time. Many routing and scheduling problems fit into this framework. For packing problems, we can represent time instants as vertices, time intervals as edges and the amount of resource available in a time interval as the capacity of the corresponding edge. The requirement of a resource between two time instants can be represented as a demand between the corresponding vertices with a certain profit associated with it. Similarly for covering problems, we can think of the time interval between two time instants as jobs, whose demands must be satisfied on every time interval on their span by the resources.

\section {Preliminaries}
In this section, we define the notation and terminology that we will use throughout
the thesis. We work with undirected graphs, unless stated otherwise. We begin with some definitions.

\subsection {Approximation algorithms and approximation factors}
Since almost all the problems considered in this thesis are NP-hard, it is unlikely that there exist polynomial-time algorithms to compute the optimal solution for them. So, our goal will be to compute an approximate solution, which is close to the optimal solution. An \emph{$\alpha$-approximation algorithm} for an optimization problem $\Pi$ is a polynomial-time algorithm that for all instances of the problem produces a solution whose value is within a factor of $\alpha$ of the value of an optimal solution for that instance. If $\alg(I)$ is the value of the solution computed by an algorithm and $\opt(I)$ is the value of the optimal solution on input instance $I \in \Pi$ then, $\opt(I) \le \alg(I) \le \alpha \cdot \opt(I)$ (for \emph{minimization} problems) or $\opt(I) \ge \alg(I) \ge \alpha \cdot \opt(I)$ (for \emph{maximization} problems) for every instance $I$. The number $\alpha$ is called the \emph{approximation factor} of the algorithm.

\subsection {Online algorithms and competitive ratios}
In the online setting, data arrives over time, and at each point of time the algorithm has to maintain a solution for the data that has already arrived. In contrast, an \emph{offline algorithm} has the entire input available for processing. Often, we can't hope to compute the optimal solution without seeing the whole input in advance. Let $\sigma$ be an input sequence and let $\alg(\sigma)$ and $\opt(\sigma)$ be the costs of the solution of the algorithm and the optimal offline solution on $\sigma$. An online algorithm is \emph{$\rho$-competitive} if for every sequence $\sigma$, $\opt(\sigma) \le \alg(\sigma) \le \rho \cdot \opt(\sigma)$ (for \emph{minimization} problems) or $\opt(\sigma) \ge \alg(\sigma) \ge \rho \cdot \opt(\sigma)$ (for \emph{maximization} problems). The number $\rho$ is called the \emph{competitive ratio} of the algorithm.

\section {Routing problems in communication networks}
A \emph{communication network} consists of \emph{nodes} communicating with each other through a set of \emph{links} interconnecting these nodes. We can think of these nodes as transmitters and receivers and the links as channels. Each channel has some \emph{capacity} or \emph{bandwidth}. A fundamental problem in communication networks is to allocate bandwidth and assign paths to connection requests. A \emph{connection request} consists of two nodes called its \emph{source} and \emph{destination}. There is a \emph{demand} associated with the request. The objective is to allocate bandwidth on some path from source to destination to satisfy the demand. Since there are several requests, it may not be possible to satisfy all demands without exceeding the capacities of some channels.

Most of these problems can be modeled as variants of the \emph{multicommodity flow problem} in a graph. Here, we are given a graph $G = (V,E)$, where $V$ is the set of vertices and $E$ is the set of edges. Let $n = |V|$ be the number of vertices and $m = |E|$ be the number of edges of the graph. Each edge $e \in E$ has a \emph{capacity} $c_e \equiv c(e)$. We are also given a set of \emph{requests} $\RR = \{R_1,\ldots,R_k\}$. Each request $R_i$ has a \emph{source vertex} $s_i$, a \emph{destination vertex} $t_i$ and a bandwidth \emph{demand} $d_i$. Sometimes, there is also a \emph{profit} $w_i$ associated with $R_i$. The goal is to route the requests without violating any edge capacity. This is the \emph{feasibility} condition. The objective function that we want to optimize varies for different problems. Here are some natural objective functions.

\begin{enumerate}
	\item What is the \emph{maximum} number of requests that can be satisfied feasibly?
	\item What is the \emph{minimum} number of rounds required to satisfy \emph{all} requests, so that in every round the set of requests that are satisfied are feasible?
\end{enumerate}

\subsection {Notations}
We summarize the symbols we will use along with their meanings in \autoref{symtab} on page \pageref{symtab}.

\begin{table}[ht]
\begin{center}
\begin{tabular}{|c|l|}
\hline
\textbf{Symbol} & \textbf{Explanation} \\
\hline
$c_{\max},c_{\min}$ & Maximum and minimum capacities. \\
$d_{\max},d_{\min}$ & Maximum and minimum demands. \\
$w_{\max},w_{\min}$ & Maximum and minimum profits. \\
$\alpha$ & Expansion of a graph. \\
$\Delta$ & Maximum degree of a graph. \\
$\omega$ & Maximum clique size of a graph. \\
$r$ & Maximum edge congestion of a graph. \\
\hline
\end{tabular}
\caption{Notations used in the thesis}
\label{symtab}
\end{center}
\end{table}

\subsection {Problem definition and motivation}
We now define the various problems that we will study.

\noindent {\edpmax} (\textsc{Maximum Edge-Disjoint Paths Problem})\\
\textsf{Input:} Graph $G = (V,E)$, requests $\RR = \{(s_i,t_i):i=1,\ldots,k\}$.\\
\textsf{Output:} A feasible subset of requests $S \subseteq \RR$ along with a path $P_i$ connecting $(s_i,t_i)$ for all $i \in S$, such that $P_i$ and $P_j$ are edge-disjoint for $i \ne j$.\\
\textsf{Objective:} Maximizing the number of feasible requests $|S|$.\\

\noindent {\ufpmax} (\textsc{The Unsplittable Flow Problem})\\
\textsf{Input:} Graph $G = (V,E,c)$, requests $\RR = \{(s_i,t_i,d_i,w_i):i=1,\ldots,k\}$.\\
\textsf{Output:} A feasible subset of requests $S \subseteq \RR$ along with a path $P_i$ connecting $(s_i,t_i)$ for all $i \in S$.\\
\textsf{Objective:} Maximizing the total profit $\sum_{i \in S} w_i$.\\

\noindent {\ufpround} (\textsc{Unsplittable Flow Problem with Rounds})\\
\textsf{Input:} Graph $G = (V,E,c)$, requests $\RR = \{(s_i,t_i,d_i):i=1,\ldots,k\}$.\\
\textsf{Output:} Partition $\RR$ into a number of sets such that each set is feasible.\\
\textsf{Objective:} Minimizing the total number of sets.\\

\noindent {\ufpbag} (\textsc{Unsplittable Flow Problem with Bag Constraints})\\
\textsf{Input:} Graph $G = (V,E,c)$, \emph{bags} of requests $\RR^1,\ldots,\RR^p$, where each bag $\RR^j = \{(s^j_i,t^j_i,d^j_i):i=1,\ldots,k\}$ has a profit $w^j$.\\
\textsf{Output:} A subset of bags $B$ and at most one request from each bag along with the paths for all selected requests such that the set of requests are feasible.\\
\textsf{Objective:} Maximizing the total profit $\sum_{\RR^j \in B} w^j$.\\

We will study these problems when the input graph is a path or a tree. Note that there is a unique path between any two vertices, and so we only need to figure out which requests to choose. These problems when restricted to a path can also be used for modeling a time-varying resource. For each time $t$, we have a vertex. The capacity of the edge $(t,t+1)$ denotes how much resource is available.

\begin{figure}[ht]
\begin{center}
\includegraphics[scale=0.7]{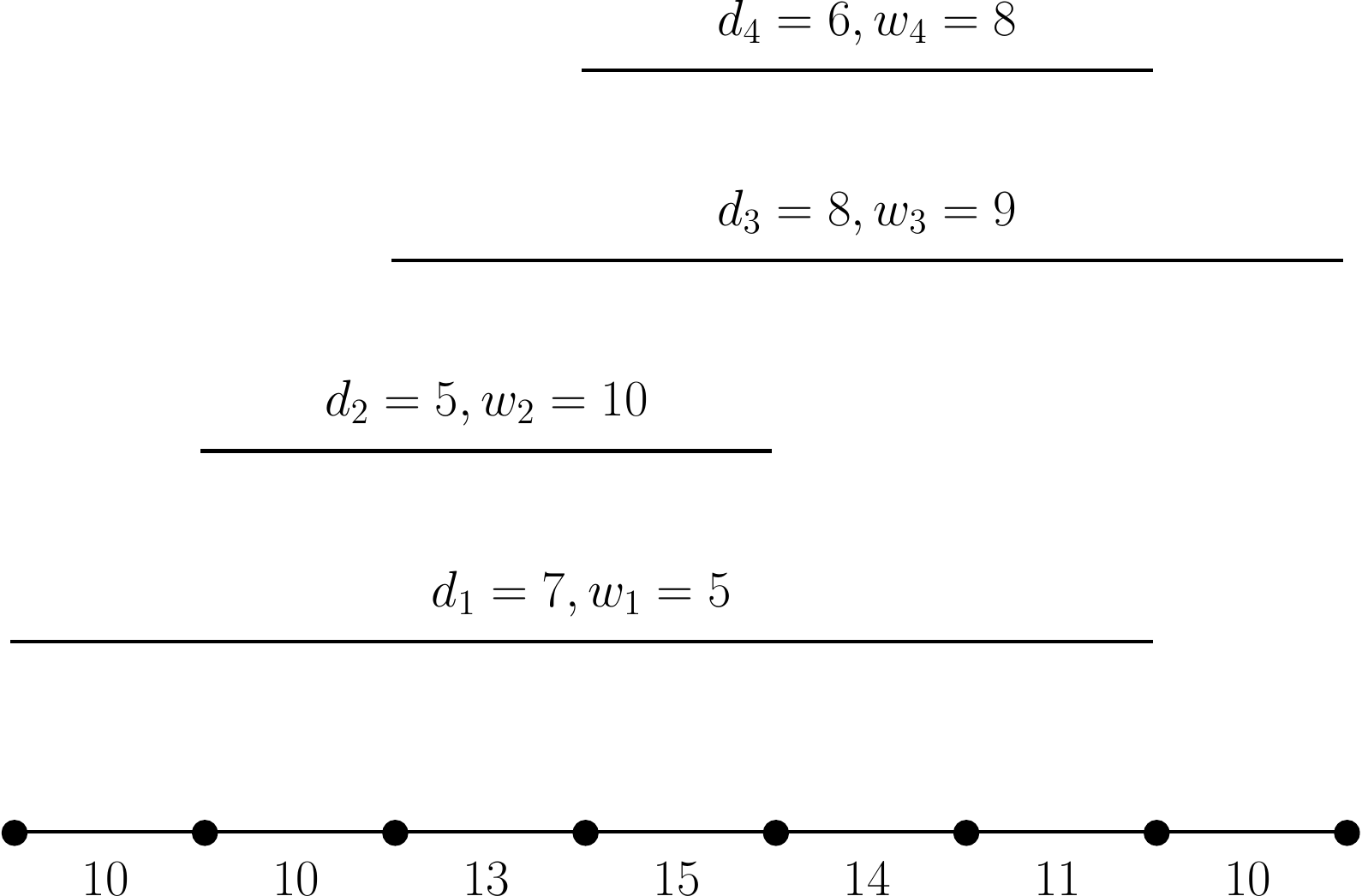}
\caption{A sample {\ufpmax} instance}
\label{ufp-example}
\end{center}
\end{figure}

In {\ufpmax}, there are a set of users who want to use different amounts of this resource over different time intervals and are ready to pay for this. The goal is to select a subset of these users to maximize the profit, while satisfying the resource availability constraint at each instant, \emph{i.e.}, the total demand of selected users at any instant does not exceed the resource available. An example of {\ufpmax} is shown in \autoref{ufp-example}.

The concept of \emph{bag constraints} (at most one request can be selected from each bag) in {\ufpbag} is quite powerful. Apart from handling the notion of release time and deadline, it can also work in a more general setting where a job can specify a set of possible time intervals where it can be scheduled. Moreover, it allows for different instances of the same job to have different bandwidth requirements, processing times and profits.

In {\ufpround}, we can model the number of \emph{copies} of the time-varying resource needed to satisfy all requests. This can also model routing in optical networks, where each copy of the resource corresponds to a distinct frequency. As the number of distinct available frequencies is limited, minimizing the number of rounds for a given set of requests is a natural objective.

One important assumption we make is the {\nba} (NBA), which states that the maximum demand requirement of any request is at most the minimum edge capacity, i.e., $\max_i d_i \leq \min_e c_e$. Note that this assumption is stronger than the \emph{feasibility} requirement, which says that the demand of any request is at most the minimum edge capacity on its source-sink path. This is a standard assumption in these settings. From a practical perspective, it should hold. From an algorithmic perspective, it is needed to ensure that the integrality gap of the linear programming relaxation is small.

\subsection {Related work}
{\edpmax} is NP-hard, even for restricted classes of graphs like planar graphs. However, {\edpmax} can be solved optimally in polynomial time for some classes of graphs. When the graph is a path, it translates to finding the maximum number of pairwise disjoint intervals. This is equivalent to finding a maximum independent set in an interval graph, which can be done in linear time \cite{GuptaLL82}. For trees, a polynomial time algorithm was given in \cite{GargVY97}. In undirected rings, {\edpmax} can also be solved optimally in polynomial time \cite{WanL98}. For undirected graphs, {\edpmax} is known to be APX-hard \cite{Erlebach01} and there is no algorithm with approximation factor $\Omega((\log n)^{\frac{1}{2} - \epsilon})$ for any $\epsilon > 0$, unless $\mathbf{NP \subseteq ZPTIME}(n^{\mathrm{polylog}(n)})$ \cite{AndrewsCZ05}. Here, $\mathbf{ZPTIME}(n^{\mathrm{polylog}(n)})$ is the set of languages that have randomized algorithms that always give the correct answer and have expected running time $n^{\mathrm{polylog}(n)} \equiv n^{O(\log^{O(1)} n)}$. For directed graphs, there is no algorithm with approximation factor $\Omega(m^{\frac{1}{2} - \epsilon})$ for any $\epsilon > 0$, unless $\mathbf{P = NP}$ \cite{GuruswamiKRSY03}. For directed graphs, {\edpmax} has a bounded-length greedy (BGA) $O(\sqrt{m})$-approximation algorithm \cite{Kleinberg96}.

\ifx \synopsis \undefined
{\edpmax} has also been studied for special type of graphs. For bounded-degree expander graphs, Kleinberg and Rubinfeld \cite{KleinbergR96} showed that the bounded-length greedy algorithm gives an $O(\log n \log \log n)$-approximation. Kolman and Scheideler \cite{KolmanS04} gave an improved $O(\log n)$-approximation by using the fact that routing number for expanders is $O(\log n)$. For two-dimensional meshes, Kleinberg and Tardos \cite{KleinbergT95} gave a randomized polynomial-time algorithm that achieves a constant-factor approximation with high probability. For hypercubes, Kolman and Scheideler \cite{KolmanS06} gave an $O(\log n)$-approximation by exploiting the fact that hypercubes have flow number $O(\log n)$.
\fi

{\ufpround} is NP-Hard, since it contains the \textsc{Bin Packing} problem as a special case, where the graph is just a single edge. \textsc{Bin Packing} is known to be APX-hard, so a \emph{polynomial time approximation scheme} (PTAS) is not possible. However, it has an \emph{asymptotic  polynomial time approximation scheme} (APTAS). There are also simple greedy algorithms like \emph{first-fit} and \emph{best-fit}, which give constant-factor approximations \cite{Borodin98,Hochbaum97, Vazirani01, WilliamsonShmoys11}. When all capacities and demands are 1, {\ufpround} reduces to the interval coloring problem on paths, for which a simple greedy algorithm gives the optimal coloring in linear time.

The {\ufpround} problem for paths has been well-studied in the context of online algorithms. Here the demands (intervals) arrive in arbitrary order, and we need to assign them a color on their arrival so that all intervals with one color form a feasible packing, \emph{i.e.,} total demand on any edge does not exceed its capacity. In this context, it is also called the \emph{interval coloring} problem. When all capacities and demands are 1, \emph{i.e.}, when no two intersecting intervals can be given the same color, the first-fit algorithm achieves a constant competitive ratio. Kierstead \cite{Kierstead88} first proved that first-fit requires at most $40 \omega$ colors to color an interval graph with clique size $\omega$. Later Kierstead and Qin \cite{KiersteadQ95} improved it to $26 \omega$. Subsequently, Pemmaraju et al. \cite{PemmarajuRV11} improved it to $8 \omega$, which is currently the best known upper bound. Chrobak and Slusarek \cite{ChrobakS88} showed that first-fit uses at least $4.4 \omega$ colors in the worst case. Kierstead and Trotter \cite{KiersteadT81} gave a different online algorithm which uses at most $3 \omega - 2$ colors. They also proved that any deterministic online algorithm in the worst case will require at least $3 \omega - 2$ colors, so this algorithm is the best possible one can hope for.

Adamy and Erlebach \cite{AdamyE03} introduced the \emph{interval coloring with bandwidth} problem. In this problem, all edge capacities are 1 and each interval has a demand in $(0, 1]$. They gave a 195-competitive algorithm for this problem. Later, the competitive ratio was improved to 10 by Narayanaswamy \cite{Narayanaswamy04} and Azar et al. \cite{AzarFLN06}. Epstein et al.~\cite{EpsteinEL09} further generalized the problem by allowing arbitrary edge capacities and arbitrary demands. They gave a 78-competitive algorithm for this problem satisfying the no-bottleneck assumption (NBA). Without NBA, they gave a $O\left(\log \left(\frac{d_{\max}}{c_{\min}}\right)\right)$-competitive algorithm. They also showed that without this assumption, there is no deterministic online algorithm for interval coloring with nonuniform capacities and demands, that can achieve a competitive ratio better than $\Omega(\log \log n)$ or $\Omega\left(\log \log \log \left(\frac{c_{\max}}{c_{\min}}\right)\right)$. Here, $c_{\max}$ and $c_{\min}$ are the maximum and minimum edge capacities of the path respectively.

{\ufpround} has been studied on trees and meshes ($n \times n$ two-dimensional grids) for the special case when all capacities and demands are 1. Bartal and Leonardi \cite{BartalL99} gave an online algorithm for trees with competitive ratio $O(\log n)$. The also showed that any online algorithm for trees cannot have competitive ratio better than $\Omega\left(\frac{\log n}{\log \log n}\right)$. For meshes, they gave matching upper and lower bounds of $O(\log n)$.

{\ufpmax} and {\ufpbag} are \emph{weakly} NP-Hard, since they contain the \textsc{Knapsack} problem as a special case, where the graph is just a single edge. For \textsc{Knapsack}, an FPTAS is known, and it has a simple greedy 2-approximation algorithm \cite{Vazirani01, WilliamsonShmoys11}. When all capacities, demands and profits are 1, {\ufpmax} specializes to {\edpmax}. Recently, it has been proved that the problem is \emph{strongly} NP-hard, even for the restricted case where all demands are chosen from $\{1,2,3\}$ and all capacities are uniform \cite{BonsmaSW11}. However, the problem is not known to be APX-hard, so a \emph{polynomial time approximation scheme} (PTAS) may still be possible.

With NBA, Chakrabarti et al.~\cite{ChakrabartiCGK07} gave the first
constant factor approximation algorithm for {\ufpmax} on the path and the approximation ratio was subsequently improved to $(2+\eps)$ for any constant $\eps > 0$ by Chekuri et al.~\cite{ChekuriMS07}. They also gave
a constant factor approximation algorithm for {\ufpmax} on trees. These algorithms are based on the idea
of rounding a natural LP relaxation of the {\ufpmax} problem. Without NBA, Bonsma et al. \cite{BonsmaSW11} gave a polynomial time $(7 + \epsilon)$-approximation algorithm for any $\epsilon > 0$, and a $25.12$-approximation algorithm with running time $O(n^4 \log n)$. Their algorithm divides the demands into three classes: small, medium and large. For small and medium demands, they use LP rounding to get a $(3 + \epsilon)$-approximation algorithm. For large demands, they model this as a maximum weight independent set problem for a set of rectangles. Using a dynamic programming based algorithm, they give a $4$-approximation
algorithm for large demands.

{\ufpmax} has also been studied for other graph classes. We mention some of the results for cycles and trees. Under NBA, it can be shown that {\ufpmax} on a cycle can be reduced to two instances of {\ufpmax} on a path, by splitting the cycle at a carefully selected edge. From this, if we have a $\rho$-approximation for {\ufpmax} on a path, we can get a $(\rho+1)$-approximation for {\ufpmax} on a cycle \cite{ChakrabartiCGK07}. Hence, by using the $(2+\epsilon)$-approximation for {\ufpmax} on a path given by \cite{ChekuriMS07}, we can immediately get a $(3+\epsilon)$-approximation for {\ufpmax} on a cycle. By directly modeling the problem as an LP, one can get an improved $(2 + \epsilon)$-approximation \cite{ChekuriMS07}. For trees, under NBA, there is a $4$-approximation for unit demands and a $48$-approximation for arbitrary demands, the profits being arbitrary in both the cases \cite{ChekuriMS07}. Without NBA, there is an $O(\log n)$-approximation for unit profits and an $O(\log^2 n)$-approximation for arbitrary profits, the demands being arbitrary in both cases \cite{ChekuriEK09}.

\ifx \synopsis \undefined
For general graphs, Kolman and Scheideler \cite{KolmanS06} gave an $O\left(\inv{\alpha} \Delta \frac{c_{\max}}{c_{\min}} \log n\right)$-approximation for {\ufpmax} with NBA, when profit of a request is equal to its demand. Without NBA, under the same assumption they gave an $O(\sqrt{m})$-approximation algorithm. Azar and Regev \cite{AzarR06} gave a combinatorial $O(\sqrt{m})$-approximation algorithm with NBA. Chakrabarti et al. \cite{ChakrabartiCGK07} gave an LP-based $O\left(\inv{\alpha} \Delta \log n\right)$-approximation with uniform edge capacities and an $O\left(\inv{\alpha} \Delta \log^2 n\right)$-approximation with arbitrary edge capacities. Azar and Regev \cite{AzarR06} showed that for directed graphs, there is no algorithm for {\ufpmax} with approximation factor $\Omega(m^{1 - \epsilon})$ for any $\epsilon > 0$, unless $\mathbf{P = NP}$.
\fi

The {\ufpbag} problem was introduced by Chakaravarthy et al. \cite{ChakaravarthyPSS10}, who gave an $O\left(\log \left(\frac{c_{\max}}{c_{\min}}\right)\right)$-approximation algorithm. Chakaravarthy et al.~\cite{ChakaravarthyCS10} gave the first constant factor approximation algorithm for the {\ufpbag} problem on paths -- the approximation ratio is 120. A related problem is the job interval selection problem for which Chuzhoy et al. \cite{ChuzhoyOR06} gave an $\left(\frac{e}{e-1}\right)$-approximation algorithm. See also Erlebach et al. for some additional results \cite{ErlebachS03}.

The round version of {\ufpbag} is hard to approximate, because scheduling jobs with interval constraints is a special case of this. Recall that here, we have a collection of $n$ jobs where each job is associated with a set of intervals on which it can be scheduled. The goal is to minimize the total number of machines needed to schedule all jobs subject to these interval constraints. In the continuous
version, the intervals associated with a job form a continuous time segment, described by a release date and a deadline. Chuzhoy et al. \cite{ChuzhoyGKN04} gave an $O\left(\sqrt{\frac{\log n}{\log \log n}}\right)$-approximation algorithm for this version. This was subsequently improved by Chuzhoy and Codenotti \cite{ChuzhoyC09} to an $O(1)$-approximation algorithm. They also showed that the linear programming formulation for the problem has an integrality gap of $\Theta\left(\frac{\log n}{\log \log n}\right)$. In the discrete version, where the set of allowed intervals for a job is given explicitly, Raghavan and Thompson \cite{RaghavanT87} gave an $O\left(\frac{\log n}{\log \log n}\right)$-approximation algorithm using randomized rounding. Chuzhoy et al. \cite{ChuzhoyN06} proved that it is $\Omega(\log \log n)$-hard to approximate the discrete version.

\subsection {Our contributions}
There has been lot of recent work on obtaining constant factor approximation algorithms
for these NP-Hard problems. Obtaining constant factor approximation algorithms for these problems without NBA remains a challenging task; the only exception being the recent result of  Bonsma et al.~\cite{BonsmaSW11} which gives a constant factor approximation algorithm for {\ufpmax} on the line.  We will assume that NBA holds in subsequent discussions.

\subsubsection{Linear Programming formulation for Max-UFP}
A natural linear programming formulation for {\ufpmax} on a path is given below. Here $x_i$ denotes the fraction of the demand $i$ that is satisfied and $I_i$ is the unique path between $s_i$ and $t_i$.
\begin{align*}
\textbf{maximize} \quad \; \sum_{i=1}^k w_i x_i \! & \qquad \qquad \textbf{(UFP-LP)} \\
\textbf{such that} \quad \sum_{i:e \in I_i} d_i x_i &\le c_e \qquad \! \forall e \in E \\
0 \le x_i &\le 1 \qquad \; \forall i \in \{1,\ldots,k\}
\end{align*}
If we replace the constraints $x_i \in [0,1]$ by the constraints $x_i \in \{0,1\}$ we get an integer program, which precisely models {\ufpmax}.

\subsubsection{Convex decomposition of a fractional LP solution}
Suppose $x$ is a feasible fractional solution for a maximization LP and $z_1,\ldots,z_k$ are feasible integral solutions for the LP, such that $x = \sum_{i=1}^k \lambda_i z_i$ and $\sum_{i=1}^k \lambda_i = \alpha$. Then the value of the best solution, say $z_{\max}$ among $z_1,\ldots,z_k$ is at least $\frac{1}{\alpha}$ fraction of the value of $x$. We can think of this as approximate convex decomposition of a fractional solution.

\subsubsection {Our results}
Starting with a simple algorithm for {\ufpround} on paths,
 we give a unified framework for these problems. We round natural LP relaxations for {\ufpmax} and {\ufpbag}. The
rounding algorithm essentially shows that one can express a fractional solution to the LP as an approximate convex combination
of integer solutions. We show how to do this using our algorithm for
{\ufpround}. This leads to improved approximation algorithms for several of these problems. More specifically, our results are:
\begin{itemize}
\item We give a $24$-approximation algorithm for the {\ufpround} problem on paths. This is much simpler than the $78$-competitive algorithm of \cite{EpsteinEL09}, and gives an improved approximation ratio.
\item We give a $17$-approximation algorithm for the {\ufpmax} problem on paths. Although a $(2+\epsilon)$-approximation is known for this problem, our approach using convex decompositions may be of independent interest.
\item We give a $65$-approximation algorithm for the {\ufpbag} problem on paths, thus improving the constant approximation factor of $120$ given by Chakaravarthy et al.~\cite{ChakaravarthyCS10}.
\item For trees, we give the first constant factor approximation algorithm for the {\ufpround} problem -- our approximation factor is $64$.
\end{itemize}
These results have appeared in \cite{ElbassioniGGKNP12}.\\

\noindent For the online version of the {\ufpround} problem on paths, we have the following result.
\begin{itemize}
\item We give a $58$-competitive algorithm for the online version of the {\ufpround} problem on paths. This is simpler than the $78$-competitive algorithm of \cite{EpsteinEL09}, and gives a better competitive ratio.
\end{itemize}
This result appears in \cite{KumarPSS12}.

\section {Resource allocation for scheduling jobs}
We consider the problem of allocating resources to schedule jobs.
As before, we are give a path $G$, and a set of jobs.
Each job $j$ is specified by a triplet $(s_j, t_j, d_j)$, where
$[s_j, t_j]$ denotes the interval corresponding to the job (also denoted by $I_j$),
and $d_j$ is its demand requirement. We shall assume that $d_j$ values are 1.
Further, we are also given a set of resources.
Each resource is specified by its starting and ending vertex, and  the capacity it offers and its associated cost.
A feasible solution is a set of resources satisfying the constraint that for any edge,
the sum of the capacities offered by the resources containing this edge is at least the demand required by
the jobs containing that edge, i.e., the selected resources must cover the jobs.
We call this the Resource Allocation problem ({\ResAll}).

The above problem is motivated by applications in cloud and grid computing.
Consider jobs that require a common resource such as network bandwidth or storage.
The resource may be available under different plans; for instance, it is common for network bandwidth to be priced
based on the time of the day to account for the network usage patterns during the day.
The plans may offer different capacities of the resource at different costs.
Moreover, it may be possible to lease multiple units of the resource under some plan by paying a cost
proportional to the number of units.

Bar-Noy et al. \cite{Bar-Noy} presented a $4$-approximation algorithm for the {\ResAll} problem.
We consider two variants of this problem.
The first variant is the partial covering version. In this problem, the input also specifies a number $k$ and
a feasible solution is only required to cover $k$ of the jobs.
The second variant is the prize collecting version wherein each job has a penalty associated with it;
for every job that is not covered by the solution, the solution incurs an additional cost,
equivalent to the penalty corresponding to the job.
These variants are motivated by the concept of service level agreements (SLA),
which stipulate that a large fraction of the client's jobs are to be completed.
We study these variants for the case where the demands of all the jobs are uniform (say $1$ unit)
and a solution is allowed to pick multiple copies of a resource by paying proportional cost.
We now define our problems formally.

\subsection{Problem definition}
We consider the graph $G=(V,E)$ which is a path with vertices numbered $1, 2, \ldots, |V|$ from left to right.
An input instance  consists of a set of  {\em jobs} ${\cal J}$, and a set of {\em resources} $\calI$. The number of jobs is $n$ and the number of resources is $m$.

Each job $j \in {\cal J}$ is specified by an interval $I_j = [s_j,t_j]$ in the path. Recall that each job has
demand requirement of 1.
Each resource $i \in \calI$ is specified
by an interval $I_i=[s(i),e(i)]$ in the path, capacity $w_i$ and cost $c_i$. We shall assume that the capacities $w_i$ are
integers.
We interchangeably refer to the resources as {\em resource intervals}. We shall also refer to the interval $I_j$ (or $I_i$) as
the {\em span} of the job $j$ (or resource $i$).
A typical scenario of such a collection of jobs and resources is shown in \autoref{fig:cc}.
\begin{figure*}[t]
\begin{center}
\fbox{
\includegraphics[scale=0.5]{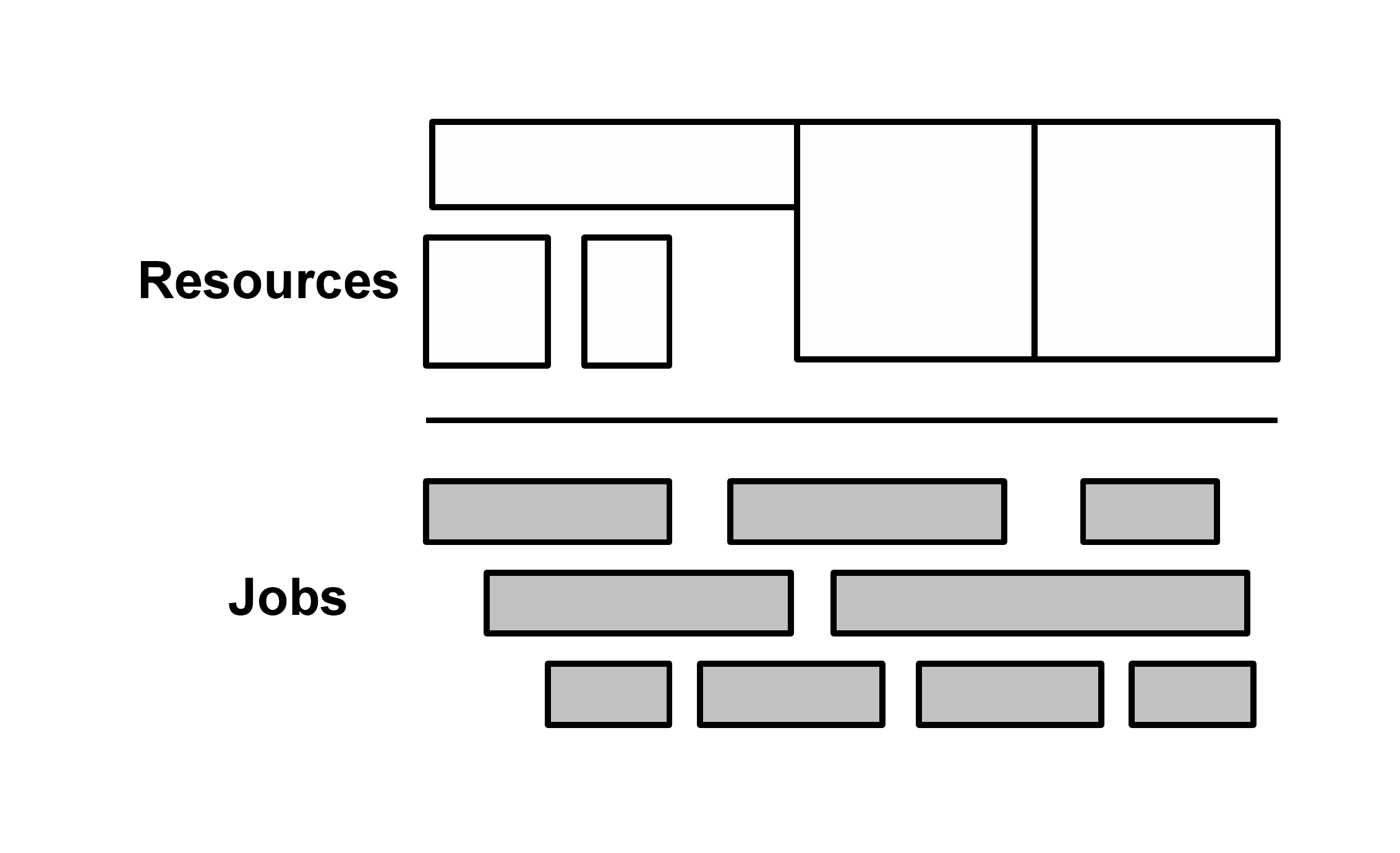}
}
\end{center}
\caption{Illustration of the input}
\label{fig:cc}
\end{figure*}
We say that a job $j$ (or resource $i$) {\em contains} an edge $e$ if the associated interval $I_j$ (or $I_i$) contains $e$; we  denote this as $j \sim e$ ($i \sim e$).
We define a {\em profile} $P: E \rightarrow \mathbb{N}$ to be a mapping that assigns an integer value
to every edge of the path. For two profiles, $P_1$ and $P_2$, $P_1$ is said to {\em cover} $P_2$,
if $P_1(e) \geq P_2(e)$ for all $e \in E$.
Given a set $J$ of jobs, the profile $P_J(\cdot)$ of $J$ is defined to be the mapping determined by
the cumulative demand of the jobs in $J$, i.e. $P_J(e) = |\{ j \in J:j \sim e \}|$.
Similarly, given a multiset $R$ of resources, its profile is: $P_R(e) = \sum_{i \in R:i \sim e} w_i$
(taking copies of a resource into account).
We say that $R$ {\em covers} $J$ if $P_R$ covers $P_J$.
The cost of a multiset of resources $R$ is defined to be the sum of the costs of all the resources
(taking copies into account).

We now formally define  the problems.
\begin{itemize}
\item  {\ResAll}: In this problem, a feasible solution is a multiset of resources $R$, which covers the set of \emph{all} jobs ${\cal J}$. The cost of the solution is the sum of the costs of the resources in $R$ (taking copies into account).	The problem is to find a feasible solution of minimum cost.
\item  {\ZeroOneResAll}: This is similar to the {\ResAll} problem, except that a resource can be used at most once to cover any job.
\item  {\PResAll}: In this problem, the input also specifies a number $k$ (called the {\em partiality parameter})
       that indicates the
	number of jobs to be covered. A feasible solution is a pair $(R,J)$ where $R$ is a multiset of resources
	and $J$ is a set of jobs such that $R$ covers $J$ and $|J| \ge k$. The cost of the solution is the sum of the costs of the resources in $R$ (taking copies into account).
	The problem is to find a feasible solution of minimum cost.
\item  {\PCResAll}: In this problem, every job $j$ also has a penalty $p_j$ associated with it.
	A feasible solution is a pair $(R,J)$ where $R$ is a multiset of resources
	and $J$ is a set of jobs such that $R$ covers $J$.
	The cost of the solution is the sum of the
	costs of the resources in $R$ (taking copies into account) and the penalties of the jobs not in $J$.
	The problem is to find a feasible solution of minimum cost.
\end{itemize}

\begin{figure*}[t]
\begin{center}
\fbox{
\includegraphics[scale=0.4]{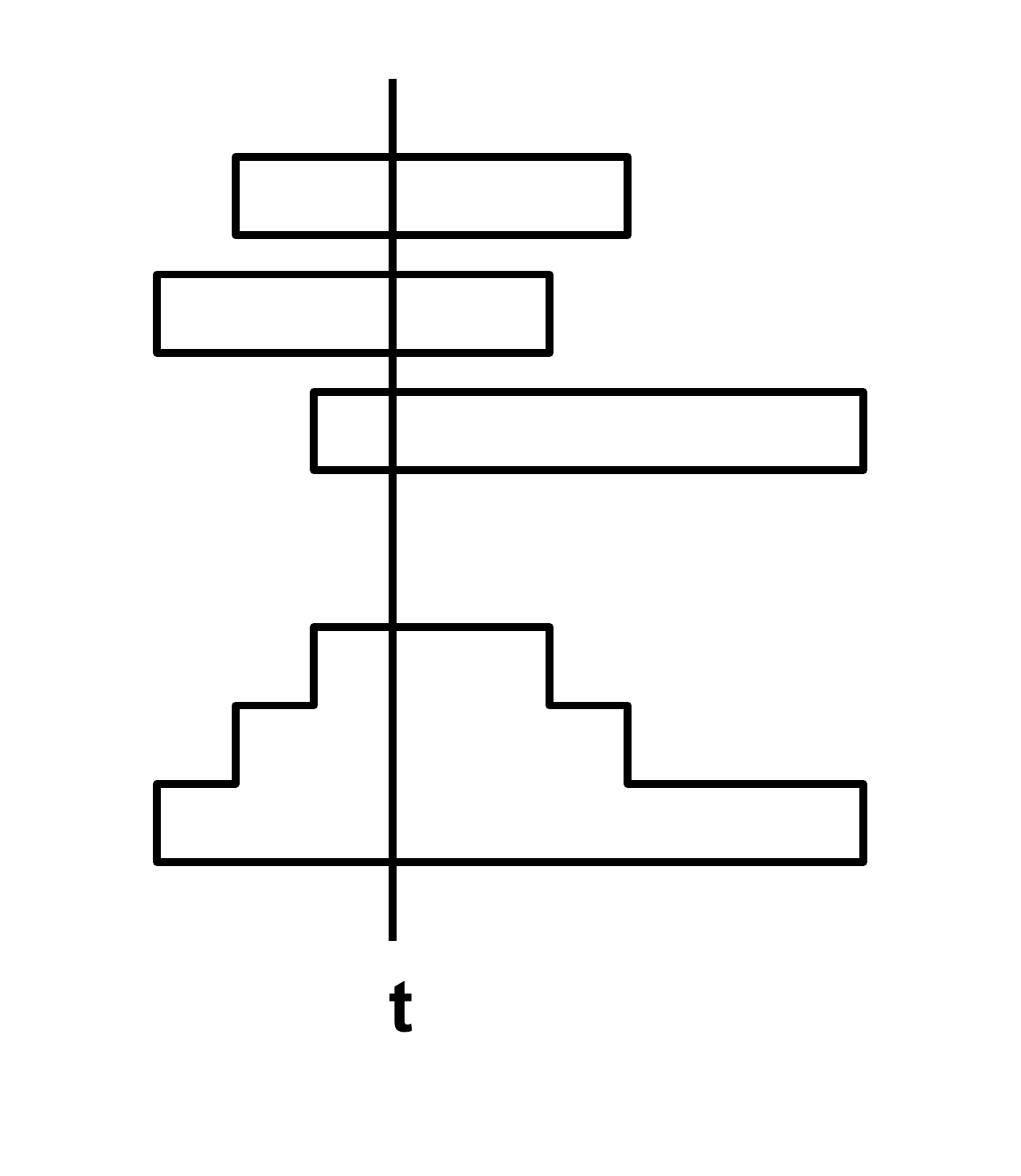}
}
\end{center}
\caption{A Mountain $M$}
\label{fig:aaa}
\end{figure*}

\subsection{Related work}
Our work belongs to the class of {\em partial} covering problems, which are a
natural variant of the corresponding full cover problems. There is a significant body of work
that consider such problems in the literature, for instance, see \cite{Garg05,Bar01,JV01,KPS11,GKS04}.

In the setting where resources and jobs are embodied as intervals, the objective of finding a minimum cost collection of
resources that fulfill the jobs is typically called the {\em full cover} problem. Full cover problems in this context have been dealt with in various earlier works \cite{Bar-Noy,bhatia07,cgk10}. Partial cover problems in the interval context have been considered earlier in \cite{esa2011}.

The work in existing literature that is closest in spirit to our result is that of
Bar-Noy et al. \cite{Bar-Noy}, and Chakaravarthy et al. \cite{esa2011}.
In \cite{Bar-Noy}, the authors consider the full cover version, and present
a $4$-approximation algorithm. In this case, all the jobs have to be covered,
and therefore the demand profile to be covered is fixed. The goal is to find the
minimum cost set of resources, for covering this profile.
In our setting,
we need to cover only $k$ of the jobs.
A solution needs to select $k$ jobs to be covered in such a manner
that  the resources required to cover the resulting demand profile has minimum cost.

In \cite{esa2011}, the authors consider a scenario, wherein the
edges have demands and a solution must satisfy the demand for at least $k$ of the edges(\textsc{PartialMultiResAll}). They give a 16-approximation algorithm for the \textsc{PartialMultiResAll} problem. They also give a 4-approximation algorithm for the {\ZeroOneResAll} problem, where each resource can be used at most once. This is a generalization of the {\ResAll} problem, where each resource can be used any number of times. In contrast, in our setting, a solution needs to satisfy $k$ {\em jobs}, wherein
each job can span multiple edges.
A job may not be completely spanned by any resource, and thus may require
{\em multiple} resource intervals for covering it.

Jain and Vazirani \cite{JV01} provide a general framework for achieving approximation algorithms for partial
covering problems, wherein the prize collecting version is considered. In this framework, under suitable conditions,
a constant factor approximation for the prize collecting version implies a constant factor approximation
for the partial version as well. However, their result applies only when the prize collecting algorithm has a
 certain strong property, called the {\em Lagrangian Multiplier Preserving} (LMP) property.
While we are able to achieve a constant factor approximation for the {\PCResAll} problem,
our algorithm does not have the LMP property. Thus, the Jain-Vazirani framework does not apply
to our scenario.

\subsection {Our contributions}
A collection of jobs $M$ is called a {\em mountain}, if there exists a edge $e$ such that
all the jobs in this collection contain the edge $e$;  (see \autoref{fig:aaa};
jobs are shown on the top and the profile is shown below).
The justification for this linguistic convention is that if we look at the profile of such a
collection of jobs, the profile forms a bimodal sequence, increasing in height until the peak, and then decreasing.
The {\em span} of a mountain is the set of edges which are contained in one of the jobs in the mountain.  
 A collection of jobs $\calM$ is called a {\em mountain range}, if the jobs can be partitioned into
a sequence $M_1, M_2, \ldots, M_r$ such that each $M_i$ is a mountain and the spans of any two mountains
are non-overlapping (see \autoref{fig:bbb}).

We show that the input set of jobs can be
partitioned into a logarithmic number of mountain ranges. Then we give a constant factor approximation algorithm for the special case of the {\PResAll} problem, where the input set of jobs form a single mountain range $\calM$. Using these two results along with dynamic programming, we get an approximation algorithm for the {\PResAll} problem.

We give a approximation factor preserving reduction from the {\PCResAll} problem to a certain full-cover problem and then use the approximation algorithm for that problem to derive an approximation algorithm for the {\PCResAll} problem.

\begin{figure*}[t]
\begin{center}
\fbox{
\includegraphics[scale=0.4]{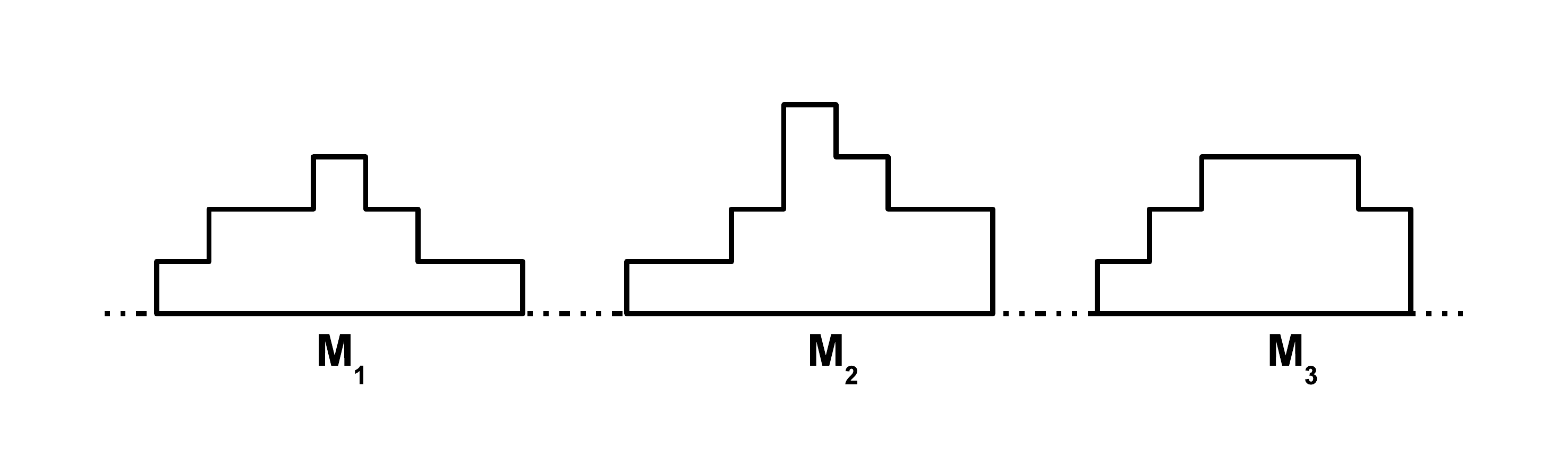}
}
\end{center}
\caption{A Mountain Range ${\cal M}=\{M_1, M_2, M_3\}$}
\label{fig:bbb}
\end{figure*}

\subsubsection {Our results}
\begin{itemize}
\item We present an $O(\log (n+m))$-approximation algorithm for the {\PResAll} problem, where
$n$ is the number of jobs and $m$ is the number of resources respectively.
\item We give a $4$-approximation algorithm for the {\PCResAll} problem, by reducing it to the {\ZeroOneResAll} problem.
\end{itemize}
These results have appeared in \cite{ChakaravarthyPRS12}.

\section {Organization of the thesis}
In Chapter 2, we study the {\ufpround} problem and give constant factor approximation algorithms for this problem on paths and trees. Building on this, we give constant factor approximation algorithms for the {\ufpmax} and the {\ufpbag} problems in Chapter 3. In Chapter 4, we study the online version of the {\ufpround} problem, also known as the {\oic} problem, and give an improved constant factor competitive algorithm. We discuss the {\PResAll} and the {\PCResAll} problems in Chapter 5 and give $O(\log (n+m))$-approximation and $4$-approximation algorithms for these two problems respectively. We conclude the thesis in Chapter 6 and discuss possible future directions on these problems along with some open problems.

\chapter {The Round-UFP Problem}
\label {chap2}

We define the {\ufpround} problem and give a constant factor approximation algorithm under NBA. We also give improved algorithms for some special cases of this problem.

\section{Preliminaries}
\label{sec:pre}
We are given a graph $G=(V,E)$, which is either a path or a tree, with edge capacities 
$c_e$ for all edges $e \in E$. We are also given a set of requests $R_1, \ldots, R_k$. Request $R_i$ has an associated source-sink pair $(s_i, t_i)$ and a demand $d_i$. We shall use $I_i$ to denote the associated unique path between $s_i$ and $t_i$ in $G$. A subset of demands will be called {\em feasible} if they can be routed without violating the edge capacities. The goal is to partition the set of demands into minimum number of colors, such that demands with a particular color are feasible.

\begin{definition}
The \emph{load} on an edge $e$, $l_e = \sum_{i : e \in I_i} d_i$, i.e., the total demand passing through the edge $e$.
\end{definition}

\begin{definition}
The \emph{congestion} of an edge $e$, $r_e = \left\lceil {\frac{l_e}{c_e}} \right\rceil,$ i.e., the ratio of the load on the edge $e$ to its capacity. Let $r = \max_{e \in E} r_e$ be the \emph{maximum congestion} on any edge in the input graph.
\end{definition}

\section{Approximation Algorithms for Round-UFP on Paths}
\subsection{A 3-approximation algorithm for uniform capacities and arbitrary demands}
We consider the special case, where each edge of the path has a capacity $c$. We separate the demands into large and small demands. A demand $d_i$ is called \emph{large} if $d_i > \frac{1}{2}c$. Otherwise, it is called \emph{small}. Let $\opt(L)$ and $\opt(S)$ be the optimum number of colors required for the instance containing only large demands and only small demands respectively. The algorithms for large and small demands are given below.

\subsubsection{An optimal algorithm for large demands}
We maintain several copies of the path, one copy for each color. We fill demands in the copies in an iterative manner. We sort the demands based on their left endpoints. Let $\RR_i$ be the set of requests starting at $v_i, 1 \le i \le n-1$. We will pack the requests in $\RR_1,\ldots,\RR_{n-1}$ in this order. Starting with the requests in  $\RR_1$, we try to allocate the requests in $\RR_i, 1 \le i \le n-1$ in one of the copies of the path, if it does not violate any edge capacities. Otherwise, we allocate a new copy and assign it there.

\begin{lemma}
If $\chi$ is the number of colors required to pack all the large demands, then $\opt(L) \ge \chi$.
\end{lemma}
\begin{proof}
Note that if two large demands share any edge, they can't be given the same color, because the total load on the edge is more than $c$. Consider the demand $d$ for which the last color $\chi$ was opened. Since $d$ could not be assigned any one of the first $\chi-1$ colors, there are $\chi-1$ large demands, one for each color, which shared an edge with $d$. Since the demands have been considered in a left to right manner, all these $\chi-1$ large demands will pass through the first edge $e$ of $d$. Together with $d$, there are $\chi$ large demands passing through the edge $e$. Hence, the optimum has to give each of them a separate color, so it will also require at least $\chi$ colors. Hence, this algorithm uses the minimum number of colors.
\end{proof}

\begin{figure}[ht]
\centering
\includegraphics[scale=0.7]{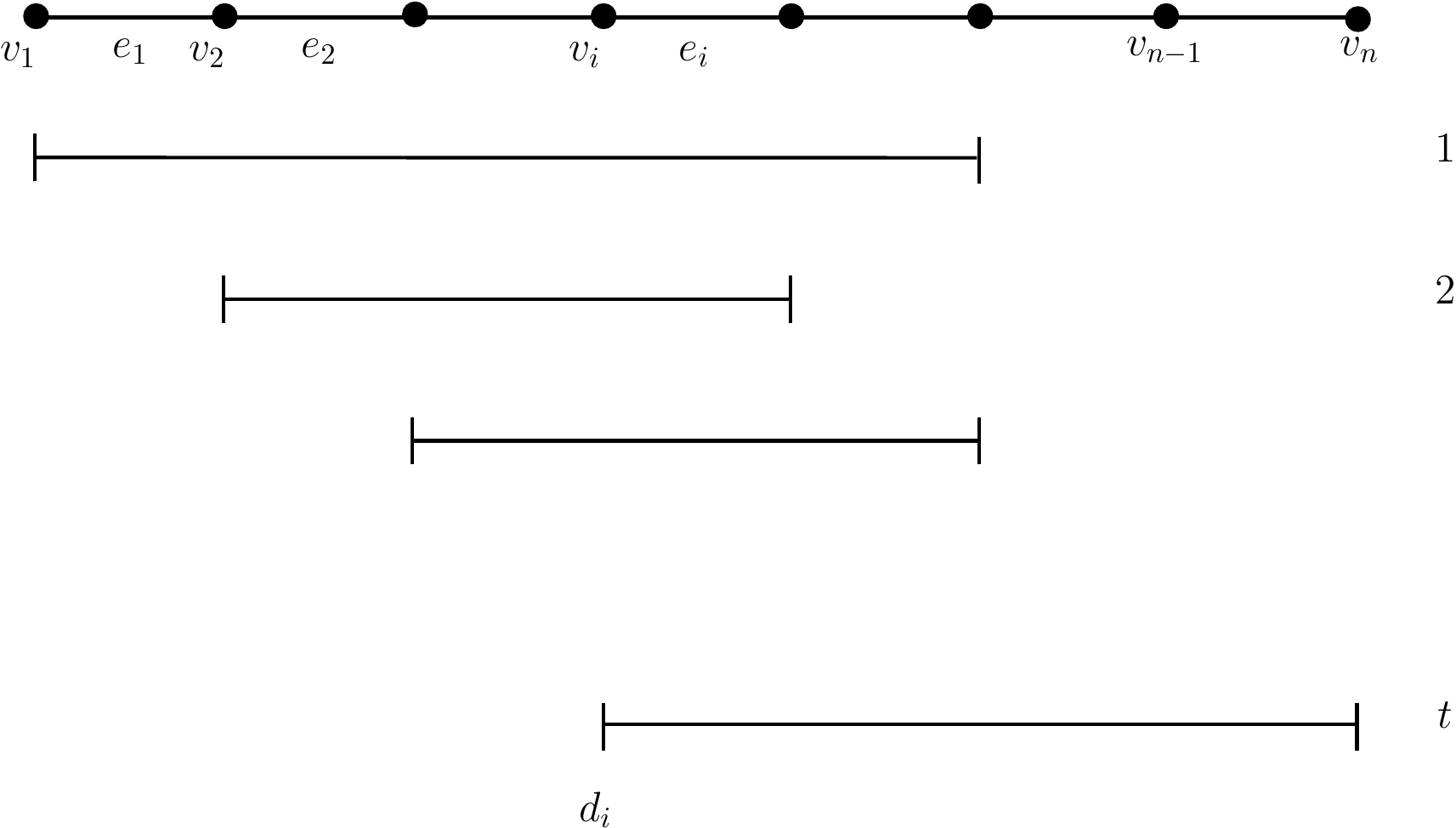}
\caption{Illustration for the analysis of small demands.}
\label{uniform-capacities}
\end{figure}

\subsubsection{A 2-approximation algorithm for small demands}
The algorithm for small demands is exactly the same as the previous algorithm for large demands.

Let $t$ be the number of copies of the path $P$ required to assign all the requests in $\RR_1,\ldots,\RR_{n-1}$. Let $l_i$ be the load on edge $e_i$, which is the sum of all demands passing through $e_i$.
\begin{lemma}
\label{lem:sml}
When all the requests in $\RR_1,\ldots,\RR_{n-1}$ have been packed, there is an edge $e_i$ such that  in at least $t-1$ copies of $P$, $l_i > \frac{1}{2}c$.
\end{lemma}

\begin{proof}
Consider the demand $d_i \in \RR_i$ (for some $i$) due to which the last color $t$ was opened. At the time $d_i$ was considered, all the requests started on or before $v_i$. Since $d_i$ could not be assigned any of the previous $t-1$ colors, there are $t-1$ edges, one for each color, such that the total load put by the existing small demands on each of these edges is strictly more than $c - d_i \ge \frac{1}{2}c$ since, $d_i \le \frac{1}{2}c$. Since the demands have been considered in a left to right manner, the load on the first edge $e_i$ of $d_i$ on each of these $t-1$ colors is at least as much. Hence, $e_i$ is the edge such that $l_i > \frac{1}{2}c$.
\end{proof}

It follows from \autoref{lem:sml} that the total load put by requests in $\RR_1,\ldots,\RR_{n-1}$ on $e_i$ is greater than $\frac{1}{2}c(t-1)$. Hence, the congestion on edge $e_i$ is more than $\frac{1}{2}(t-1)$, since the edge capacity is $c$. Thus, $r \ge r_{e_i} > \frac{1}{2}(t-1)$. Hence, $t < 2r + 1$, which implies that $t \le 2r$, since $t$ is an integer. Since, we can assign all the requests in $\RR_1,\ldots,\RR_{n-1}$ using $t \le 2r \le 2 \cdot \textsc{opt(S)}$ copies, this is a 2-approximation algorithm.

\subsubsection{A 3-approximation algorithm}
We solve the instance containing only large demands and the instance containing only small demands separately. We have, $\alg(L) = \opt(L)$ and $\alg(S) \le 2 \cdot \opt(S)$. Moreover, $\opt \ge \max\{\opt(L),\opt(S)\}$. Hence the total number of colors required by the algorithm is
\begin{align*}
\alg &= \alg(L) + \alg(S)\\
&\le \opt(L) + 2 \cdot \opt(S)\\
&\le 3 \cdot \opt.
\end{align*}

\subsubsection{Number of colors in terms of the congestion bound $r$}
For large demands, the total load on the edge $e$, $l_e > \frac{1}{2}c \cdot \chi$. On the other hand, $l_e \le rc$. Hence, $rc > \frac{1}{2}c \cdot \chi$ and so $\alg(L) = \chi < 2r$, which implies that $\chi \le 2r-1$, since $\chi$ is an integer. For small demands, $\alg(L) \le 2r$. Hence, $\alg \le 2r-1 + 2r = 4r-1$.

\subsubsection{Running time}
Since we are only sorting the demands based on their left endpoints and maintaining a set of copies, the running time of the algorithm is polynomial.

\subsection{A 24-approximation algorithm for arbitrary capacities and arbitrary demands}
\label{sec:round}
We consider an instance $\I$ of the {\ufpround} problem given by a path $G$ on $n$ points, 
and a set of requests $R_1, \ldots, R_m$. 
Let $\opt$ denote an optimal solution, and $\col(\opt)$ denote the
number of colors used by $\opt$. We begin with a few definitions.

\begin{definition}
The {\em bottleneck capacity} $b_i$ of a request $R_i$ is the smallest capacity of an edge in the interval between $s_i$ and $t_i$ -- such an edge is called the {\em bottleneck edge}
for request $R_i$. A demand $d_i$ is said to be {\em small} if $d_i \leq \frac{1}{4}b_i$, else it is a {\em large} demand. More generally, a demand $d_i$ is said to be {\em $\delta$-small} if $d_i \leq \delta b_i$. Otherwise, it is a {\em $\delta$-large} demand.
\end{definition}

Clearly, $\col(\opt) \ge r$. 
We give an algorithm $\cal A$ which
uses $O(r)$ colors. This will give a constant factor approximation algorithm for this problem. We first
consider the case of large demands. We will use the following result of Nomikos et al.~\cite{NomikosPZ01}.

\begin{lemma} \label{lem:nomikos}
Consider an instance of {\ufpround} where all capacities are integers and all demands $D_i$ have bandwidth requirement
$d_i = 1$. Then, one can color these demands with $r$ colors. 
\end{lemma}

\begin{lemma}
\label{lem:large}
We can color all large demands with at most $8r$ colors.
\end{lemma}
\begin{proof}
We first scale all capacities and demands such that the minimum capacity $c_{\min}$ becomes 1. Now, we round all
capacities down to the nearest integer, and we increase all the demands $d_i$ to 1. Note that 
this will affect the congestion of an edge $e$ by a factor of at most 8. Since $c_e \ge c_{\min} = 1$, rounding $c_e$ down to the nearest integer will reduce it by a factor of at most $2$ (which will happen for a real number less than but arbitrarily close to 2). Since all demands are of size at least 
$\frac{1}{4}$ (because they are large demands, so $d_i > \frac{1}{4}b_i \ge \frac{1}{4}c_{\min} = \frac{1}{4}$), we may increase the requirement of a demand by a factor of at most 4. 
Thus, the value of $r$ will increase by a factor of at most 8. Now, we invoke the result in Lemma~\ref{lem:nomikos}.
This proves the lemma. 
\end{proof}

We now consider the more non-trivial case of small demands. We divide the edges into classes based on their
capacities.
We say that an edge $e$ is of {\em class}
$l$ if $2^l \le c_e < 2^{l+1}$. We use $\cl(e)$ to denote the class of $e$.
 For a demand $D_j$, let $l_j$ be the smallest class such that the interval
$I_j$ contains an edge of class $l_j$. The {\em critical edge} of demand $D_j$ is defined as the first edge
(as we go from left to right from $s_j$ to $t_j$) in $I_j$ of class $l_j$. Note that the critical edge could be different
from the bottleneck edge, though both of them would be of class $l_j$.
\begin{lemma}
\label{lem:small}
The small demands can be colored with at most $16 r$ colors.
\end{lemma}
\begin{proof}
We maintain $16 r$ different solutions to the instance $\I$, where a solution  routes a subset of
the demands. We will be done if we can assign each demand to one of these solutions. Let us call these
solutions $\S_1, \ldots, \S_K$, where $K = 16 r$.
 We first describe
the routing algorithm and then show that it has the desired properties.

We arrange the demands
in order of their left end-points -- let this ordering be $D_1, \ldots, D_m$. Let $e_j$ be the critical edge
of $D_j$. When we consider $D_j$, we send it to a solution $\S_l$ for which the total requirements of
demands containing $e_j$ is at most $c_{e_j}/16$. At least one such solution must exist, otherwise $r_e >
\frac{16 r \cdot c_{e_j}/16 }{c_{e_j}} = r, $ a contradiction. This completes the description of
how we assign each demand to one of the solutions. We now prove that each of the solutions $\S_l$ is feasible.

Fix a solution $\S_l$ and an edge $e$.
Suppose $e$ is of class $i$.  Let $\D(\S_l)$ be the demands routed in $\S_l$ which contain the edge $e$.
Among such demands, let $D_u$ be the last demand for which the critical edge is to the left of $e$ (including
$e$) -- let
$e'$ be such an edge. Clearly, $\cl(e') \geq i$. For an integer $i'  \leq i$, let $e^{(i')}$ be the first
edge of class $i'$ to the right of $e$ (so, $e^{(i)}$ is same as $e$).

First consider the demands in $\D(\S_l)$ which are considered before (and including $D_u$). All of these demands
go through $e'$ (because all such demands begin before $D_u$ does and contain $e$). So, the total requirement
of such demands, excluding $D_u$, is at most $c_{e'}/16$ -- otherwise we would not have assigned $D_u$ to
this solution. Because $D_u$ is a small demand and $\cl(e') \geq i$, the total requirements of such demands (including $D_u$) is at most
$$\frac{2^{i+1}}{16} + \frac{c_e}{4} \le \frac{c_e}{8} + \frac{c_e}{4} = \frac{3c_e}{8}.$$

Now consider the demands in $\D(\S_l)$ whose critical edges are to the right of $e$ -- note that, such an
edge must be one of $e^{(i')}$ for some $i' <  i$. Similar to the argument above, 
the total requirements of
such demands is at most 
$$\sum_{i'=0}^{i-1} \left( \frac{2^{i'+1}}{16} + \frac{2^{i'+1}}{4} \right) = \frac{5}{16} \sum_{i'=0}^{i-1} 2^{i'+1} \leq \frac{5 \cdot 2^{i+1}}{16} = \frac{5 \cdot 2^i}{8} \leq \frac{5c_e}{8}.$$
Here, we have used the fact that $c_e \ge 2^i$. Thus, we see that the total requirements of demands in $\D(\S_l)$ is at most
$$\frac{5c_e}{8} + \frac{3c_e}{8} \leq c_e. $$
Hence the solution is feasible. This proves the lemma.
\end{proof}

Combining the above two lemmas, we get the following theorem.
\begin{theorem}
\label{thm:round}
Given an instance of {\ufpround}, there is an algorithm  for this problem which uses
at most $24 \cdot \col(\opt)$ colors, and hence it is a $24$-approximation algorithm. Further, if all
demands are small, then one can color the demands using at most $16 \cdot \col(\opt)$ colors.
\end{theorem}

\subsubsection{Running time}
For large demands, we are using the algorithm by Nomikos et al.~\cite{NomikosPZ01}, which is polynomial-time. Scaling the capacities and demands requires polynomial-time. For small demands, sorting the demands and maintaining several copies of the path can be done in polynomial-time. The critical edge of a demand can also be found in polynomial-time. Hence, the overall algorithm runs in polynomial-time.

\subsection{How bad can the congestion bound be?}
\subsubsection{With NBA}
We show an example, where even the optimal coloring requires $2r$ colors. Suppose there is a single edge of capacity 1. There are $2k$ copies of a (large) demand $\frac{1}{2} + \epsilon$, where $\epsilon \ll \frac{1}{k}$. Since, no two demands can be given the same color, the optimal coloring requires $2k$ colors, while the congestion bound $r$ is $k+1$. Hence, $\opt \approx 2r$, for large $k$.

\begin{figure}[ht]
\centering
\includegraphics[scale=0.7]{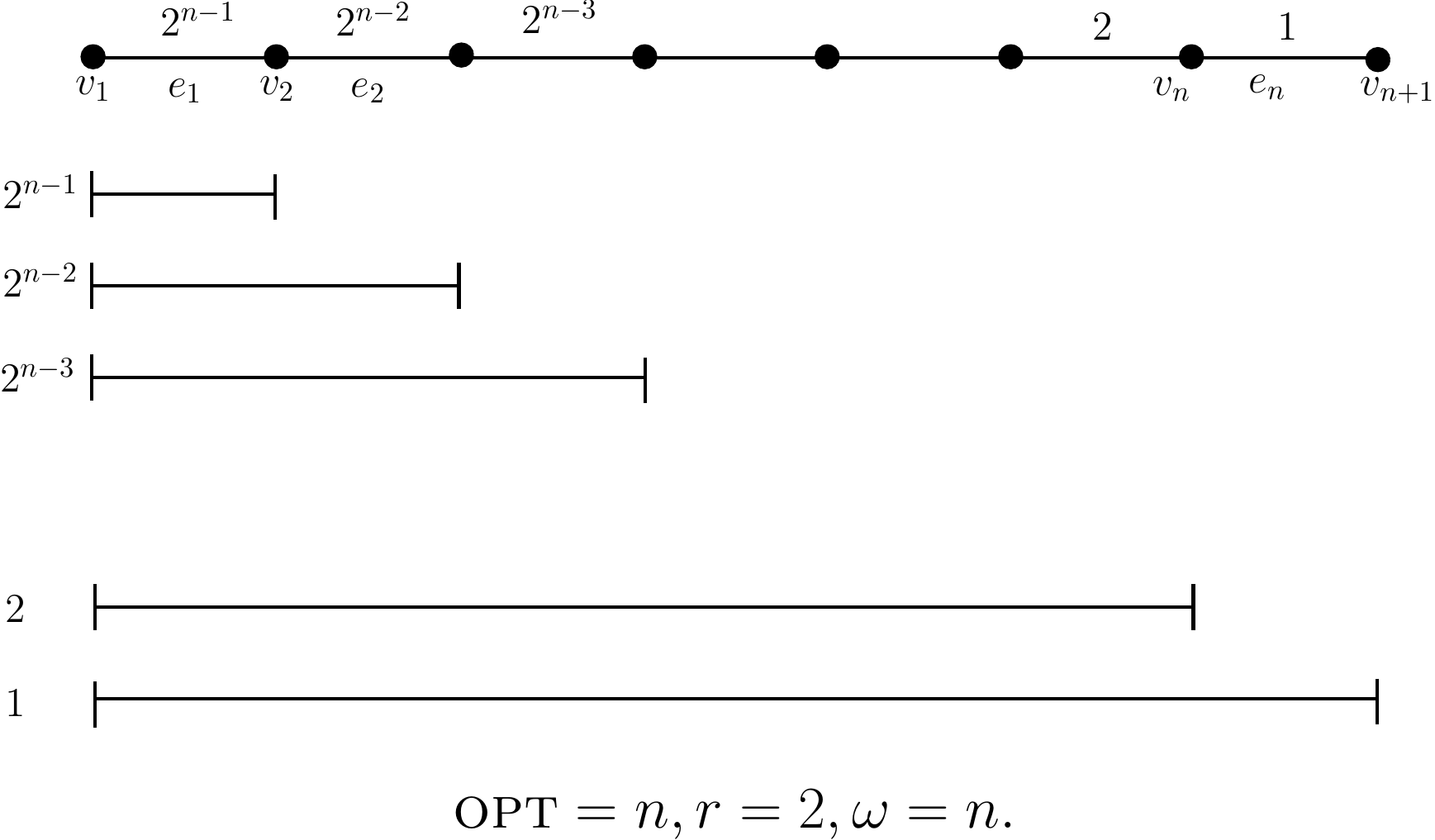}
\caption{An example where $\opt = n, r = 2, \omega = n$.}
\label{bad-congestion-bound}
\end{figure}

\subsubsection{Without NBA}
We show an example where even the optimal coloring requires $n$ colors, whereas the congestion bound is $2$. In \autoref{bad-congestion-bound}, the capacities are geometrically decreasing and $c(e_i) = 2^{n-i-1}$ for $1 \le i \le n-1$. The demands $D_j$ are between $v_1$ and $v_j$ for $2 \le j \le n$ and $d_j = 2^{n-j}$. Since, no two demands can be given the same color, the optimal coloring requires $n$ colors, while the congestion bound $r$ is $2$. So the ratio $\frac{\opt}{r}$ is $\Omega(n)$. Note that in this example, NBA is not satisfied. Here, $\omega$ is the maximum number of intervals that can't be assigned the same color, which is also the maximum clique size in the corresponding interval graph. Note that $\omega$ is a lower bound on $\opt$, so any solution requires at least $\omega$ colors.

\section{Approximation Algorithms for Round-UFP on Trees}
\label{sec:roundtree}
We now consider the {\ufpround} problem on trees. Consider an instance $\I$ of this problem as described
in Section~\ref{sec:pre}.
We consider the case of large and small demands separately. Let $\D^\lar$ be the set of large demands
and $D^\sm$ be the set of small demands.
\begin{lemma}
There is a $32$-approximation algorithm for the {\ufpround} problem on trees, when we only have demands in $\D^\lar$.
\end{lemma}
\begin{proof}
Chekuri et. al. \cite{ChekuriMS07} gave a $4$-approximation algorithm for coloring a set of demands when
all demands have requirement 1, and the capacities are integers. In fact, their algorithm uses
at most $4 r$ colors. In our case, first observe that if $D_i \in \D^\lar$, then
$d_i$ lies between $\frac{1}{4}c_{\min}$ and $c_{\min}$. We create a new instance $\I'$, where we round-up the
requirement of each demand $D_i$ to $c_{\min}$. Further, we round-down the capacity of each edge to the
nearest multiple of $c_{\min}$. We claim that
our algorithm uses at most $32 \cdot \col(\opt)$ colors, where $\col(\opt)$ denotes the number of colors
used by the optimal solution for the large demands in $\I$. Indeed, by increasing the requirements of
the large demands, and decreasing the capacities of the edges, we affect the congestion of an edge by
at most $4 \cdot 2 = 8$. Now this is a uniform demands instance, which is the same as a unit demands instance by scaling the capacities and demands. We lose a further factor of $4$ by using the $4$-approximation algorithm of  Chekuri et. al. Hence, the result follows.
\end{proof}

\begin{lemma}
There is a $32$-approximation algorithm for the {\ufpround} problem on trees, when we only have demands in $\D^\sm$.
\end{lemma}
\begin{proof}
The proof is very similar to that of Lemma~\ref{lem:small}. We maintain $16 r$ solutions.
For a demand $D_i$, let $a_i$ denote the
least common ancestor of $s_i$ and $t_i$. We consider the demands in a bottom-up order of $a_i$.
For a demand $D_i$, we define two critical edges: the $s_i$-critical edge is the critical edge
on the $a_i-s_i$ path, and the $t_i$-critical edge is the critical edge on the $a_i-t_i$-path.
We send $D_i$ to the solution in which both these critical edges have been used till $\frac{1}{16}$ of their
total capacity only. Again it is easy to check that such a solution will exist. The rest of the argument
now follows as in the proof of Lemma~\ref{lem:small}.
\end{proof}

\begin{theorem}
\label{thm:tree}
There is a $64$-approximation algorithm for the {\ufpround} problem on trees.
\end{theorem}
\begin{proof}
Follows from the two previous lemmas.
\end{proof}

\subsubsection{Running time}
For large demands, we are using the algorithm by Chekuri et. al. \cite{ChekuriMS07}, which runs in polynomial-time. Scaling the capacities and demands requires polynomial-time. For small demands, sorting the demands and maintaining several copies of the tree can be done in polynomial-time. The critical edge of a demand can also be found in polynomial-time. Hence, the overall algorithm runs in polynomial-time.

\chapter {The Max-UFP and the Bag-UFP Problems}
\label {chap3}

We define the {\ufpmax} and the {\ufpbag} problems and give constant factor approximation algorithms for both the problems under NBA. We are given a graph $G=(V,E)$, which is either a path or a tree, with edge capacities 
$c_e$ for all edges $e \in E$. We are also given a set of requests $R_1, \ldots, R_k$. Request $R_i$ has an associated source-sink pair $(s_i, t_i)$, a demand $d_i$, and a profit $w_i$. We shall use $I_i$ to denote the associated unique path between $s_i$ and $t_i$ in $G$. In order to route a request $R_i$, we send $d_i$ amount of flow from $s_i$ to $t_i$ along the (unique) path between them in $G$. A subset of demands will be called {\em feasible}, if they can be simultaneously routed without violating the edge capacities.

In the {\ufpmax} problem, we would like to find a feasible subset of demands of maximum total profit. In the {\ufpbag} problem, we are given sets, which we will call {\em bags}, $\D^1, \ldots, \D^p$, where each 
 set  $\D^j$ consists of a set of requests $R^{j}_1, \ldots, R^j_{n_j}$. As before, each request $R^j_i$ is specified by an interval $I^j_i$ and a bandwidth
requirement $d^j_i$. We are also given profits $p^j$ associated with each of the bags $\D^j$.
A feasible solution to such an instance picks at most one demand from each of the bags -- the selected demands should form a feasible set of routable demands. The profit of such a solution is the total profit of the bags from which we select a demand. The goal is to maximize the total profit.

We require our instances to satisfy NBA. We use the notion of congestion, bottleneck capacity, large demands and small demands, as defined in \autoref{chap2}. We will use ideas from {\ufpround} to give a constant factor approximation for {\ufpmax}, and then extend it to {\ufpbag}.

\section{Linear Programming formulation for Max-UFP}
A natural linear programming formulation for {\ufpmax} on a path is given below. Here $x_i$ denotes the fraction of the demand $i$ that is satisfied and $I_i$ is the unique path between $s_i$ and $t_i$.
\begin{align*}
\textbf{maximize} \quad \; \sum_{i=1}^k w_i x_i \! & \qquad \qquad \textbf{(UFP-LP)} \\
\textbf{such that} \quad \sum_{i:e \in I_i} d_i x_i &\le c_e \qquad \! \forall e \in E \\
0 \le x_i &\le 1 \qquad \; \forall i \in \{1,\ldots,k\}
\end{align*}
If we replace the constraints $x_i \in [0,1]$ by the constraints $x_i \in \{0,1\}$, we get an integer program, which precisely models {\ufpmax}.

\begin{definition}
The \emph{integrality gap} of an integer program is the worst-case ratio over all
instances of the problem of the value of an optimal solution to the integer programming formulation
to the value of an optimal solution to its linear programming relaxation.
\end{definition}

\subsection{Integrality gap of the UFP-LP without NBA}
Chakrabarti et al. \cite{ChakrabartiCGK07} showed that the integrality gap of the above LP is $\Theta\left(\log \frac{d_{\max}}{d_{\min}}\right)$. This can be as bad as $\Omega(n)$ without NBA, as the example in \autoref{bad-integrality-gap} shows. In this example, $c(e_i) = 2^i$ for $i = 1,\ldots,n$. There is a demand of $2^i$ between $v_i$ and $v_{n+1}$ for $i = 1,\ldots,n$. For all such demands the profit is 1. Note that the optimum integral solution can route at most one demand, to get a profit of $OPT = 1$, while the optimal fractional LP solution can route each demand to the extent of $\frac{1}{2}$ ($x_i = \frac{1}{2}$), to get a profit of $OPT_f = \frac{n}{2}$. Hence, $\frac{OPT_f}{OPT} = \frac{n}{2} = \Omega(n)$. Note that in this example, NBA is not satisfied. Further, $d_{\max} = 2^n$ and $d_{\min} = 2$, so the bound $\Theta\left(\log \frac{d_{\max}}{d_{\min}}\right)$ is asymptotically tight.

\begin{figure}[ht]
\centering
\includegraphics[scale=0.7]{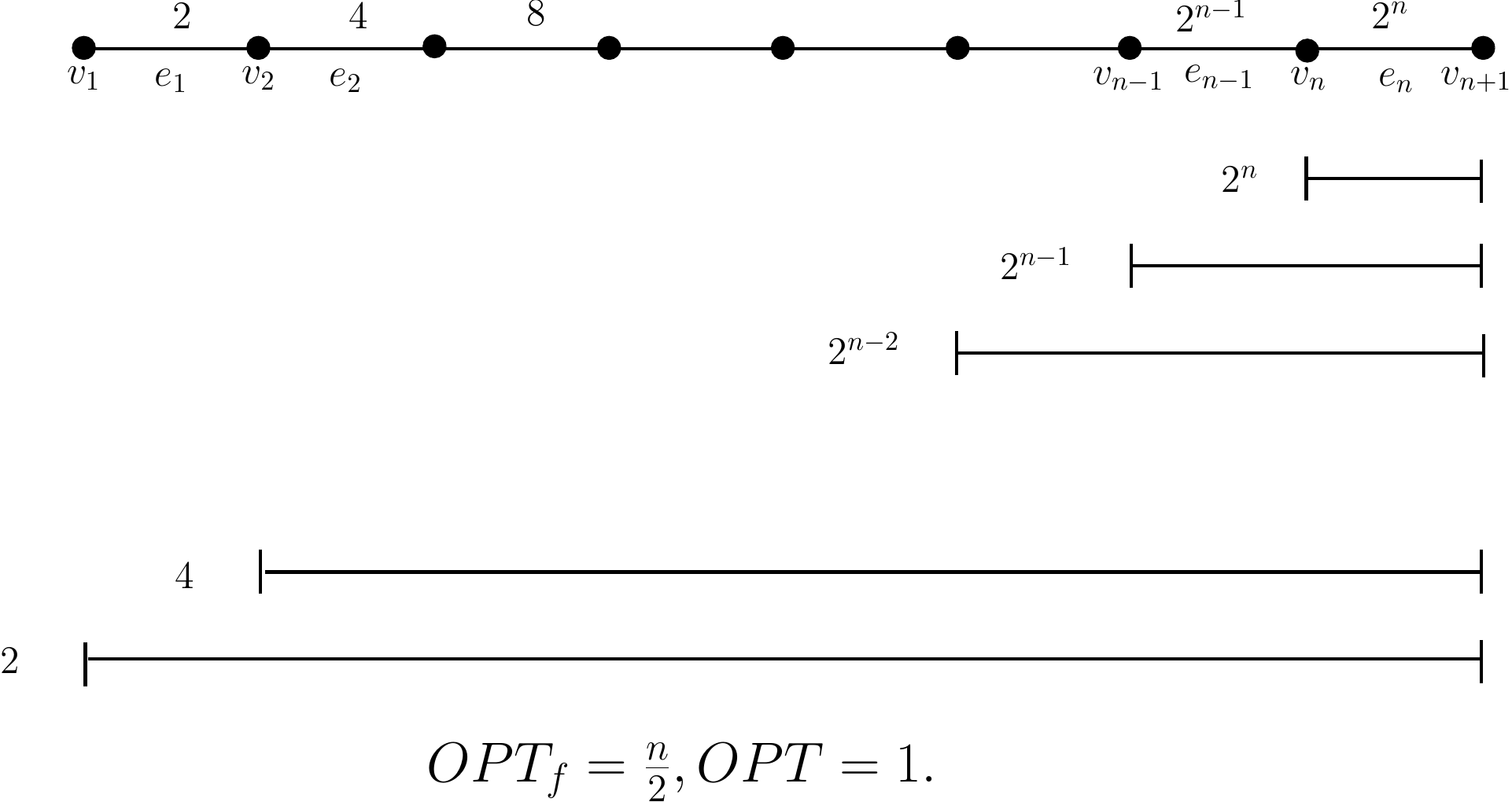}
\caption{An example where $OPT_f = \frac{n}{2}, OPT = 1.$}
\label{bad-integrality-gap}
\end{figure}

\subsection{Integrality gap of the UFP-LP with NBA}
In \autoref{bad-integrality-gap-nba}, the capacities and demands are as shown. All profits are 1. Here $d_1 = 2c, d_2 = d_3 = c+\epsilon$. The LP has a feasible solution given by $x_1 = \frac{1}{2}, x_2 = x_3 = \frac{c}{c+\epsilon}$. Hence, LP has a profit of $\frac{1}{2} + \frac{2c}{c+\epsilon} \approx 2.5$. Since routing any demand integrally will block the other demands, the IP can get a profit of at most 1. Hence, the integrality gap of the UFP-LP on this example is 2.5.

\begin{figure}[ht]
\centering
\includegraphics[scale=0.7]{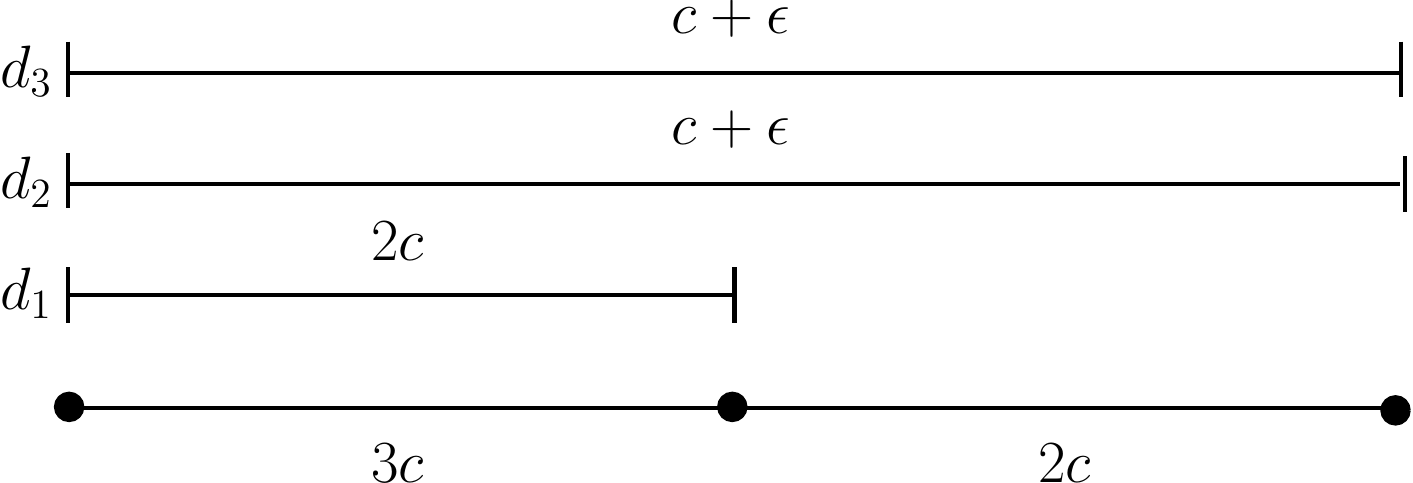}
\caption{Integrality gap of 2.5 for paths.}
\label{bad-integrality-gap-nba}
\end{figure}

\section{Approximation Algorithm for Max-UFP}
\label{sec:max}
In this section we show how ideas from {\ufpround} can be used to derive a constant factor approximation algorithm for {\ufpmax}.
Consider an instance $\I$ of {\ufpmax}. As before, we divide the demands into small and large demands.
For large demands, Chakrabarti et al.~\cite{ChakrabartiCGK07} 
showed that one can find the optimal solution by dynamic
programming. For completeness, we include the result below.

\begin{lemma}
\label{lem:maxl}
The number of $\delta$-large demands crossing any edge in a feasible solution is at most $\frac{2}{\delta}\left(\frac{1}{\delta} - 1\right)$. Hence, an optimum solution can be found in $n^{O(1/\delta^2)}$ time using dynamic programming.
\end{lemma}

\begin{figure}[ht]
\centering
\includegraphics[scale=0.7]{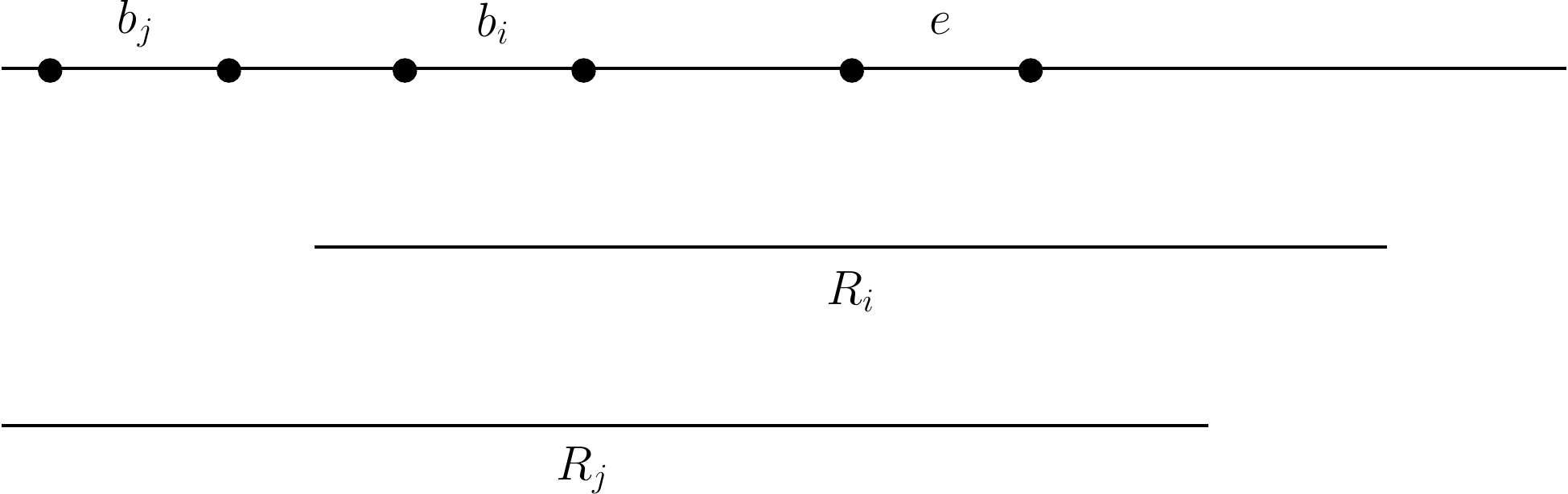}
\caption{Illustration for the analysis of large demands.}
\label{large-demands}
\end{figure}

\begin{proof}
Consider the set of requests $S_e$ passing through the edge $e$ in a feasible solution. The bottleneck edge of any such request will either be $e$ or on its left or on its right. Let $S_l$ be the set of requests whose bottleneck edge is $e$ or on its left. Similarly, let $S_r$ be the set of requests whose bottleneck edge is on the right of $e$. Among all requests in $S_l$, let $R_i$ be the one whose bottleneck edge $b_i$ is the rightmost. Then all requests in $S_l$ will pass through $b_i$, since they pass through $e$ and their bottleneck edge is on the left of $b_i$.

Consider any other demand $R_j \in S_l, j \ne i$. We know that $d_i > \delta c_{b_i}$, whereas $d_i \le c_{\min} \le c_{b_j}$, using NBA. Together this gives, $c_{b_j} > \delta c_{b_i}$. Hence, $d_j > \delta c_{b_j} > \delta^2 c_{b_i}$. The load put by $R_i$ on $b_i$ is $d_i > \delta c_{b_i}$. The load put by $R_j$ for $j \ne i$ on $b_i$ is $d_j > \delta^2 c_{b_i}$. The total number of such $R_j$ is strictly less than $\frac{(1 - \delta)c_{b_i}}{\delta^2 c_{b_i}}$. Together with the request $R_i$, the number of requests in $S_l$ is at most $\frac{1}{\delta}\left(\frac{1}{\delta} - 1\right)$. Similarly, $|S_r| \le \frac{1}{\delta}\left(\frac{1}{\delta} - 1\right)$. Hence, $|S_e| \le |S_l| + |S_r| \le \frac{2}{\delta}\left(\frac{1}{\delta} - 1\right)$.
\end{proof}

Note that, according to our definition, large demands are $\frac{1}{4}$-large. Now we consider the small demands. The following lemma gives an approximation algorithm for small demands.

\begin{lemma}
\label{lem:maxs}
If there are only small jobs, then there is a $16$-approximation algorithm for {\ufpmax}.
\end{lemma}

\begin{proof}
We write the following natural LP relaxation for this problem -- a variable $x_i$
for demand $D_i$ which is 1 if we include it in our solution, and 0 otherwise.
\begin{align}
\nonumber
\max & \sum_i w_i x_i \\
\label{eq:cap}
\sum_{i : e \in I_i} d_i x_i & \leq c_e \ \ \ \ \mbox{for all edges $e$} \\
\nonumber
0 & \leq x_i \leq 1 \ \ \ \ \mbox{for all demands $i$}
\end{align}

Let $x^\star$ be an optimal solution to the LP relaxation. Let $K$ be an integer such that all
 the variables
$x^\star_i$ can be written as $\frac{\alpha_i}{K}$ for some integer $\alpha_i$. Now we construct an
instance $\I'$ of {\ufpround} as follows. For each (small) demand $D_i$ in $\I$, we create $\alpha_i$
copies of it. Rest of the parameters are same as those in $\I$. First observe that inequality~(\ref{eq:cap})
implies that $\sum_{i : e \in I_i} d_i \alpha_i \leq K c_e, \forall e \in E$. Thus, the congestion of each edge in $\I'$ is at most $K$. Using Lemma \ref{lem:small} for small demands, we can color the demands with at most $16 K$ colors. It follows that the best solution
among these $16 K$ solutions will have profit at least $\frac{1}{16} \cdot \sum_i w_i x_i^\star$.
\end{proof}

Thus, we get the following theorem.

\begin{theorem}
\label{thm:max}
There is a $17$-approximation algorithm for the {\ufpmax} problem.
\end{theorem}
\begin{proof}
Given an instance $\I$, we divide the demands into large and small demands. For large demands, we
compute the optimal solution using Lemma \ref{lem:maxl}, whereas for small demands we compute a solution with
approximation ratio 16 using Lemma \ref{lem:maxs}. Then we pick the better of the two solutions.

Consider an optimal solution $\opt$ with profit $\profit(\opt)$. Let $\profit^l(\opt)$ be the profit for large demands and
$\profit^s(\opt)$ be the profit for small demands. If $\profit^l(\opt) \geq \frac{1}{17} \cdot \profit(\opt)$,
then our solution for large demands will also be at least $\frac{1}{17} \cdot \profit(\opt)$.
Otherwise, $\profit^s(\opt) \geq \frac{16}{17} \cdot \profit(\opt)$. In this case, our solution for small demands will have value at least $\frac{1}{16} \cdot \frac{16}{17} \cdot
\profit(\opt) = \frac{1}{17} \cdot \profit(\opt)$.
\end{proof}

\subsection{Running time}
\label{sec:ufpmax-time}
We can find the optimal solution for the instance containing only large demands using dynamic programming in polynomial-time. Indeed, since the large demands are $\frac{1}{4}$-large, we can compute the optimal solution in $O(n^{16})$ time using \autoref{lem:maxl}. For small demands, we have to find an optimal solution to the linear programming relaxation for {\ufpmax}. This can be done in polynomial-time using the ellipsoid method. Solving the {\ufpround} instance $\I'$ using \autoref{lem:small} can also be done in polynomial-time. We can make $K$ polynomial in the input as follows. If the value of the variables $x^\star_i$ in the LP are less than $\inv{k}$ ($k$ is the number of demands), then we can ignore them. Otherwise, we can round them to the nearest multiple of $\inv{k}$. This will cause a small error of at most $\inv{k}$, which can be ignored. Since, $K$ can be taken as the least common multiple of the denominators of the variables $x^\star_i$, this will make $K$ polynomial in the input. Hence, the overall running time of the algorithm is polynomial.

\section{Approximation Algorithm for Bag-UFP}
\label{sec:bmax}
We now extend the above algorithm to the {\ufpbag} problem. Consider an instance $\I$ of this problem.
As before, we classify each of the requests $R^j_i$ as either large or small. For each bag, $\D^j$,
let $\D^{j,\lar}$ be the set of large demands in $\D^j$ and $\D^{j, \sm}$ be the set of small demands
in $\D^j$. Again, we have two different strategies for large and small demands.
\begin{lemma}
\label{lem:bagl}
If there are only large jobs, then there is a $48$-approximation algorithm for {\ufpbag}.
\end{lemma}
\begin{proof}
Suppose, we have the further restriction that the selected intervals need to be disjoint. From
\autoref{lem:maxl}, we know that the number of $\frac{1}{4}$-large demands crossing any edge in a feasible solution is at most $2 \cdot 4 \cdot (4 - 1) = 24$. Hence, if the demands are disjoint, the value of the objective function will reduce by a factor of at most 24. However, for
the latter  problem, we can use the 2-approximation algorithm of Berman et al.~\cite{BermanD00} 
and Bar-Noy et al.~\cite{Bar-NoyGNS99}. This gives a $48$-approximation algorithm.
\end{proof}

\begin{lemma}
\label{lem:bags}
If there are only small jobs, then there is a $17$-approximation algorithm for {\ufpbag}.
\end{lemma}
\begin{proof}
As in the case of {\ufpmax} problem, we first write an LP relaxation, and then
use an algorithm similar to the
one used for the {\ufpround} problem. We have a variable $x_i^j$
for demand $D_i^j$, which is 1 if we include it in our solution and 0 otherwise, and a variable $y^j$ which is 1 if we choose a demand from the bag $\D^j$ and 0 otherwise. The LP relaxation is as follows.
\begin{align}
\nonumber
\max & \sum_j p^j y^j \\
\label{eq:bagcap}
\sum_{i : e \in I_i^j} d^j_i x^j_i & \leq c_e \ \ \ \ \mbox{for all edges $e$} \\
\label{eq:bag}
\sum_{i} x^j_i & \le y^j \ \ \ \ \mbox{for all bags $\D^j$} \\
\nonumber
0 & \leq x^j_i \leq 1 \ \ \ \ \mbox{for all demands $i$} \\
\nonumber
0 & \leq y^j \leq 1 \ \ \ \ \mbox{for all bags $\D^j$}
\end{align}

Let $x,y$ be an optimal solution to the LP above. Again, let $K$ be a large enough
integer such that $y^j = \frac{\alpha^j}{K}, x^j_i = \frac{\beta^j_i}{K}$, where $\alpha_j$ and
$\beta^j_i$ are integers for all $j$ and $i$. Now we consider an instance of {\ufpround} where
we have $\beta^j_i$ copies of the demand $D^j_i$. The only further restriction is that no two demands
from the same bag can get the same color. Inequality~(\ref{eq:bagcap}) implies that $\sum_{i : e \in I_i^j} d^j_i \beta^j_i \leq K c_e, \forall e \in E$. So the congestion bound
is $K$. We proceed as in the proof of Lemma~\ref{lem:small}, except that now we have $17K$ different solutions. When we consider the demand $D^j_i$, we ignore the solutions which contain a demand from the bag $\D^j$. Inequality~(\ref{eq:bag}) implies that $\sum_{i} \beta^j_i \le \alpha^j \le K, \forall j$. Hence, there will be at most $K$ such solutions.
For the remaining $16K$ solutions, we argue as in the proof of Lemma~\ref{lem:small}.
\end{proof}

\begin{theorem}
\label{thm:bag}
There is a $65$-approximation algorithm for the {\ufpbag} problem.
\end{theorem}
\begin{proof}
This follows from the two previous lemmas. We argue as in the proof of Theorem~\ref{thm:max}.
\end{proof}

\subsection{Running time}
For the instance containing only large demands, we are using the 2-approximation algorithms of Berman et al.~\cite{BermanD00} or Bar-Noy et al.~\cite{Bar-NoyGNS99}, both of which runs in polynomial-time. Hence, the instance containing only large demands can be solved in polynomial-time. For small demands, we have to find an optimal solution to the linear programming relaxation for {\ufpbag}. This can be done in polynomial-time using the ellipsoid method. Solving the constructed {\ufpround} instance using \autoref{lem:small} can also be done in polynomial-time. We can make $K$ polynomial in the input using the technique in \autoref{sec:ufpmax-time}. Hence, the overall running time of the algorithm is polynomial.

\section{Approximation Algorithm for Max-UFP on Trees}
Consider an instance $\I$ of {\ufpmax} on trees. We will show how the approximation algorithm for the {\ufpround} problem can be used to obtain a constant factor approximation algorithm for the {\ufpmax} problem.

\begin{theorem}
\label{thm:maxtree}
There is a $64$-approximation algorithm for the {\ufpmax} problem on trees.
\end{theorem}

\begin{proof}
We write the following natural LP relaxation for this problem -- a variable $x_i$
for demand $D_i$ which is 1 if we include it in our solution, and 0 otherwise.
\begin{align}
\nonumber
\max & \sum_i w_i x_i \\
\label{eq:captree}
\sum_{i : e \in I_i} d_i x_i & \leq c_e \ \ \ \ \mbox{for all edges $e$} \\
\nonumber
0 & \leq x_i \leq 1 \ \ \ \ \mbox{for all demands $i$}
\end{align}

Let $x^\star$ be an optimal solution to the LP relaxation. Let $K$ be an integer such that all
 the variables
$x^\star_i$ can be written as $\frac{\alpha_i}{K}$ for some integer $\alpha_i$. Now we construct an
instance $\I'$ of {\ufpround} as follows. For each demand $D_i$ in $\I$, we create $\alpha_i$
copies of it. Rest of the parameters are same as those in $\I$. First observe that inequality~(\ref{eq:captree})
implies that $\sum_{i : e \in I_i} d_i \alpha_i \leq K c_e, \forall e \in E$. Thus, the congestion of each edge in $\I'$ is at most $K$. Using \autoref{thm:tree}, we can color the demands with at most $64K$ colors. It follows that the best solution
among these $64K$ solutions will have profit at least $\frac{1}{64} \cdot \sum_i w_i x_i^\star$.
\end{proof}

Although this is worse than the 48-approximation algorithm of Chekuri et al. \cite{ChekuriMS07}, this illustrates the power of our approach. We can handle all these problems in a unified framework.

\subsection{Running time}
We can find an optimal solution to the linear programming relaxation for {\ufpmax} on trees in polynomial-time using the ellipsoid method. Constructing the {\ufpround} instance $\I'$ can also be done in polynomial-time. The {\ufpround} instance $\I'$ can be solved in polynomial-time as shown in \autoref{sec:roundtree}. We can make $K$ polynomial in the input using the technique in \autoref{sec:ufpmax-time}. Hence, the overall running time of the algorithm is polynomial.

\chapter {Online Algorithms for the Interval Coloring Problem}
\label {chap4}

In this chapter, we consider the {\ufpround} problem in an on-line setting. As before, we are given a path $G = (V,E)$ with edge capacities $c_e$ on edge
$e$. Requests arrive in an on-line manner. A request $R_i$ is specified by a triplet $(s_i, t_i, d_i)$, where $s_i$ is the starting vertex, $t_i$ is
the destination vertex and $d_i$ is the actual bandwidth requirement. We shall also use $I_i$ to denote the interval $[s_i, t_i]. $
The on-line algorithm needs to color the demand on its arrival, such that
the set of demands with the same color can be routed feasibly in the path $G$. The goal is to minimize the number of colors. Again, we shall assume that
the requests satisfy the no-bottleneck assumption (NBA). Indeed, without this assumption,  it is known that any deterministic on-line algorithm will have  competitive ratio of $\Omega\left(\max \left\{\log \log n, \log \log \log \left(\frac{c_{\max}}{c_{\min}}\right)\right\}\right)$, where
$c_{\max}$ and $c_{\min}$ are the maximum and minimum edge capacities of the path respectively~\cite{EpsteinEL09}.
\newcommand{\hatc}{{\hat c}}
\section{Preliminaries}
We fix a time $n$, and consider the requests which have arrived till time $n$, i.e., $R_1, \ldots, R_n$. Recall that for an edge $e$, the load $l_e$
on $e$ is the total demand of requests which contain $e$, i.e., $\sum_{i : e \in I_i} d_i$. Also, the congestion on $e$, $r_e =\left\lceil \frac{l_e}{c_e}\right\rceil$. Let $r = \max_e r_e$ be the maximum congestion on any edge.
 Clearly, $r$ is a lower bound on the minimum number of colors required to color the requests.
 For a set of requests $S$, let $l_e(S) = \sum_{i:e \in R_i, R_i \in S} d_i$ be the load put by the requests in $S$ on edge $e$.

We can assume without loss of generality that $c_{\min} = 1$. Since NBA is satisfied, this implies $d_{\max} \leq 1$. We now round down the edge capacities $c_e$ to the nearest power of 2. Let $\hatc_e$ denote these rounded capacities. Note that $\hatc_{\min}$ remains 1.

The \emph{bottleneck edge} $b_i$ of a request $R_i$ is an edge of minimum capacity (with respect to $\hatc$) in $I_i$, i.e., $b_i = \argmin_{e \in I_i} \hatc_e$. The capacity of the bottleneck edge, i.e.,
 $\hatc(b_i)$ is called the \emph{bottleneck capacity} of the request $R_i$. The \emph{class} of a request $R_i$ is defined as $\ell_i = \log_2 \hatc(b_i)$. Note that the class of a request can be between $0$ and $\log_2 \hatc_{\max}$.

For a request $R_i$ in class $j \ge 1$, we shall call it a {\em small demand} if $d_i \le \min(1, 2^{j-3})$. Since, $\hatc(b_i) = 2^j$, $d_i \le \frac{\hatc(b_i)}{8}$. For a demand $d_i$ in class $0$, we call it a small demand if $d_i \le \frac{1}{4}$. Since, $\hatc(b_i) = 1$, $d_i \le \frac{\hatc(b_i)}{4}$. Otherwise, we shall call the request a {\em large demand}. Note that large demands can exist only in classes 0, 1 and 2 (see the table below).

\begin{table}[ht]
\begin{center}
\begin{tabular}{|c|c|c|c|}
\hline
\textbf{Class} & \textbf{Small demands} & \textbf{Large demands} & \textbf{Bottleneck capacity} \\
\hline
0 & $\left(0, \frac{1}{4}\right]$ & $\left(\frac{1}{4}, 1\right]$ & 1 \\
1 & $\left(0, \frac{1}{4}\right]$ & $\left(\frac{1}{4}, 1\right]$ & 2 \\
2 & $\left(0, \frac{1}{2}\right]$ & $\left(\frac{1}{2}, 1\right]$ & 4 \\
3 & $(0, 1]$ & \textsc{none} & 8 \\
\vdots & \vdots & \vdots & \vdots \\
$j$ & $(0, 1]$ & \textsc{none} & $2^j$ \\
\hline
\end{tabular}
\caption{Schematic representation of classes and capacities of demands}
\end{center}
\end{table}

\section{Our algorithm}
In this section, we give a $58$-competitive algorithm for the online interval coloring problem. When a request comes, we determine whether it is small or large. We handle small demands and large demands separately. For small demands, we give a $32$-competitive algorithm. For large demands, we give a $26$-competitive algorithm. The details of the algorithms are given in the following sections.

\subsection{Small demands}
In this section, we give a $32$-competitive algorithm for small demands. For this, we first consider a special case when all edges have the same capacity. We shall then show that for non-uniform capacities, we
can derive several instances of uniform capacity instances. We can then apply our algorithm for uniform capacities to each of these instances.

\subsubsection{Uniform Capacities}
\label{sec:uniform}
We consider the online {\ufpround} problem for the special case when all edges have capacity $1$ and each demand $d_i$ is at most $1/4$. This is without any loss of generality, since we can always scale the demands with the common capacity $c$ to make the capacity of each edge to be $1$. We call it the {\ufprounduniform} problem.

 We shall assign each
arriving request a {\em level}. Let $S_l$ be the set of requests which have been assigned level $l$. We shall show that the set of requests in each level can be colored with one color. Suppose we have already processed
requests $R_1, \ldots, R_{i-1}$. Suppose these requests have been partitioned into levels $S_1, \ldots, S_{k_i}$. When the request $R_i$ arrives, we find the smallest index $k$ such that for every edge $e \in I_i$,
the total load of the requests in $\cup_{k'=1}^k S_{k'}$ (including $R_i$), i.e., $l_e(\cup_{k'=1}^k S_{k'} \cup \{R_i\})$, is at most  $\frac{k}{4}$. If no such index is found,  start a new level $k_i+1$, and assign $R_i$
to $S_{k_i+1}$. For an edge $e$ and level $k$, we say that $e$ is {\em critical} for $R_i$ on level $k$, if $e \in I_i$ and $l_e(\cup_{k'=1}^k S_{k'} \cup \{R_i\}) > \frac{k}{4}$.
Note that $e$ is an edge which prevented $R_i$ to be put on level $k$. The complete algorithm is given as \autoref{algo-small}.

\begin{algorithm}
\textbf{Algorithm} Color\;
\textbf{Input}: Demands $R_i = (s_i,t_i,d_i)$ coming online\;
\textbf{Output}: A feasible coloring of demands\;
\tcp{$k_i$ is the number of levels used.}
$k_i \leftarrow 1$\;
\While{there are still requests in the input}
{
	let $R_i = (s_i,t_i,d_i)$ be the next request\;
	find the smallest level $k \in \{1,\ldots,k_i\}$ such that for every edge $e \in I_i$,
the total load of the requests in $\cup_{k'=1}^k S_{k'}$ (including $R_i$), i.e., $l_e(\cup_{k'=1}^k S_{k'} \cup \{R_i\})$, is at most  $\frac{k}{4}$\;
	\If{no such level is found}
	{
		$k_i = k_i+1$\;
		go to line 7.
	}
	\tcp{assign $R_i$ to level $k$.}
	$S_k \leftarrow S_k \cup \{R_i\}$\;
}
\caption{An online algorithm for $\frac{1}{4}$-small demands and unit edge capacity}
\label{algo-small}
\end{algorithm}

We now analyze this algorithm. Let $r$ denote the maximum congestion of an edge if we consider the requests $R_1, \ldots, R_n$. As argued earlier, $r$ is a lower bound on the minimum number of colors needed to color
these demands.
\begin{lemma}
\label{small-uniform1}
The number of levels $k_n$ is at most $4r$.
\end{lemma}

\begin{proof}
Consider the first request $R$ which gets assigned to level $S_{k_n}$. It must be the case there is a critical edge $e$ for $R$ on level $k_n-1$. So,
$l_e\left(\left(\bigcup_{k'=1}^{k_n-1} S_{k'}\right) \cup \{R\}\right) > \frac{1}{4}(k_n-1)$. On the other hand, since the maximum congestion on any edge is $r$, $l_e\left(\left(\bigcup_{k'=1}^{k_n-1} S_{k'}\right) \cup \{R\}\right) \le r$. Together, this implies that $\frac{1}{4}(k_n-1) < r$. Hence, $k_n < 4r + 1$, which implies that $k_n \le 4r$, since $k_n$ is an integer. This proves the desired result.
\end{proof}

\begin{lemma}
\label{small-uniform2}
For a level $k$ and edge $e$, the load on $e$ by demands in $S_k$ is at most $1$.
\end{lemma}

\begin{proof}
First consider the case $k=1$. For an edge $e$, let $R_j$ be the last demand  containing $e$ which was added to $S_1$. Then, by the algorithm, total load on $e$ (including $R_j$) is at most $1/4$.

Now, consider $k > 1$. We call an edge $e$ critical if it is critical for some demand in $S_k$ on level $k-1$. Note that
each demand in $S_k$ must contain at least one critical edge (otherwise it should have been added to level $k-1$ or earlier). Fix a critical edge $e$. Let the demands containing $e$ which get added to $S_k$ (in the order of arrival) be $R_{k_1}, \ldots, R_{k_j}$. Note that $e$ must be critical for $R_{k_j}$ on level $k-1$. Hence, $l_e \left( \cup_{k'=1}^{k-1} S_{k'} \right)  \geq \frac{k-1}{4}-d_j$, where $d_j$ is the demand of request $R_{k_j}$. Since $R_{k_j}$ is added
to $S_k$, it must be the case that $l_e \left( \cup_{k'=1}^{k} S_{k'} \right) \leq \frac{k}{4}$. Subtracting the second inequality from the first we get, $l_e(S_k) \leq \frac{1}{4} + d_j \leq \frac{1}{2}$, since $d_j \leq \frac{1}{4}$.

Now consider an edge $e$ which is not critical. Let $e_L$ and $e_R$ be the nearest critical edges on its left and right respectively. Clearly, any request in $S_k$ containing $e$ must contain either $e_L$ or $e_R$. But 
the total load on the latter edges is at most $1/2$. Hence, the load on $e$ is at most $1$. This proves the lemma. 
\end{proof}

\noindent
We now conclude with the  main result of this section. 
\begin{lemma}
\label{small-uniform4}
The number of colors required by our algorithm is at most $4r$. Hence, it is a $4$-competitive algorithm for the {\ufprounduniform} problem.
\end{lemma}

\begin{proof}
From \autoref{small-uniform1}, we know that the number of levels is at most $4r$.  \autoref{small-uniform2} shows that we can  color  the requests in each level using one color. Hence, the number of colors required  is at most $4r$.
\end{proof}

\subsection{Algorithm for Small Demands}
We now describe our algorithm for small demands, where we will use the $4$-competitive algorithm for the {\ufprounduniform} problem. For each class $l$, we create a new instance of the {\ufprounduniform} problem $\I_l$, where all requests are of class $l$. We shall use the algorithm of the previous section to color the demands in $\I_l$. If $l=0$, the path in $\I_0$ is the same as the path $G$, but we set all edge capacities to 1.
Now consider the case $l \ge 1$.
We first contract all edges  $e$ in $G$ for which $\hatc_e < 2^l$. For the remaining edges, we
set their capacity to $2^{l-1}$. This gives the path in instance $\I_l$. Observe that in $\I_l$, the demands $d_i$ are at most $2^{l-3}$, and hence at most $1/4$ times the capacity of the edges (by definition of small demands), and so $\I_l$ is indeed an instance of the {\ufprounduniform} problem. Note that if a demand of class $l$ contains an edge $e$, then $\hatc_e \geq 2^l$, and
so this edge will not get contracted in the path in $\I_l$. 

We can now describe the algorithm for coloring small demands. When a demand of class $l$ arrives, we color it using \autoref{algo-small} on $\I_l$. Hence, the number of colors used by our algorithm is the maximum over
all values of $l$ of the number of colors needed for coloring $\I_l$.  

Since a particular color may be present in several of the colorings for the instance $\I_l$, we need to show that we can indeed put together the requests which have been colored with this color in different instances $\I_l$.
\begin{lemma}
For a color $c$, let $S_l^c$ be the requests of class $l$ which get colored with $c$ (in $\I_l$). Then $\cup_l S_l^c$ form a feasible set of requests in $G$ with edge capacities $c_e$. 
\end{lemma}

\begin{proof}
Fix an edge $e$ with $\hatc_e = 2^k$. Then this edge appears in $\I_0, \ldots, \I_k$ with edge capacities $1, 2, \ldots, 2^{k-1}$. Further, no demand in $\cup_{l \geq k+1} S_l^c$ contains $e$. 
Hence, the total load on $e$ due to the demands in $\cup_l S_l^c$ is at most $1 + 2 + \cdots + 2^{k-1} = 2^k \leq c_e$. 
\end{proof}

Let $r$ be the maximum congestion of any edge on the instance containing only the small demands. 
\begin{lemma}
\label{small1}
The number of colors used by our algorithm is at most $32r$. 
\end{lemma}

\begin{proof}
Fix a class $l$. Let $r_l$ be the maximum congestion of an edge in $\I_l$. Let $r$ denote the maximum congestion of any edge in the original instance $\I$. We first argue that 
$r_l \leq 8 r$. We define another instance $\hat{\I}_l$, which is the same as $\I_l$ except that the edge $e$ has capacity $\hatc_e$. Let $\hat{r}_l$ be the maximum congestion of any edge in $\hat{\I}_l$. 
We first argue that $r_l \leq 2 \hat{r}_l$. Indeed, if $e$ is an edge in $\I_l$, then we know that $\hatc_e \geq 2^l$. We also know that any demand of class $l$ must contain 
at least one edge $e$ with $\hatc_e = 2^l$. So, for an edge $e \in \I_l$, either (i) $\hatc_e = 2^l$, in which case congestion on $e$ is at most $\hat{r}_l$, or (ii) $\hatc_e > 2^l$. 
In the latter case, let $e_L$ and $e_R$ be the nearest edges on the left and the right of $e$ with rounded capacities $2^l$. Now, any demand in $\I_l$ which passes through
$e$ must contain either $e_L$ or $e_R$. Hence, congestion on $e$ is at most $2 \hat{r}_l$. 

Now, we argue that $\hat{r}_l \leq 4 r$. Consider an edge $e$ with $\hatc_e = 2^l$. In $\I_l$, we set the capacity of this edge to $2^{l-1}$. Hence, the capacity of this edge in $\I_l$ is at least 
$c_e/4$. So, congestion of $e$ in $\I_l$ is at least $r_e/4$. So, we get that $r \ge \frac{\hat{r}_l}{4}$. Equivalently, $\hat{r}_l \leq 4 r$. Thus, we get $r_l \le 2 \hat{r}_l \leq 8r$. Using Lemma~\ref{small-uniform4}, our algorithm colors the demands in $\I_l$ using at most 
$4 r_l \leq 32r$ colors. This proves the desired result. 
\end{proof}

\section{Large demands}
We now describe our algorithm for coloring large demands. Recall that large demands exist only in classes 0, 1 and 2. 
We will color them using the  algorithm given in~\cite{Narayanaswamy04}. The colorings for class 0 and class 2 will share colors, but these will be disjoint from the coloring for class 1. For convenience, we state the result below. 

\begin{theorem}
\label{Narayanaswamy}
\cite{Narayanaswamy04}
Suppose all edges in the path have capacity 1 and let $r$ be the maximum congestion of an edge. 
The number of colors required for coloring requests with demands in  $\left(0, \frac{1}{4}\right]$, $\left(\frac{1}{4},\frac{1}{2}\right]$ and $\left(\frac{1}{2},1\right]$ are at most $4r$, $3 r$ and $3r$ respectively. Hence, coloring all requests requires at most $10r$ colors.
\end{theorem}

\noindent 
We define three instances $\I_0,\I_1$ and $\I_2$. For $l \in \{0,1\}$, we construct the instance $\I_l$ by contracting all edges $e$ for which $\hatc_e \neq 2^l$. For $\I_2$, we contract all edges $e$ for which $\hatc_e \neq 2^l$ and then reduce the capacity of edges by half. We first show that it is sufficient to color large demands of class $l$ using $\I_l$ only. 

\begin{lemma}
\label{large1}
Fix a class $l \in \{0,1,2\}$. 
Consider a coloring of demands of class $l$ restricted to the instance $\I_l$. Then, coloring of class 1 is feasible. Further, for any color $c$, the set of demands in $\I_0 \cup \I_2$ which are colored with $c$ is feasible.
\end{lemma}

\noindent
{\bf Note:} It is possible that for a demand $R_i$ of class $l$, we contracted some of the edges in $I_i$ while constructing the path in $\I_l$. Hence, while considering
the coloring in $\I_l$, we will consider only those edges of $I_i$ which do not get contracted. 
\begin{proof}
We would have contracted two types of edges in $\I_l$~: 
\begin{itemize}
\item \textbf{Edges $e$  with  $ \hatc_e < 2^l$}: Since no class $l$ demand passes through them, contracting these edges does not matter. 
\item \textbf{Edges $e$ with  $\hatc_e > 2^i$}: Consider such an edge $e$. So, $\hatc_e \geq 2^{l+1}$. Consider any coloring of class $l$ requests in $\I_l$. Let $e_L$ and $e_R$ be the 
nearest edges with $\hatc$ values $2^l$ to the left and the right of $e$ respectively (in the original graph). Then, any request of class $l$ through $e$ must contain either $e_L$ or
$e_R$. Hence, the total load on $e$ due to such demands (of this color) is at most $\hatc_{e_L} + \hatc_{e_R} \leq 2^{l+1}$. Hence, demands of this color do not violate the edge capacity of
$e$.\qedhere
\end{itemize}
\end{proof}

We now show how to color class $l$ demands in the instance $\I_l$. Let $r^{(l)}$ denote the maximum congestion of an edge in $\I_l$ (where the requests are all the class $l$ demands). 

\begin{lemma}
\label{lem:classes}
We can color the demands of class $l$, for $l \in \{0,1,2\}$ using the  following number of colors: 
\begin{itemize}
\item {\bf Class 0 demands :} These can be colored with at most $6 r^{(0)}$ colors. 
\item {\bf Class 1 demands :} These can be colored with at most $7 r^{(1)}$ colors. 
\item {\bf Class 2 demands :} These can be colored with at most $3 r^{(2)}$ colors. 
\end{itemize}
\end{lemma}
\begin{proof}
First consider class 0 demands. So the demands lie in the range $\left( \frac{1}{4}, 1 \right]$, and the capacity of each edge in $\I_0$ is 1. We now partition the demands into two parts: $\left(\frac{1}{4},\frac{1}{2}\right]$ and $\left(\frac{1}{2},1\right]$. The claim now follows from \autoref{Narayanaswamy}. Now consider class 1 demands. They have demands in 
$\left(\frac{1}{4},1\right]$ and all edges have capacity 2. We scale down edge capacities and demands by a factor of 2. Now demands lie in the range $\left(\frac{1}{8},\frac{1}{2}\right]$. We partition these
into two parts based on their demands : $\left(\frac{1}{8},\frac{1}{4}\right]$ and $\left(\frac{1}{4},\frac{1}{2}\right]$. The result again follows from \autoref{Narayanaswamy}. Finally, we consider
class 2 demands. These have demands in the range $ \left(\frac{1}{2},1\right]$ and all edges have capacity 2. We scale down these values by a factor of 2. So now the demands lie in the range 
$[\left(\frac{1}{4},\frac{1}{2}\right]$, and the result again follows from \autoref{Narayanaswamy}.
\end{proof}

\begin{lemma}
\label{large2}
We can color all the large demands in the original instance using at most $26r$ colors.
\end{lemma}

\begin{proof}
Note that since the rounded edge capacities $\hatc_e \leq 2 c_e$, $r^{(l)} \leq 2r$ for $l \in \{0,1\}$. Since we halve the edge capacities, $r^{(2)} \le 4r$. We will use Lemma~\ref{lem:classes}. For class $\I_1$ at most $7 r^{(1)} \le 14r$ colors are required. Colors for $\I_0$ and $\I_2$ can be shared. These two classes can be colored using $\max \left(6 r^{(0)}, 3 r^{(2)} \right) \le \max \left(6 \cdot 2r, 3 \cdot 4r \right) \le 12r$ colors. Thus, the total number of colors needed for coloring large demands is at most $14r+12r=26r$.
\end{proof}

Our final algorithm now colors the small and the large demands separately. Combining Lemma~\ref{small1} and Lemma~\ref{large2}, we get 

\begin{theorem}
\label{main}
Our algorithm for online {\ufpround} is 58-competitive. 
\end{theorem}

\begin{proof}
We can color all the small demands using at most $32r$ colors and all the large demands using at most $26r$ colors. We know that $\opt \ge r$. So, the total number of colors required to color any instance is at most $32r+26r=58r \le 58 \cdot \opt$. Hence, the result follows.
\end{proof}

\subsection{Running time}
Requests can be grouped based on their classes in polynomial-time. Finding the appropriate level of a request in Algorithm \ref{algo-small} (for uniform capacities) can be done in polynomial-time. Since, for small demands, we are using Algorithm \ref{algo-small} at most a polynomial number of times, the resulting algorithm runs in polynomial-time. For large demands, there are only three classes -- 0, 1 and 2. For each class, we are using the algorithm in \cite{Narayanaswamy04}, which runs in polynomial-time. Hence, the algorithm for large demands also runs in polynomial-time. Hence, the overall algorithm runs in polynomial-time.

\chapter {Scheduling Resources for a Partial Set of Jobs}
\label {chap5}

\section{Introduction}
We consider the problem of allocating resources to schedule jobs.
We are given a path $G$, and a set of jobs.
Each job $j$ is specified by a triplet $(s_j, t_j, d_j)$, where
$[s_j, t_j]$ denotes the interval corresponding to the job (also denoted by $I_j$),
and $d_j$ is its demand requirement. We shall assume that $d_j$ values are 1.
Further, we are also given a set of resources.
Each resource is specified by its starting and ending vertex, and  the capacity it offers and its associated cost.
A feasible solution is a set of resources satisfying the constraint that for any edge,
the sum of the capacities offered by the resources containing this edge is at least the demand required by
the jobs containing that edge, i.e., the selected resources must cover the jobs.
We call this the Resource Allocation problem ({\ResAll}).

We study two variants of the problem. The first variant is the partial covering version.
The second variant is the prize collecting version. We study these variants for the case where the  solution is allowed to pick multiple copies of a resource by paying proportional cost.

\section{Problem Definition}
We consider the graph $G=(V,E)$ which is a path with vertices numbered $1, 2, \ldots, |V|$ from left to right.
An input instance  consists of a set of  {\em jobs} ${\cal J}$, and a set of {\em resources} $\calI$.

Each job $j \in {\cal J}$ is specified by an interval $I_j = [s_j,t_j]$ in the path. Recall that each job has
demand requirement of 1.
Each resource $i \in \calI$ is specified
by an interval $I_i=[s(i),e(i)]$ in the path, capacity $w_i$ and cost $c_i$. We shall assume that the capacities $w_i$ are
integers.
We interchangeably refer to the resources as {\em resource intervals}. We shall also refer to the interval $I_j$ (or $I_i$) as
the {\em span} of the job $j$ (or resource $i$).
A typical scenario of such a collection of jobs and resources is shown in \autoref{fig:cc1}.

\begin{figure*}[t]
\begin{center}
\fbox{
\includegraphics[scale=0.5]{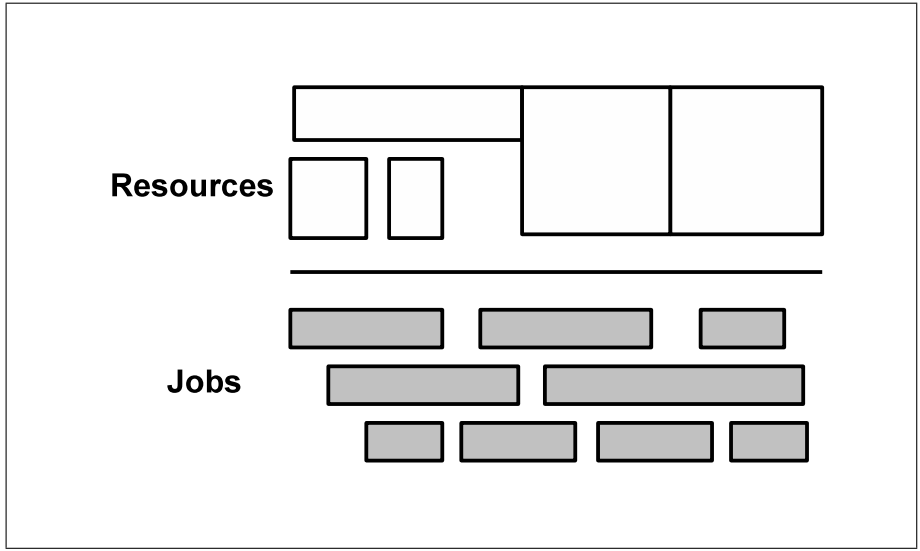}
}
\end{center}
\caption{Illustration of the input}
\label{fig:cc1}
\end{figure*}

We say that a job $j$ (or resource $i$) {\em contains} an edge $e$ if the associated interval $I_j$ (or $I_i$) contains $e$; we  denote this as $j \sim e$ ($i \sim e$).
We define a {\em profile} $P: E \rightarrow \mathbb{N}$ to be a mapping that assigns an integer value
to every edge of the path. For two profiles, $P_1$ and $P_2$, $P_1$ is said to {\em cover} $P_2$,
if $P_1(e) \geq P_2(e)$ for all $e \in E$.
Given a set $J$ of jobs, the profile $P_J(\cdot)$ of $J$ is defined to be the mapping determined by
the cumulative demand of the jobs in $J$, i.e. $P_J(e) = |\{ j \in J:j \sim e \}|$.
Similarly, given a multiset $R$ of resources, its profile is: $P_R(e) = \sum_{i \in R:i \sim e} w_i$
(taking copies of a resource into account).
We say that $R$ {\em covers} $J$ if $P_R$ covers $P_J$.
The cost of a multiset of resources $R$ is defined to be the sum of the costs of all the resources
(taking copies into account).

We now describe the two versions of the problem.
\begin{itemize}
\item  {\PResAll}: In this problem, the input also specifies a number $k$ (called the {\em partiality parameter})
       that indicates the
	number of jobs to be covered. A feasible solution is a pair $(R,J)$ where $R$ is a multiset of resources
	and $J$ is a set of jobs such that $R$ covers $J$ and $|J| \ge k$. The cost of the solution is the sum of the costs of the resources in $R$ (taking copies into account).
	The problem is to find a feasible solution of minimum cost.
\item  {\PCResAll}: In this problem, every job $j$ also has a penalty $p_j$ associated with it.
	A feasible solution is a pair $(R,J)$ where $R$ is a multiset of resources
	and $J$ is a set of jobs such that $R$ covers $J$.
	The cost of the solution is the sum of the
	costs of the resources in $R$ (taking copies into account) and the penalties of the jobs not in $J$.
	The problem is to find a feasible solution of minimum cost.
\end{itemize}

\section{Outline of the Main Algorithm}
\label{sec:overview}
In this section, we outline the proof of our main result:

\begin{theorem}
\label{thm:xAAA}
There exists an $O(\log (n+m))$-approximation algorithm for the {\PResAll} problem,
where $n$ is the number of jobs and $m$ is the number of resources.
\end{theorem}

The proof of the above theorem goes via the claim that the input set of jobs can be
partitioned into a logarithmic number of {\em mountain ranges}.
A collection of jobs $M$ is called a {\em mountain} if there exists an edge $e$, such that
all the jobs in this collection contain $e$; the specified edge where the jobs
intersect will be called the {\em peak} edge of the mountain (see Figure~\ref{fig:aa};
jobs are shown on the top and the profile is shown below).
The justification for this linguistic convention is that if we look at the profile of such a
collection of jobs, the profile forms a bimodal sequence, increasing in height until the peak, and
then decreasing.
The {\em span} of a mountain $M$ is the set of edges which are contained in any of the jobs in the mountain,
i.e., $\cup_{j \in M} I_j$.
A collection of jobs $\calM$ is called a {\em mountain range}, if the jobs can be partitioned into
a sequence $M_1, M_2, \ldots, M_r$ such that each $M_i$ is a mountain and the spans of any two mountains
are non-overlapping (see Figure \ref{fig:bb}).

\begin{figure*}[t]
\begin{center}
\fbox{
\includegraphics[scale=0.5]{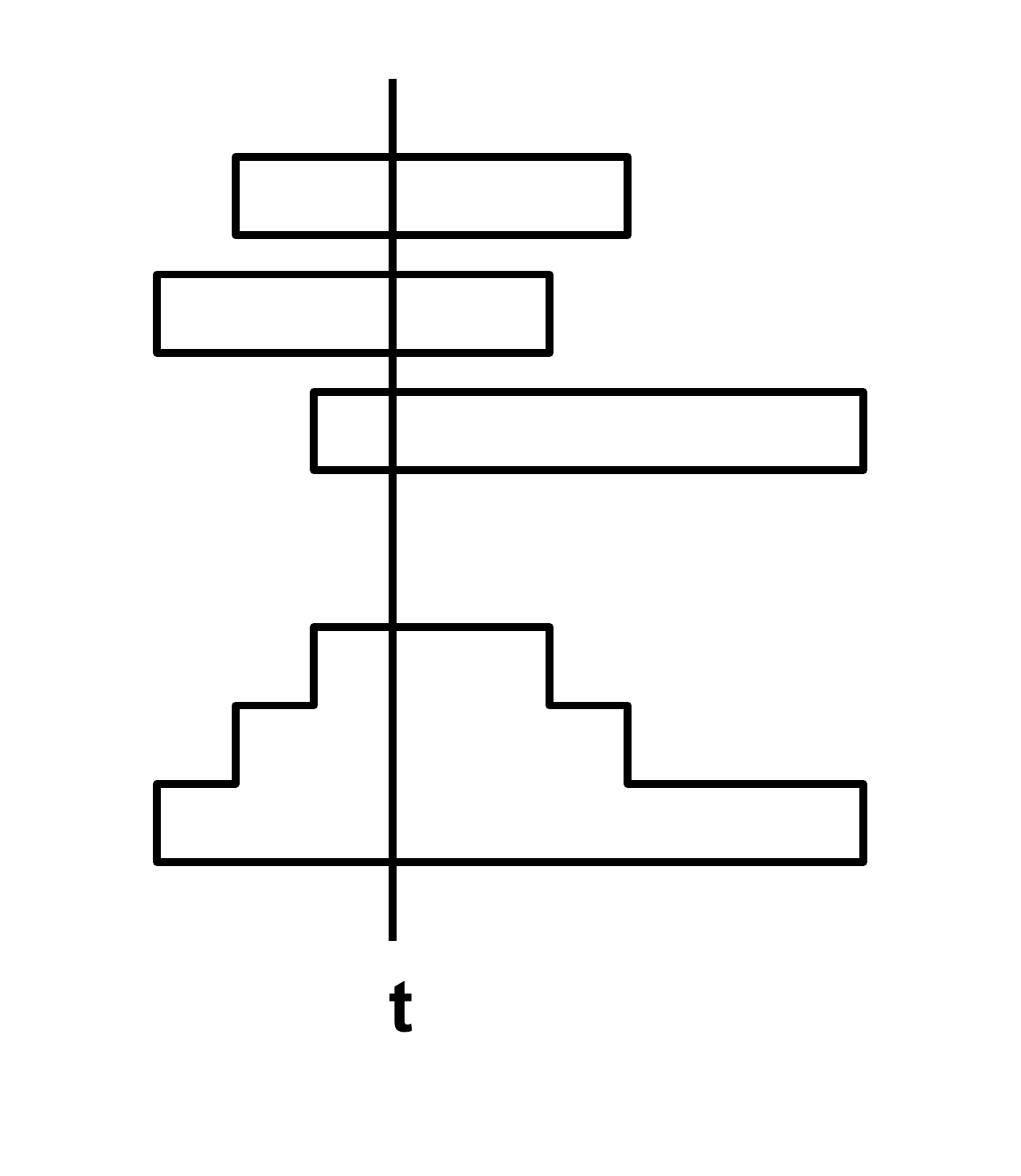}
}
\end{center}
\caption{A Mountain $M$}
\label{fig:aa}
\end{figure*}

\begin{figure*}[t]
\begin{center}
\fbox{
\includegraphics[scale=0.5]{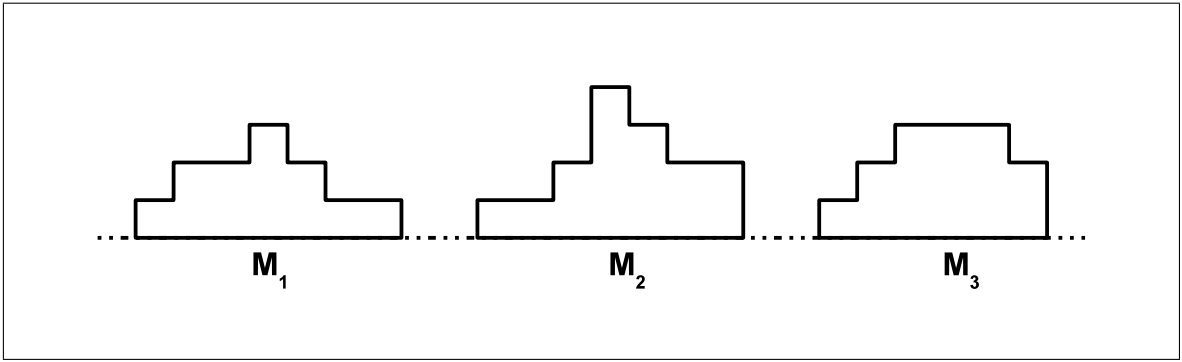}
}
\end{center}
\caption{A Mountain Range ${\cal M}=\{M_1, M_2, M_3\}$}
\label{fig:bb}
\end{figure*}

We prove a decomposition lemma which shows that the input set of jobs can be partitioned into a logarithmic number of mountain ranges. Hence, our decomposition lemma implies that it is sufficient to get a good approximation for the case of a mountain range. It is not difficult to argue that one can extend this result to several mountain ranges by employing dynamic programming. We only need to know how many jobs to satisfy in each mountain range. For a single mountain range, we will prove the following result.

\begin{theorem}
\label{thm:xCCC}
There exists a constant factor approximation algorithm for the special case of the {\PResAll} problem,
wherein the input set of jobs form a single mountain range $\calM$.
\end{theorem}

\noindent To prove \autoref{thm:xCCC}, we need the following results.

\begin{enumerate}
	\item A constant factor approximation for the case of a mountain.
	\item Extending this result to a mountain range.
\end{enumerate}

The first part is accomplished by the following theorem. The proof is given in
Section \ref{sec:mountain}.

\begin{theorem}
\label{thm:xDDD}
There exists a $8$-approximation algorithm for the special case of the {\PResAll} problem
wherein the input set of jobs form a single mountain $M$.
\end{theorem}

For the second part, we will collapse each mountain into a single edge. This can be done if resources are wide, i.e., they span the mountains which they intersect. But this may not always be the case. We need to solve a related problem.

{\it  Problem Definition (\lspc):}
We are given a demand profile over the set of edges $E$,
which specifies an integral demand $d_e$ for every edge $e$.
The input resources are of two types, {\em short} and {\em long}.
A short resource spans only one edge, whereas a long resource can span one or more edges.
Each resource $i$ has a cost $c_i$ and a capacity $w_i$.
The input also specifies a {\em partiality parameter} $k$.
A feasible solution $S$ consists of a multiset of resources $S$ and a coverage profile:
an integer $k_e$ for each edge $e$ satisfying $k_e \leq d_e$.
The solution should have the following properties:
(i) $\sum_e k_e \geq k$;
(ii) at any edge $e$, the sum of capacities of the resource intervals from $S$ containing $e$ is at least $k_e$;
(iii) for any edge $e$, at most one of the short resources containing $e$
is picked (however, multiple copies of a long resource may be included).
The objective is to find a feasible solution having minimum cost.
See Figure~\ref{fig:dd} for an example (in the figure, short resources are shaded).

\begin{theorem}
\label{thm:xEEE}
There exists a $16$-approximation algorithm for the {\lspc} problem.
\end{theorem}

\begin{figure*}[t]
\begin{center}
\fbox{
\includegraphics[scale=0.5]{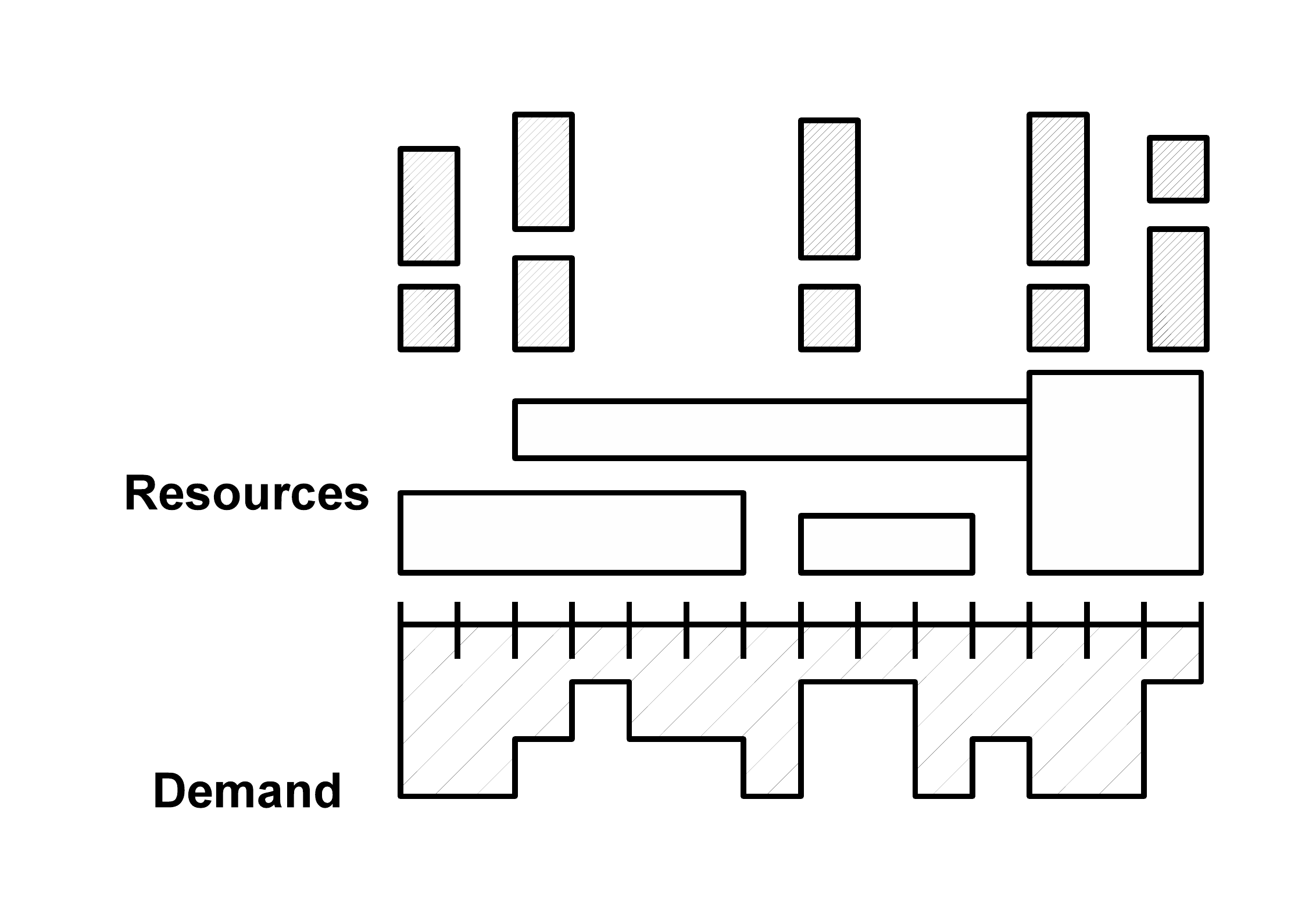}
}
\end{center}
\caption{The {\lspc} problem}
\label{fig:dd}
\end{figure*}

\section{Overview of Our Algorithm}
In this section, we give an overview of our algorithm and describe the various results needed to prove the claimed approximation guarantee. We start with some notations.

\label{sec:mountain}
For a job $j$, let its {\em length} be $\ell_j = |I_j|$.
Let  $\ell_{\min}$ be the shortest job length, and $\ell_{\max}$ the longest job length.
The proof of \autoref{lem:XXX} is inspired by the algorithm for the {\ufpmax} problem, due to Bansal et al.~\cite{BansalFKS09}.

\begin{lemma}
\label{lem:XXX}
The input set of jobs can be partitioned into groups, $\calM_1, \calM_2, \ldots, \calM_L$, such that
each $\calM_i$ is a mountain range and $L \leq 4\cdot\lceil{\log \frac{\ell_{\max}}{\ell_{\min}}}\rceil$.
\end{lemma}

\begin{proof}
\label{sec:DDD}
We first categorize the jobs according to their  lengths into $r$ categories $C_1$, $C_2$, $\cdots, C_r$, where
$r = \lceil \log \frac{\ell_{\max}}{\ell_{\min}}\rceil$.
The category $C_i$ consists of all the jobs with lengths in the range $[2^{i-1}{\ell_{\min}}, 2^i{\ell_{\min}})$.
Thus all the jobs in any single category have comparable lengths:
any two jobs $j_1$ and $j_2$ in the category satisfy $\ell_{1} < 2\ell_{2}$, where
$\ell_1$ and $\ell_2$ are the lengths of $j_1$ and $j_2$ respectively.

Consider any category $C$ and let the lengths of the jobs in $C$ lie in the range $[\alpha, 2\alpha)$.
We claim that the category $C$ can be partitioned into $4$ groups $G_0, G_1, G_2, G_3$, such that
each $G_i$ is a mountain range.
To see this,
we divide the set of jobs in $C$ into classes $H_1, H_2, \ldots, H_q, \ldots$ where $H_q$ consists of the jobs
containing the vertex $q \cdot \alpha$ (a job contains a vertex if the associated interval contains
this vertex). Here $q$ can possibly take any integer value.
Note that every job belongs to some class since all the jobs have
length at least $\alpha$; if a job belongs to more than one class, assign it to any one class arbitrarily.
Clearly, each class $H_q$ forms a mountain because any job in $H_q$ contains the vertex $q \alpha$.
For $0 \le i \le 3$, let $G_i$ be the union of the classes $H_q$ satisfying $q \equiv (i \mod 4)$.
Since each job has length at most $2\alpha$, two classes $H_j$ and $H_{j+4}$ can't have an overlap, as they are separated by a distance $4\alpha$. Hence, each $G_i$ is a mountain range.
Thus, we get a decomposition of the input jobs into $4r$ mountain ranges.
\end{proof}

Assuming \autoref{lem:XXX} and Theorem~\ref{thm:xCCC}, we now outline the proof of
Theorem~\ref{thm:xAAA}.
Let the optimal solution consists of $k_i$ jobs from mountain range $\calM_i, i=1,\ldots,L$ (where $L$ is given by \autoref{lem:XXX}), such that $k=\sum_i k_i$.
Thus, if we knew $k_1, k_2, \ldots, k_L$, we could invoke Theorem \ref{thm:xCCC} on each mountain range $\calM_i$
(along with $k_i$ as the partiality parameter) to determine a set of resources $R_i$ having cost within a constant factor
of the optimum for this mountain range. Taking the union of $R_1, R_2, \ldots, R_L$ yields a feasible solution $R$
for the original problem instance. It is not difficult to argue that $R$ is within a factor of $cL$ of the optimum solution.
The only issue in the above approach is that we do not know the values $k_1, k_2, \ldots, k_L$ (guessing them explicitly
would take exponential time). However, this issue can be handled by using dynamic programming. The details are given below.

\begin{proof}[Proof of Theorem~\ref{thm:xAAA}]
Assuming Theorem \ref{thm:xCCC},  we prove Theorem \ref{thm:xAAA}.
Let $\cJ$ be the input set of jobs, $\calI$ be the input set of resources and $k$ be the partiality parameter.
Invoke Lemma \ref{lem:XXX} on the input set of jobs $\cJ$ and obtain a partitioning of $\cJ$
into mountain ranges $\calM_1, \calM_2, \ldots, \calM_{L}$, where $L=4\cdot \ceil{\log(\ell_{\max}/\ell_{\min})}$.
Theorem \ref{thm:xCCC} provides a $c$-approximation algorithm $\cA$ for the {\PResAll} problem wherein
the input set of jobs form a single mountain range, where $c$ is some constant.
We shall present a $(cL)$-approximation algorithm for the {\PResAll} problem.

For $1\leq q \leq L$ and $1\leq \kappa \leq k$, let $\cA(q, \kappa)$ denote the
cost of the (approximately optimal) solution returned by the algorithm in Theorem \ref{thm:xCCC}
with $\calM_q$ as the input set of jobs, $\calI$ as the input set of resources and $\kappa$
as the partiality parameter.
Similarly, let $\opt(q, \kappa)$ denote the cost of the optimal solution for covering $\kappa$ of the jobs in the
mountain range $\calM_q$. Theorem~\ref{thm:xCCC} implies that $\cA(q,\kappa) \leq c\cdot \opt(q,\kappa)$.

The algorithm employs dynamic programming.
We maintain a $2$-dimensional DP table $\DP[\cdot,\cdot]$.
For each $1\leq q\leq L$ and $1\leq \kappa \leq k$,
the entry $\DP[q,\kappa]$ would
store the cost of a (near-optimal) feasible solution
covering $\kappa$ of the jobs from $\calM_1 \cup \calM_2 \cup \cdots \cup \calM_q$.
The entries are calculated as follows.
\begin{equation*}
\DP[q,\kappa] = \min_{\kappa' \leq \kappa} \{\DP[q-1, \kappa - \kappa'] + \cA(q,\kappa')\}.
\end{equation*}

The above recurrence relation considers covering $\kappa'$ jobs from the mountain $M_q$,
and the remaining $\kappa - \kappa'$  jobs from the mountains $M_1, \cdots, M_{q-1}$.
Using this dynamic program, we compute a feasible solution to the original problem instance
(i.e., covering $k$ jobs from all the mountain ranges $\calM_1,\calM_2, \ldots, \calM_L$);
the solution would correspond to the entry $\DP[L,k]$.
Consider the optimum solution $\opt$ to the original problem instance.
Suppose that $\opt$ covers $k_q$ jobs from the mountain range $\calM_q$ (for $1\leq q \leq L$),
such that $k_1 + k_2 +\cdots + k_L = k$.
Observe that
\begin{eqnarray*}
\DP[L,k]
\leq \sum_{q=1}^L \calA(q, k_q)
\leq c\cdot{\sum_{q=1}^L \opt(q, k_q)},
\end{eqnarray*}
where the first statement follows from the construction of the dynamic programming table
and the second statement follows from the guarantee given by algorithm $\calA$.
However the maximum of $\opt(q, k_q)$ (over all $q$) is a lower bound for $\opt$
(we cannot say anything stronger than this since $\opt$ might
use the same resources to cover jobs across multiple subsets $\calM_q$).
This implies that $\DP[L,k] \leq c\cdot L \cdot\opt$. This proves
the $(cL)$-approximation ratio.

It is easy to see that $L$ is $O(\log(n+m))$ as argued below.
It suffices if we consider only those vertices  where some job or resource
starts or ends; the other vertices can be ignored.
Such a transformation will not affect the set of feasible solutions.
Thus, without loss of generality, we can assume that the number of vertices is at most $2(n+m)$.
Therefore, $\ell_{\max} \leq 2(n+m)$ and $\ell_{\min} \geq 1$.
Hence, the overall algorithm has an $O(\log (n+m))$ approximation ratio.
\end{proof}

We now sketch the proof of Theorem \ref{thm:xCCC}. As mentioned earlier, there are two parts: single mountains and extension to mountain ranges via the {\lspc} problem.

For the case of a single mountain, we prove \autoref{thm:xDDD}. The basic intuition is as follows. Given the structure of the jobs, we will show that there is a
{\em near-optimal} feasible solution that exhibits a nice property:
the jobs discarded from the solution are extremal either in their
left end-points or their right end-points.
Let $\J = \{ j_1, j_2, \ldots, j_n \}$ be the input set of jobs.

\begin{lemma}
\label{BBB}
Consider the {\PResAll} problem for a single mountain.
Let $S=(R_S, J_S)$ be a feasible solution such that $R_S$ covers the set of jobs $J_S$ with $|J_S|=k$.
Let $C_S$ denote its cost.
Let $L = < l_1, l_2, \ldots, l_n >$ denote the jobs in increasing order of their left end-points.
Similarly, let $R = < r_1, r_2, \ldots, r_n >$ denote the jobs in decreasing order of their right end-points.
Then, there exists a feasible solution $X=(R_X, J_X)$ having cost at most $2 \cdot C_S$ such that
\begin{equation}\label{eqn:aa}
{\J} \setminus J_X = \{ {l_i} : i \le q_1 \} \cup \{ {r_i} : i \le q_2 \}
\end{equation}
for some $q_1,q_2 \ge 0$ where $|\cJ \setminus J_X|=n-k$.
\end{lemma}
\begin{proof}
We give a constructive proof to determine the sets $J_X$ and $R_X$.
We initialize the set $J_X$=${\cal J}$. At the end of the algorithm,
the set $J_X$ will be the desired set of jobs covered by the solution.
The idea is to remove the jobs that extend most to the right or the left from the consideration of $J_X$.
The most critical aspect of the construction is to ensure that whenever we exclude any job from consideration
of $J_X$ that is already part of $J_S$, we do so in pairs of the leftmost and rightmost extending jobs of $J_S$
that are still remaining in $J_X$.
We terminate this process when the size of $J_X$ equals the size of $J_S$, i.e., $k$.
We also initialize the set $U=\phi$.
At the end of the algorithm, this set will contain the set of jobs removed from ${\J}$ that belonged to $J_S$
while constructing $J_X$.

We now describe the construction of $J_X$ formally.
We maintain two pointers {\lptr} and {\rptr}; {\lptr} indexes the jobs in the sequence ($L$) of their left end-points
and {\rptr} indexes the jobs in the sequence ($R$) of their right end-points.
We keep incrementing the pointer {\lptr} and removing the corresponding job from $J_X$
(if it has not already been removed) until either the size of $J_X$ reaches $k$
or we encounter a job (say {\ljob}) in $J_X$ that belongs to $J_S$;  we do not yet remove the job {\ljob}.
We now switch to the pointer {\rptr} and start incrementing it and removing the corresponding job from $J_X$
(if it has not already been removed) until either the size of $J_X$ reaches $k$
or we encounter a job (say {\rjob}) in $J_X$ that belongs to $J_S$;  we do not yet remove the job {\rjob}.
If the size of $J_X$ reaches $k$, we have the required set $J_X$.

Now suppose that $|J_X| \ne k$.
Note that both {\lptr} and {\rptr} are pointing to jobs in $S$.
Let {\ljob} and {\rjob} be the jobs pointed to by {\lptr} and {\rptr} respectively (note that
these two jobs may be same).

We shall remove one or both of {\ljob} and {\rjob} from $J_X$ and put them in $U$.
We classify these jobs into three categories: {\em single}, {\em paired} and {\em artificially paired}.

Suppose that $|J_X| \ge k+2$.
In this case, we have to delete at least 2 more jobs; so we delete both {\ljob} and {\rjob} and add them to $U$
as {\em paired} jobs.
In case {\ljob} and {\rjob} are the same job, we just delete this job and add it to $U$ as a {\em single} job.
We also increment the {\lptr} and {\rptr} pointers to the next job indices in their respective sequence.
We then repeat the same process again, searching for another pair of jobs.

Suppose that $|J_X|=k+1$.
In case {\ljob} and {\rjob} are the same job, we just delete this job and get the required set $J_X$ of size $k$;
We add this job to the set $U$ as a {\em single} job.
On the other hand, if {\ljob} and {\rjob} are different jobs,
we remove {\ljob} from $J_X$ and add it to $U$ as {\em artificially paired} with
its pair as the job {\rjob}; note that we do not remove {\rjob} from $J_X$.

This procedure gives us the required set $J_X$.
We now construct $R_X$ by simply doubling the resources of $R_S$; meaning, that for each
resource in $R_S$, we take twice the number of copies in $R_X$.
Clearly $C_X = 2 \cdot C_S$.
It remains to argue that $R_X$ covers $J_X$.
For this, note that $U=J_S-J_X$ and hence $|U|=|J_X-J_S|$ (because $|J_X|=|J_S|=k$).
We create an arbitrary bijection $f : U \rightarrow J_X-J_S$.
Note that $J_X$ can be obtained from $J_S$ by deleting the jobs in $U$ and adding the jobs of $J_X - J_S$.
We now make an important observation:
\begin{observation}
\label{obs1}
For any {\em paired} jobs or {\em artificially paired} jobs $j_1$, $j_2$ added to $U$,
all the jobs in $J_X$ are contained within the
span of this pair, i.e., for any $j$ in $J_X$, $s_j \ge \min \{ s_{j_1}, s_{j_2} \}$ and $t_j \le \max \{ t_{j_1}, t_{j_2} \}$.
Similarly for any {\em single} job $j_1$ added to $U$, all jobs in $J_X$ are contained in the span of $j_1$.
\end{observation}

For every {\em paired} jobs, $j_1$, $j_2$, Observation~\ref{obs1} implies that taking 2 copies of the
resources covering $\{ j_1, j_2 \}$ suffices to cover $\{ f({j_1}), f({j_2}) \}$.
Similarly, for every {\em single} job $j$, the resources covering $\{ j \}$ suffice to cover $\{ f(j) \}$.
Lastly for every {\em artificially paired} jobs $j_1, j_2$ where $j_1 \in U$ and $j_2 \notin U$, taking 2 copies
of the resources covering $\{ j_1, j_2 \}$ suffices to cover $\{ f({j_1}), j_2 \}$.
Hence the set $R_X$ obtained by doubling the resources $R_S$ (that cover $J_S$) suffices to cover the jobs in $J_X$.
\end{proof}

Recall that Bar-Noy et al.~\cite{Bar-Noy} presented a $4$-approximation algorithm for
the {\ResAll} problem (full cover version). Our algorithm for handling a single mountain works
as follows.
Given a mountain consisting of the collection of jobs $\J$ and the number $k$,
do the following for all possible pairs of numbers $(q_1, q_2)$ such that the set
$J_X$ defined as per Equation~\ref{eqn:aa} in Lemma~\ref{BBB} has size $k$.
For the collection of jobs $J_X$, consider the issue of selecting a minimum cost set of
resources to cover these jobs; note that this is a full cover problem. Thus, the $4$-approximation
of \cite{Bar-Noy} can be applied here.
Finally, we output the best solution across
all choices of $(q_1, q_2)$.  Lemma~\ref{BBB} shows that this is an $8$-factor approximation
to the {\PResAll} problem for a single mountain. This completes the proof of \autoref{thm:xDDD}.

\autoref{thm:xEEE} is proved in Section \ref{sec:lspc}. The reduction to the {\lspc} problem is given in Section \ref{app:red}.

\section{{\lspc} Problem: Proof of Theorem \ref{thm:xEEE}}
\label{sec:lspc}
Finally, we complete the description of our algorithm by providing a $16$-approximation algorithm for the {\lspc} problem.
We extend the notion of profiles and coverage to intervals of the path.
For an interval $[a,b]$, we say that an edge $e \in [a,b]$ if both of its end-points lie in $[a,b]$.
Let $[a,b] \subseteq [1,|V|]$ be a {\em range}.
By a profile over $[a,b]$, we mean a function $Q$ that assigns a value $Q(e)$ to each edge $e \in [a,b]$.
A profile $Q$ defined over a range $[a,b]$ is said to be {\em good}, if for all edges $e \in [a,b]$,
$Q(e)\leq d_e$ (where $d_e$ is the input demand at $e$).
In the remainder of the discussion, we shall only consider good profiles and so, we shall simply write
``profile" to mean a ``good profile".
The {\em measure} of $Q$ is defined to be the sum  $\sum_{e\in [a,b]} Q(e)$.

Let $S$ be a multiset of resources and let $Q$ be a profile over a range  $[a,b]$.
We say that $S$ is {\em good}, if for any edge $e$,
it includes at most one short resource containing $e$.
We say that $S$ covers the profile $Q$,
if for any edge $e\in [a,b]$, the sum of capacities of resources
active in $S$ and containing $e$ is at least $Q(e)$.
Notice that $S$ is a feasible solution to the input problem instance,
if there exists a profile $Q$ over the entire range $[1,|V|]$ such that
$Q$ has measure $k$ and $S$ is a cover for $Q$.
For an edge $e$, let $Q^{\sh}_S(e)$
denote the capacity of the unique short resource from $S$ containing $e$, if one exists; otherwise,
$Q^{\sh}_S(e) = 0$.

Let $S$ be a good multiset of resources and let $Q$ be a profile over a range  $[a,b]$.
For a long resource $i\in S$, let $f_S(i)$ denote the number of copies of $i$ included in $S$.
The multiset $S$ is said to be a {\em single long resource assignment cover} (SLRA cover) for $Q$,
if for any edge $e \in [a,b]$, there exists a long resource $i\in S$ such that
$w_if_S(i)\geq Q(e)-Q^{\sh}_S(e)$ (intuitively, the resource $i$ can cover the residual demand by itself,
even though there are other long resources in $S$ containing $e$).

We say that a good multiset of resources $S$ is an {\em SLRA solution} to the input {\lspc} problem instance,
if there exists a profile $Q$ over the range $[1,|V|]$ having measure $k$ such that $S$ is an SLRA cover for $Q$.
The lemma below shows that near-optimal SLRA solutions exist.

\begin{lemma}
\label{lem:SLRA}
Consider the input instance of the {\lspc} problem.
There exists an SLRA solution having cost at most 16 times the cost of the optimal solution.
\end{lemma}

To prove Lemma \ref{lem:SLRA}, we will use the following lemma, which is a reformulation of Theorem 1 {in} \cite{esa2011}.
For a multiset of resources $S$, let $c(S)$ denote its cost.
\begin{lemma}
\label{lem:esa-SLRA}
\cite{esa2011} Let $\wh{S}$ be a multiset of long resources covering a profile $\wh{Q}$ over the range $[1,|V|]$.
Then, there exists a multiset of long resources $S'$ such that $S'$ is a SLRA cover for $Q$
and $c(S')\leq 16\cdot c(\wh{S})$.
\end{lemma}

\begin{proof}[Proof of Lemma \ref{lem:SLRA}]
Let $\opt$ be the optimum solution and let $Q$ be the profile of measure $k$ covered by $\opt$.
Let $\opt_l$ and $\opt_s$ be the multiset of long and short resources contained in $\opt$, respectively.
Define $Q_l$ to be the residual profile over $[1,|V|]$: $Q_l(e)=Q(e)- Q^{\sh}_S(e)$.
The multiset $\opt_l$ covers the profile $Q_l$.
Invoke Lemma \ref{lem:esa-SLRA} on $\opt_l$ and $Q_l$ (taking $\wh{S}=\opt_l$ and $\wh{Q}=Q_l$)
and obtain a multiset of long resources $S'$ which forms a SLRA cover for $Q_l$.
Construct a new multiset $S$, by taking the union of $S'$ and $\opt_s$.
Notice that $S$ is a SLRA solution.  The cost  of $S'$ is at most 16 times the cost of $\opt_l$.
So, $S$ has cost at most 16 times the cost of $\opt$.
\end{proof}

Surprisingly, we can find the {\em optimum} SLRA solution $S^*$ in polynomial time,
as shown in Theorem~\ref{thm:xGGG} below.
Lemma \ref{lem:SLRA} and Theorem \ref{thm:xGGG} imply that
$S^*$ is a $16$-factor approximation to the optimum solution.
This completes the proof of Theorem~\ref{thm:xEEE}.

\begin{theorem}
\label{thm:xGGG}
The optimum SLRA solution $S^*$ can be found in polynomial time.
\end{theorem}

The rest of the section is devoted to proving Theorem~\ref{thm:xGGG}.
The algorithm goes via dynamic programming.
The following notation is useful in our discussion.
\begin{itemize}
\item
Let $S$ be a good set of resources consisting of only short resources, and let $[a,b]$ be a range.
For a profile $Q$ defined over $[a,b]$, and an integer $h$,
 $S$ is said to be an {\em $h$-free cover} for $Q$, if for any $e\in [a,b]$,
$Q^{\sh}_S(e)\geq Q(e)-h$.
The set $S$ is said to be an {\em $h$-free $q$-cover} for $[a,b]$,
if there exists a profile $Q$ over $[a,b]$  such that $Q$ has measure $q$ and $S$
is a $h$-free cover for $Q$.
\item
Let $S$ be a good multiset of resources and let $[a,b]$ be a range.
For a profile $Q$ defined over $[a,b]$, and an integer $h$,
the multiset $S$ is said to be an {\em $h$-free SLRA cover} for $Q$,
if for any edge $e\in [a,b]$ satisfying $Q(e)-Q^{\sh}_S(e) > h$,
there exists a long resource $i\in S$ such that $w_if_S(i)\geq Q(e)-Q^{\sh}_S(e)$.
For an integer $q$, we say $S$ is an {\em $h$-free SLRA $q$-cover} for the range $[a,b]$,
if there exists a profile $Q$ over $[a,b]$ such that $Q$ has measure $q$ and $S$
is a $h$-free SLRA cover for $Q$.
\end{itemize}
Intuitively, $h$ denotes the demand covered by long resources already selected (and their cost accounted for)
in the previous stages of the algorithm; thus,
edges whose residual demand is at most $h$ can be ignored.
The notion of ``$h$-freeness" captures this concept.

We shall first argue that any $h$-free SLRA cover $S$ for a profile $Q$ over a range $[a,b]$ exhibits
certain interesting decomposition property.
Intuitively, in most cases, the range can be partitioned into two parts (left and right),
and $S$ can be partitioned into two parts $S_1$ and $S_2$ such that
$S_1$ can cover the left range and $S_2$ can cover the right range
(even though resources in $S_1$ may contain some edges in the right range and those in $S_2$
may be contain edges in the left range).
In the cases where the above decomposition is not possible,
there exists a long resource spanning almost the entire range.

\begin{lemma}
\label{lem:decomp}
Let $[a,b]$ be any range, $Q$ be a profile over $[a,b]$ and let $h$ be an integer.
Let $S$ be a good multiset of resources providing an $h$-free SLRA-cover for $Q$.
Then, one of the following three cases holds:
\begin{itemize}
\item
The set of short resources in $S$ form a $h$-free cover for Q.
\item
{\it Vertex-cut: } There exists a vertex $v^*$, $a\leq v^*\leq b-1$, and a partitioning of $S$ into $S_1$ and $S_2$
such that $S_1$ is an $h$-free SLRA-cover for $Q_1$ and $S_2$ is an $h$-free SLRA-cover for $Q_2$,
where $Q_1$ and $Q_2$ profiles are obtained by restricting $Q$ to $[a,v^*]$ and $[v^*+1,b]$, respectively.
\item
{\it Interval-cut:}
There exists a long resource $i^*\in S$
such that the set of short resources in $S$ forms a $h$-free cover for both $Q_1$ and $Q_2$, where
$Q_1$ and $Q_2$  are the profiles obtained by restricting $Q$ to $[a,s_{i^*}-1]$ and  $[t_{i^*}+1,b]$
respectively.
\end{itemize}
\end{lemma}

We first extend the notion of an SLRA cover to subsets of edges.
Let $\calT\subseteq E$ be a set of edges and let $\wh{Q}$ be a profile over the set $\calT$.
A good multiset of resources $S$ is said to be a SLRA cover for $\calT$,
if for any edge $e\in \calT$, there exists a long resource $i\in S$ such that
$w_if_{S}(i)\geq Q(e)-Q^{\sh}_{S}(e)$. We will use the following lemma, which is a reformulation of Lemma 4 in \cite{esa2011}.

\begin{lemma}
\label{lem:esa-timecut}
Let $\wh{S}$ be a multiset consisting of only long resources.
Let $\wh{Q}$ be a profile over a non-empty set of edges  $\calT'\subseteq [a,b]$,
for some $a$ and $b$.
Suppose $\wh{S}$ is a SLRA cover for $\wh{Q}$. Then one of the following properties is true:
\begin{itemize}
\item
There exists a vertex $v^*\in [a,b-1]$ and a partition of $\wh{S}$ into $\wh{S}_1$ and $\wh{S}_2$ such that
$\wh{S}_1$ is a SLRA cover for $\wh{Q}_1$ and $\wh{S}_2$ is a SLRA cover for $\wh{Q}_2$,
where $\wh{Q}_1$ and $\wh{Q}_2$ are the profiles obtained by restricting $\wh{Q}$
to the edges in $\calT'\cap [a,v^*]$ and $\calT'\cap [v^*+1,b]$, respectively.
\item
There exists a resource $i^*\in \wh{S}$ spanning all the edges in $\calT'$.
\end{itemize}
\end{lemma}

\begin{proof}[Proof of Lemma \ref{lem:decomp}]
Consider a good multiset of resources $S$ forming a $h$-free SLRA cover for a profile $Q$ over a range $[a,b]$.
Define the set of edges $\calT'$:
\[
\calT' = \{e\in [a,b]~:~Q(e)-Q^{\sh}_S(e)>h\}.
\]
If $\calT'$ is empty, then $S$ is a $h$-free cover for $Q$;
this corresponds to the first case of Lemma \ref{lem:decomp}.
So, assume $\calT'\neq \emptyset$.
Define a profile $\wh{Q}$ over the edges in $\calT'$: for any $e\in \calT'$,
let $\wh{Q}(e) = Q(e)-Q^{\sh}_S(e)$.
Notice that $S$ is a SLRA cover for the profile $\wh{Q}$.
Invoke Lemma \ref{lem:esa-timecut} (with $\wh{S}=S$).
Let us analyze the two cases of the above lemma.
Consider the first case in Lemma~\ref{lem:esa-timecut}.
In this case, there exists a vertex $v^*$ and a partitioning of $S$ into $S_1$ and $S_2$,
with the stated properties.
In this case, we see that $S_1$ and $S_2$ are $h$-free SLRA covers
for $[a,v^*]$ and $[v^*+1,b]$, respectively. This corresponds to the second case of Lemma \ref{lem:decomp}.
Consider the second case in Lemma~\ref{lem:esa-timecut}. In this case, there exists a long resource $i^*\in S$ such that $i^*$ spans
all the edges in $\calT'$. This means that for any $e\in [a,s_{i^*}-1]$ or $e\in [t_{i^*}+1,b]$,
$Q(e)-Q^{\sh}_S(e) \leq h$. Otherwise, $e \in \calT'$ and $i^*$ will contain $e$. This corresponds to the third case of Lemma \ref{lem:decomp}.
\end{proof}

We now discuss our dynamic programming algorithm.
Let $H=\max_{e \in E} d_e$ be the maximum of  the input demands.
The algorithm maintains a table $M$ with an entry for each triple $\langle [a,b], q, h\rangle$,
where $[a,b]\subseteq [1,|V|]$, $0\leq q\leq k$ and $0\leq h\leq H$.
The entry $M([a,b],q,h)$ stores the cost of the optimum $h$-free SLRA $q$-cover for the range $[a,b]$;
if no solution exists, then $M([a,b],q,h)$ will be $\infty$.
Our algorithm outputs the solution corresponding to the entry $M([1,|V|],k,0)$; notice that
this is optimum SLRA solution $S^*$. Since we are computing $M([1,|V|],k,0)$, the computation will depend only on $k$ and not on $H=\max_{e \in E} d_e$, as $h=0$. Computation of entries in both the tables $M$ and $A$ requires polynomial time, as is evident from the recurrence relations.

In order to compute the table $M$, we need an auxiliary table $A$.
For a triple $[a,b]$, $q$ and $h$, let $A([a,b],q,h)$ be the optimum $h$-free $q$-cover for $[a,b]$
(using only the short resources); if no solution exists $A([a,b],q,h)$ is said to be $\infty$.
It is straightforward to compute the table $A$ and this is explained in Section \ref{sec:table-A}.

\begin{figure*}
\fbox{
\begin{minipage}{0.95\textwidth}
\begin{eqnarray*}
E_1 &=& A([a,b],q,h).\\
E_2 &=& \min_{\substack{c\in [a,b-1] \\ q_1\leq q}} M([a,c],q_1,h) + M([c+1,b],q-q_1,h).\\
E_3 &=&
\quad
\min_{
\substack{(i\in \calL, \alpha\leq H)~:~\alpha w_i>h\\
           q_1,q_2,q_3~:~q_1+q_2+q_3 = q
         }
}
\begin{pmatrix}
\alpha \cdot c_i
+ A([a,s_i-1],q_1,h) \\
+M([s_i,t_i],q_2,\alpha w_i)
+ A([t_i+1,b],q_3,h) \\
\end{pmatrix}
\end{eqnarray*}
\end{minipage}
}
\caption{Recurrence relation for $M$}
\label{fig:formula}
\end{figure*}

Based on the decomposition lemma (Lemma \ref{lem:decomp}), we can develop a recurrence relation
for a triple $[a,b]$, $q$ and $h$.
We compute $M([a,b],q,h)$ as the minimum over three quantities $E_1$, $E_2$ and $E_3$ corresponding to the three cases
of the lemma. Intuitive description of the three quantities is given below and precise formulas are provided
in Figure \ref{fig:formula}. In the figure, $\calL$ is the set of all long resources\footnote{The input demands
$d_e$ are used in computing the table $A(\cdot,\cdot,\cdot)$}.
\begin{itemize}
\item
{\it Case 1: }
No long resource is used and so, we just use the
corresponding entry $A([a,b],q,h)$ of the table $A$.
\item
{\it Case 2: } There exists a vertex-cut $v^*$. We consider all possible values
of $v^*$. For each possible value of $v^*$, we try all possible ways in which $q$ can be divided between the left and right ranges.
\item
{\it Case 3: }
There exists a long resource $i^*$ such that the ranges to the left of and to the right of
$i^*$ can be covered solely by short resources.
We consider all the long resources $i$ and also the number of copies $\alpha$ to be picked.
Once $\alpha$ copies of $i$ are picked, $i$ can cover all edges with residual demand at most
$\alpha w_i$ in an SLRA fashion, and so the subsequent recursive calls can ignore these edges.
Hence, this value is passed to the recursive call.
We also consider different ways in which $q$ can be split into three parts -- left, middle and right.
The left and right parts will be covered by the solely short resources and the middle part will use both
short and long resources.
Since we pick $\alpha$ copies of $i$, a cost of $\alpha c_i$ is added.
\end{itemize}
We set $M([a,b],q,h)=\min\{E_1, E_2, E_3\}$. For the base case: for any $[a,b]$, if $q=0$ or $h = H$,
then the entry is set to zero.

The order in which the entries of the table are filled is explained in Section~\ref{sec:DP-ordering}. Computation of the entries in $A$ is explained in Section~\ref{sec:table-A}.
Using Lemma \ref{lem:decomp}, we can argue that the above recurrence relation correctly computes
all the entries of $M$. For the sake of completeness, a proof is included in Section~\ref{sec:recur-proof}.

\subsection{DP Ordering}
\label{sec:DP-ordering}
Define a partial order $\prec$ as follows.
For pair of triples $z=([a,b],q,h)$ and $z'=([a',b'],q',h')$,
we say that $z\prec z'$, if one of the following properties is true:
(i)$[a',b']\subseteq [a,b]$;
(ii) $[a,b]=[a',b']$ and $q<q'$;
(iii) $[a,b]=[a',b']$, $q=q'$ and $h>h'$.
Construct a directed acyclic graph (DAG) $D$ where the triples
are the vertices and an  edge is drawn from a triple $z$ to a triple $z'$,
if $z\prec z'$. Let $\pi$ be a topological ordering of the vertices in $D$.
We fill the entries of the table $M$ in the order of appearance in $\pi$.
Notice that the computation for any triple $z$ only refers to triples appearing
earlier than $z$ in $\pi$.

\subsection{Computing the table $A$}
\label{sec:table-A}
We now describe how to compute the auxiliary table $A$.
For a triple consisting of an edge $e$, $q\leq k$ and
$h\leq H$, define $\gamma(e,q,h)$ as the cheapest cost of
covering $q-h$ demand from the short resources containing $e$.
This is a Knapsack problem and can be computed by dynamic programming. Time-complexity of the Knapsack problem is $O(n_e(q-h))$, where $n_e$ is the number of short resources containing $e$.


Then, for a triple $\langle[a,b],q, h\rangle$, the entry $A([a,b],q,h)$ is governed
by the following recurrence relation. Of the demand $q$ that needs to be covered,
the optimum solution may cover a demand  $q_1$ from the edge $e$, and a demand $q-q_1$ from the  range $[a,b-1]$.
We try all possible values for $q_1$ and choose the best:
\[
A([a,b],q,h) =
\min_{
  \substack{
     q_1\leq \min\{q, d_b\}
   }
}
A([a,b-1],q -q_1,h) + \gamma(b,q_1,h).
\]
It is not difficult to verify the correctness of the above recurrence relation.

\subsection{Correctness of the Recurrence Relation (Figure \ref{fig:formula})}
\label{sec:recur-proof}
We prove \autoref{thm:xGGG} by induction on the position in which a triple appears in the topological ordering $\pi$.
The base case corresponds to triples that do not have a parent in $D$. \autoref{thm:xGGG} is trivially true in this case.

Consider any triple $z=([a,b],q,h)$. Let $S$ be the optimum $h$-free SLRA $q$-cover for $[a,b]$.
There exists a profile $Q$ over $[a,b]$ such that $Q$ has measure $q$ and $S$ is a $h$-free SLRA cover
for $Q$. Let us invoke Lemma \ref{lem:decomp} and consider its three cases.

Suppose the first case of the lemma is true.
Let $S_s$ be the set of short resources contained in $S$.
Then, $S_s$ is a $h$-free cover for $Q$. Therefore $E_1=A([a,b],q,h)\leq c(S_s) \leq c(S)$.

Suppose the second case of the lemma is true.
Let $v^*$ be the vertex  and $S_1$ and $S_2$ be the partition given by the lemma.
Let $Q_1$ and $Q_2$ be the profiles obtained by restricting $Q$ to the ranges $[a,v^*]$ and $[v^*+1,b]$,
respectively. Let the measures of $Q_1$ and $Q_2$ be $q_1$ and $q_2$, respectively.
Then $S_1$ is a $h$-free $q_1$-cover for $[a,v^*]$ and $S_2$ is a $h$-free $q_2$-cover for $[v^*+1,b]$.
Therefore, by induction, $M([a,v^*],q_1,h)\leq c(S_1)$
and $M([v^*+1,b],q_2,h)\leq c(S_2)$.
In computing the quantity $E_2$, we try all possible ways of partitioning the range $[a,b]$ and dividing the number $q$.
Hence, $E_2\leq c(S_1)+c(S_2)$. Since $c(S)=c(S_1)+c(S_2)$, we see that $E_2\leq c(S)$.

Suppose the third case of lemma is true.
Let $i^*$ be the long resource given by the lemma.
Let $S_1$ be   short resources  in $S$ that contain edges in $[a,s_{i^*}-1]$.
Similarly, let $S_3$ be the set of short resources  in $S$ that contain edges in $[t_{i^*}+1,b]$.
Let $S_2$ be the multiset of long resources  in $S$ and the set of short resources  in $S$
that contain edges in $[a,b]$.
Let $Q_1$, $Q_2$ and $Q_3$ be the profiles obtained by restricting $Q$ to the ranges $[a,s_{i^*}-1]$,
$[s_{i^*},t_{i^*}]$ and $[t_{i^*}+1,b]$, respectively.
The lemma guarantees that $S_1$ and $S_2$ are $h$-free covers for $Q_1$ and $Q_3$ respectively.
Let $q_1$, $q_2$ and $q_3$ be the measures of $Q_1$, $Q_2$ and $Q_3$, respectively.
We see that $A([a,s_{i^*}+1],q_1,h)\leq c(S_1)$ and $A([t_{i^*}+1,b],q_3,h)\leq c(S_3)$.
Let $\alpha^*=f_S(i^*)$ be the number of copies of $i^*$ present in $S$.
Notice that if $\alpha^* w_{i^*}\leq h$, then $i^*$ is not a useful resource,
because $i^*$ will be covering only edges in $[s_{i^*},t_{i^*}]$ with residual demands at most $h$;
but all such edges are free and need not be covered.
So, without loss of generality, assume that $\alpha^* w(i^*)>h$.
Since $i^*$ spans the entire range $[s_{i^*},t_{i^*}]$,
the resource $i^*$ can cover all edges in the above range with residual demands at most $\alpha^* w(i^*)$.
Let $S_2'=S_2-\{i^*\}$. Notice that $S_2'$ is a $(\alpha^* w_i)$-free SLRA cover for the profile $Q_2$.
Therefore, $S_2'$ is a $(\alpha^* w_i)$-free $q_2$-cover for the range $[s_{i^*},t_{i^*}]$.
Hence, by induction, $M([s_{i^*},t_{i^*}],q_2,\alpha^* w(i^*)) \leq c(S_2')$.
Therefore, $E_3 \leq c(S_1)+c(S_2)+c(S_3)=c(S)$.

The quantity $E= \min \{E_1, E_2, E_3\}$; so
$E\leq c(S)$. The proof is now complete.
\qed

\subsubsection{Running time}
Computing $E_1$ requires the table entry $A([a,b],q,h)$. Once we have filled the table $A$ in polynomial time, this requires constant time. So, computing $E_1$ requires polynomial time. Computing $E_2$ requires the table entries $M([a,v^*],q_1,h)$ and $M([v^*+1,b],q_2,h)$ for all possible values of $a \le v^* \le b$. So we need to compute $2(b-a+1)$ entries of the table $M$ which is polynomial. Hence, computing $E_2$ requires polynomial time. Computing $E_3$ requires the table entries $A([a,s_{i^*}+1],q_1,h)$, $M([s_{i^*},t_{i^*}],q_2,\alpha^* w(i^*))$ and $A([t_{i^*}+1,b],q_3,h)$. Since each of them can be computed in polynomial time, computing $E_3$ requires polynomial time. Hence, the overall running time is polynomial.

\section{Single Mountain Range: Proof of Theorem~\ref{thm:xCCC}}
\label{app:red}
In this section, we prove Theorem~\ref{thm:xCCC} via a reduction to {\lspc}. Recall that in the {\lspc} problem, we are given a demand profile over the set of edges $E$,
which specifies an integral demand $d_e$ for every edge $e$.
The input resources are of two types, {\em short} and {\em long}.
A short resource spans only one edge, whereas a long resource can span one or more edges.
Each resource $i$ has a cost $c_i$ and a capacity $w_i$.
The input also specifies a {\em partiality parameter} $k$.
A feasible solution $S$ consists of a multiset of resources $S$ and a coverage profile:
an integer $k_e$ for each edge $e$ satisfying $k_e \leq d_e$.
The solution should have the following properties:
(i) $\sum_e k_e \geq k$;
(ii) at any edge $e$, the sum of capacities of the resource intervals from $S$ containing $e$ is at least $k_e$;
(iii) for any edge $e$, at most one of the short resources containing $e$
is picked (however, multiple copies of a long resource may be included).
The objective is to find a feasible solution having minimum cost.

The reduction proceeds in two steps.
\subsection{First Step}
Let the input instance be $\cA$, wherein the input set of jobs form a mountain range $\cM = \{M_1, M_2, \cdots, M_r\}$. We will transform the instance $\cA$ to an instance $\cB$, with some nice properties:
(1) the input set of jobs in $\cB$ also form a mountain range;
(2) every resource $i$ in the instance $\cB$ is either narrow or wide (see Section~\ref{sec:overview} for the definitions);
(3) the cost of the optimum solution for the instance $\cB$ is at most $3$ times the optimal cost for
the instance $\cA$;
(4) given a feasible solution to $\cB$, we can construct a feasible solution to $\cA$ preserving the cost.

Consider each resource $i$ in $\cA$ and let $M_p, M_{p+1}, \cdots, M_q$ (where $1 \leq p \leq q \leq r$) be the sequence of mountains that $i$ intersects. Clearly, $i$ fully spans the mountains $M_{p+1}, \cdots, M_{q-1}$.
We will split the resource $i$ into at most $3$ new resources $i_1, i_2, i_3$; we say that $i_1$, $i_2$ and $i_3$
are {\em associated with} $i$.
The resource $i_2$ will fully span the mountains $M_{p+1}, \cdots, M_{q-1}$.
The span of the resource $i_1$ is the intersection of the span of  $i$ with the mountain $M_p$. Likewise, the span of the resource $i_3$ is the intersection of the span of $i$ with the
mountain $M_q$.
The capacities and the costs of $i_1$, $i_2$ and $i_3$ are declared to be the same as that
of $i$. We include $i_1, i_2, i_3$ in $\cB$.
The input set of jobs and the partiality parameter $k$, in $\cB$ are identical to that of $\cA$.
This completes the reduction.

It is easy to see that the first two properties are satisfied by $\cB$.
Let us now consider third property .
Given any solution $S$ for the instance $\cA$, we can construct a solution $S'$ for $\cB$ as follows.
For each copy of resource $i$ picked in $S$, include a single copy of $i_1$, $i_2$ and $i_3$ in $S'$.
Clearly, the cost of the solution $S'$ is at most thrice that of the cost of $S$.
Regarding the fourth property, given a solution $S$ to $\cB$, we can construct a solution $S'$ to $\cA$
as follows. Consider any resource $i$ in $\cA$ and let $i_1$, $i_2$ and $i_3$ be the resources in $\cB$
associated with $i$. Let $f_1, f_2, f_3$ be the number of copies of $i_1,i_2,i_3$ picked by solution $S$.
Let $f=\max\{f_1,f_2,f_3\}$. Include $f$ copies of the resource $i$ in the solution $S'$.
It is easy to see that $S'$ is a feasible solution to $\cA$ and that the cost of $S'$ is
at most the cost of $S$.

\subsection{Second Step}
In this step we reduce the problem instance $\cB$ to an {\lspc} instance $\cC$, with the following properties:
(1) the cost of the optimum solution for the instance $\cC$ is at most $8$ times the optimal cost for
the instance $\cB$;
(2) Given a feasible solution to $\cC$, we can construct a feasible solution to $\cB$ preserving the cost.

\subsubsection*{Reduction}
In the instance $\cC$, we retain only the peak edges of the various mountains in the instance $\cB$
so that the number of edges in $\cC$ is the same as the number of mountains $r$ in $\cB$. Let the mountain ranges
in $\cB$ ordered from left to right be $M_1, \ldots, M_r$, with $e_p$ being the peak edge of $M_p$.
For any  peak edge $e$ in the instance $\cB$, let $d_e$ be the number of jobs in $\cB$ that
contain the edge $e$; we assign  demand $d_e$ to the edge $e$ in the instance $\cC$.
 For any wide resource $i$ in $\cB$, fully spanning mountains $M_p, M_{p+1}, \cdots, M_q$,
create a long resource $i'$ in $\cC$ with the span $e_p, e_{p+1}, \ldots, e_q$. The cost and capacity of $i'$ are
the same as that of $i$.

The narrow resources in the instance $\cB$ are used to construct the short resources in the instance $\cC$
as follows.
Consider any specific mountain $M$ in the instance $\cB$ along with the collection of narrow resources $R$ that are
contained in the span of $M$, and let $e$ be the peak edge of $M$. Let $\cA_{SM}$ be the algorithm implied in Theorem~\ref{thm:xDDD} for the single mountain $M$. For any integer $\kappa$ ($1 \leq \kappa \leq d_e$), we add a short resource $i_s^{e,\kappa}$ with capacity $\kappa$. The cost $C$ of this resource is determined as follows. We apply $\cA_{SM}$ on $M$, with $\kappa$ as the
partiality parameter, and the set of narrow resources $R$ as the only resources.
Then, Theorem~\ref{thm:xDDD} gives us a solution of cost $C$ consisting of a multiset $R'$ of some resources in $R$,
that covers $\kappa$ of the jobs in the mountain $M$. The cost of the short resource $i_s^{e,\kappa}$ will be $C$.
We will call the (multi)set of narrow resources $R' \subseteq R$ in the instance $\cB$
as {\em associated} with the short resource $i_s^{e,\kappa}$.  This completes the description of the instance $\cC$
of the {\lspc} problem.

\subsubsection*{Validity of the reduction}
We will now argue the validity of the reduction.
Let us consider the first property:
the cost of the optimum solution to the instance $\cC$ has cost at most $8$ times
the cost of the optimum solution to the instance $\cB$.
The following lemma is useful for this purpose.

\begin{lemma}
\label{lem:LLL}
Let $J$ be a subset of jobs and $R$ be multiset of resources in the instance $\cB$
such that $R$ covers $J$ (note that $R$ contains only narrow or wide resources and $J$ forms a mountain range).
Let $R_1$ and $R_2$ be the narrow and the wide resources in $R$ respectively. Let $R_2'$ be a multiset
constructed by picking twice the number of copies of each resource in $R_2$.
Then, $J$ can be partitioned into two sets $J_1$ and $J_2$ such that $J_1$ is solely covered by the resources in $R_1$
and $J_2$ is solely covered by the resources in $R_2'$.
\end{lemma}
\begin{proof}
For now, we assume that the mountain range comprises of a single mountain.
Let $P_R(\cdot)$, $P_{R_1}(\cdot)$, $P_{R_2}(\cdot)$ and $P_{R_2'}(\cdot)$ denote the profile of the resources in $R, R_1, R_2$ and $R_2'$ respectively.
Note that $P_{R_2}(\cdot)$ is a uniform bandwidth profile having uniform height, say $h$. This is because these correspond to wide resources, which span all of this mountain.
Let $J_L$ be the first $h$ jobs among all the jobs in $J$ sorted in ascending ordered by their left end-points.
Similarly,
let $J_R$ be the first $h$ jobs among all the jobs in $J$ sorted in descending order by their right end-points.
Intuitively, $J_L$ and $J_R$ correspond to the $h$ left-most and the $h$ right-most jobs in the mountain.

Let $J_2 = J_L \cup J_R$ and $J_1 = J \setminus J_2$.
Let $P_J(\cdot)$, $P_{J_1}(\cdot)$ and $P_{J_2}(\cdot)$ denote the profiles of the jobs in $J$, $J_1$ and $J_2$ respectively.

Note that the profile $P_{R_2'}(t)$ has height $2h$ throughout the span of the mountain whereas the profile
$P_{J_2}(\cdot)$ has height at most $2h$ at any edge. Thus $R_2'$ covers $J_2$.

We will now show that $R_1$ covers $J_1$.
Note that $P_{J_1}(e)=P_J(e)-P_{J_2}(e)$ for any edge $e$.
We partition the edges into two parts: $E_0=\{e:P_{J_1}(e)=0 \}$ and $E_{>0}=\{e:P_{J_1}(e)>0 \}$.
For the edges in $E_0$, there are no jobs remaining in $J_1$ for $R_1$ to cover.
For the edges in $\E_{>0}$, we note that $P_{J_1}(e) \le P_J(e)-h$
(because $J_2$ comprises of the left-most $h$ and right-most $h$ jobs of the mountain).
Also note that the profile $P_{R_1}(e) = P_R(e)-P_{R_2}(e) = P_R(e)-h$.
Since, $R$ covers $J$, this implies that $R_1$ is sufficient to cover $J_1$.

The proof can easily be extended to a mountain range as the mountains within a mountain range are disjoint.
\end{proof}

We are now ready to show that our reduction is valid. Let $\opt(\cB)$ and $\opt(\cC)$ be the cost of an optimal solution for the instances $\cB$ and $\cC$ respectively. 
\begin{lemma}
$\opt(\cC) \le 8 \cdot \opt(\cB)$. Further, given a feasible solution for $\cC$, one can convert it to a feasible solution for $\cB$ without increasing the cost.
\end{lemma}

\begin{proof}
Let $\opt=(R,J)$ denote the optimal solution for the problem instance $\cB$,
where $J$ is the set of jobs picked by the solution and $R$ is the set of resources covering $J$ (we have $|J|=k$).
Let $R_1$ and $R_2$ be the set of narrow and wide resources in $R$.
Apply Lemma \ref{lem:LLL} for the solution $(R,J)$ and obtain a partition of $J$ into $J_1$ and $J_2$
along with $R_1$ (covering $J_1$) and $R_2'$ (covering $J_2$).
Let $\calM = M_1, M_2, \ldots, M_r$ be the input mountain range in the instance $\cB$ with peak edges
$e_1, e_2, \ldots, e_r$, respectively. Consider any mountain $M_q$.
Let $k_q$ be the number of jobs picked in $J$ from the mountain $M_q$.
Let $R_{1,q}$ be the set of (narrow) resources from $R_1$ contained within the span of $M_q$.
Thus, the set of resources $R_{1,q}$ cover the set of jobs in $M_q\cap J_1$ and let $k_q' = |M_q\cap J_1|$.
Corresponding to the value $k_q'$, we would have included a short resource, say $i_q^{e_1,k_q'}$
in the instance $\cC$, where $e_1$ is the peak edge of $M_q$; cost of $i_q$ is at most $8$ times the cost of $R_{1,q}$
(as guaranteed by Theorem \ref{thm:xDDD}).
The set of long resources in $R_2'$ cover at least $k_q-k_q'$ jobs within the mountain $M_q$.

Construct a solution to the instance $\cC$ by including $i_1, i_2, \ldots, i_q$;
and for each copy of a wide resource $i$ in $R_2'$,
include a copy of its corresponding long resource. Notice that this is a feasible solution to the instance $\cC$.
The cost of the short resources $\{i_1, i_2, \ldots, i_q\}$ is at most $8$ times the cost of $R_1$
and the cost of the long resources is the same as that of $R_2'$, which is at most twice that of $R_2$.
Cost of $\opt$ is the sum of costs of $R_1$ and $R_2$.
Hence, cost of the constructed solution is at most $8$ times the cost of $\opt$.

We now prove the second property:
let $S$ be a given a solution to the instance $\cC$ of the {\lspc} problem of cost $c$;
the solution also provides a coverage profile, $k_e$ for each edge $e$ (such that $\sum_e k_e = k$).
We produce a feasible solution $S'=(R',J')$ to the instance $\cB$ with the same cost $c$.
For each long resource picked by $S$, we retain the corresponding
wide resource in $R'$ (maintaining the number of copies).
Consider any edge $e$ in the {\lspc} instance and let $M$ be the corresponding mountain in the instance $\cB$.
The solution $S$ contains at most one short resource $i_s^{e,k_e'}$ containing $e$ of capacity $k_e'=w_{i_s}$.
Consider the multiset of short resources $R'$ in the instance $\cB$ associated with the resource $i_s^{e,k_e'}$.
The multiset $R'$ covers a set of $k_e'$ jobs contained in the mountain $M$.
Include all these $k_e'$ jobs in $J'$. Choose any other $k_e-k_e'$ jobs contained in $M$
and add these to $J'$; notice that the wide resources retained in $R'$ can cover these jobs.
This way we get a solution $S'$ for the instance $\cB$.
Cost of the solution $S'$ is at most the cost of $S$.
\end{proof}

\noindent
{\it Proof of Theorem \ref{thm:xCCC}: }
By composing the reductions given in the two steps,
we get a reduction from the {\PResAll} problem on a single mountain range to the {\lspc} problem.
The first step and the second step incur a loss in approximation of $3$ and $8$, respectively.
Thereby, the combined reduction incurs a loss of $24$.
Theorem \ref{thm:xEEE} provides a $16$-approximation algorithm for the {\lspc} problem.
Combining the reduction and the above algorithm, we get an algorithm for the {\PResAll} for a single mountain
range with an approximation ratio of $16\times 24 = 384$.

Note that the running time of the algorithm depends on $\max_{e \in E} d_e$. We can assume that $d_e$ is polynomially bounded for all $e \in E$, because initially all demands are 1 and so resources must have polynomially bounded capacity. Hence, the algorithm runs in polynomial time.

\section{Overall Algorithm}
Now that we have completed the description of the algorithm, we give an overall review of the algorithm.

\begin{enumerate}
	\item Use the decomposition \autoref{lem:XXX} to partition the input jobs into a set of mountain ranges.
	\item We obtain a constant factor approximation algorithm where the input jobs form a mountain.
	\item We then extend this result to a mountain range by reducing the problem to the {\lspc} problem.
	\item We extend this to several mountain ranges by using dynamic programming.
\end{enumerate}

\section{The {\PCResAll} problem}
In this section, we consider the {\PCResAll} problem. We prove the following:

\begin{theorem}
\label{DDD}
There is a $4$-factor approximation algorithm for the {\PCResAll} problem.
\end{theorem}
The proof proceeds by exhibiting a reduction from the {\PCResAll} problem to
the following full cover problem.

\noindent
{\it  Problem Definition:} We are given a demand profile which specifies an integral demand $d_e$ for each edge $e$.
The input resources are of two types, called S-type (short for single) and M-type (short for multiple).
A resource $i$ has
a capacity $w_i$, and cost $c_i$. A valid solution consists of a multiset of resources such that
it includes at most $1$ copy of any S-type resource; however arbitrarily many copies of any M-type resource may be picked. A feasible solution $S$ is a valid solution such that for any edge $e$, the total
capacity of the resources in $S$ containing $e$ is at least the demand $d_e$ of the edge $e$.
The objective is to find a feasible solution having minimum cost.
We call this problem the Single Multiple Full Cover ({\smfc}) problem.

The full cover problem, {\ZeroOneResAll} is considered in \cite{esa2011}.
The {\ZeroOneResAll} problem specifies demands  for
edges, and a feasible solution consist of a set of resources such that the demand of every edge is
fulfilled by the cumulative capacity of the resources containing that edge. The main
qualification is that in this problem setting, any resource may be picked up {\em at most
once}. In \cite{esa2011},  it is shown that this problem admits a $4$-factor
approximation algorithm. The {\smfc} problem easily reduces to the {\ZeroOneResAll} problem:
{\em S-type} resources may be picked up at most once, and keep copies of the {\em M-type}
resources so that it suffices to select any one of the copies.
Thus the algorithm and the performance guarantee claimed in \cite{esa2011} also implies the following:
\begin{theorem}\label{EEE}
There is a $4$-factor approximation to the {\smfc} problem.
\end{theorem}

We proceed to exhibit our reduction from the {\PCResAll} problem to the {\smfc} problem.
Given an instance $\I$ of the {\PCResAll} problem, we will construct an instance $\O$ of the
{\smfc} problem, such that any optimal solution $\opt(\I)$ can be converted (at no extra
cost) into an optimal solution $\opt(\O)$ for the instance $\O$. Consider any job $j$
in the instance $\I$; we will create a S-type resource $r(j)$ in the instance $\O$ corresponding to
$j$. The resource $r(j)$ will have the same length, left and right end-points as those of the job $j$, and will
have a cost $p_j$ (the penalty associated with job $j$). The resources in instance $\I$
will be labeled as M-type resources in the instance $\O$. The other parameters, such as
demands of edges, are inherited by $\O$ from the instance $\I$.

We show that
any feasible solution $S_\I$ to the {\PCResAll} problem corresponds to a feasible solution $S_\O$ (of the same cost) for the {\smfc} problem. Let ${\cal J}'$ denote the set of jobs that are not covered by the
solution $S_\I$ (thus, the solution pays the penalty for each of the jobs in $\cJ'$).

The multiset of resources in $S_\O$ consists of the (M-type) resources that exist in the solution $S_\I$, and
the S-type resources $r(j)$ in $\O$ corresponding to every job $j$ in $\cJ'$.
Any job $j$ that is actually covered by the set of resources in $S_\I$ is also covered
in the solution $S_\O$, and the resources utilized to cover the job are the same. A job $j$ that is
not covered by the resources in $S_\I$ pays a penalty $p_j$ in the solution $S_\I$; however this
job $j$ in $\O$ can be covered by the S-type resource $r(j)$ in the solution
$S_\O$. Thus, the solution $S_\O$ is a feasible solution to the instance $\O$, and has cost
equal to the cost of the solution $S_\I$.

In the reverse direction, suppose we are given a solution $S_\O$ to the instance $\O$.
We will convert the solution into a {\em standard} form, i.e. a solution in which if a
S-type resource $r(j)$ (for some job $j$) is included, then this resource is used to
cover job $j$. Suppose job $j$ is covered by some other resources in the solution
$S_\O$, while resource $r(j)$ covers some other jobs (call this set $J'$). We can clearly {\em exchange}
the resources between job $j$ and the set of jobs $J'$ so that job $j$ is covered by
resource $r(j)$. So we may assume that the solution $S_\O$ is in standard form.
But now, given a standard form solution $S_\O$, we can easily construct a
feasible solution $S_\I$ for the {\PCResAll} instance $\I$: if a job $j$ in $S_\O$
is covered by the S-type resource $r(j)$, then in $S_\I$, this job will not be
covered (and a penalty $p_j$ will be accrued); all jobs $j$ in $S_\O$ that are
covered by M-type resources will be covered by the corresponding resources
in $S_\I$.

This completes the reduction, and the proof of Theorem~\ref{DDD}.

\chapter {Conclusion and Open Problems}
\label {chap6}

In this thesis, we presented several algorithms for solving the {\ufpround}, {\ufpmax} and {\ufpbag} problems on paths and trees. We saw that some special cases of the {\ufpround} problem can have much better algorithms. We also showed how an algorithm for {\ufpround} can be used to solve the {\ufpmax} and {\ufpbag} problems. The idea of convex decomposition of fractional LP solutions is useful for this. We gave improved constant factor approximation algorithms for all these problems under the \emph{no bottleneck assumption}. We also studied the {\oic} problem and gave a constant factor competitive algorithm. Finally, we studied the {\PResAll} and the {\PCResAll} problems and gave $O(\log (n+m))$-approximation and $4$-approximation algorithms for them. There are several areas where there is a scope for improvements. We discuss some of them below.

For {\ufpround} on paths, we gave a 3-approximation algorithm for the case of uniform capacities. This algorithm requires $4r$ colors, where $r$ is the maximum congestion. However, we don't know of any example where the optimum coloring requires more than $2r$ colors. Moreover, our greedy algorithm when directly applied (without partitioning into small and large demands) also requires at most $2r$ colors on all examples that we have tried. It will be good to prove that this (or some other algorithm) requires at most $2r$ colors or prove that there is an example which requires more than $2r$ colors.

For arbitrary capacities and demands with NBA, we believe that the 24-approximation algorithm can be improved significantly. Again there is no example where the optimum coloring requires more than $2r$ colors. To improve the constant factor (24), we may need to consider $\frac{1}{2}$-small and $\frac{1}{2}$-large demands. It may also be the case that if we don't divide the demands into these two classes, a much better approximation is possible. But we need some new techniques for doing this.

Improving the $(2+\epsilon)$-approximation for {\ufpmax} with NBA is a formidable challenge. If we follow the small and large demands paradigm, to get a 2-approximation we need to have optimal solutions for both these instances, which is not possible for small demands (as it is NP-hard). So, we have to consider the demands together. Here, some new ideas are required to handle them together, as the existing techniques don't work well for these two classes.

For {\ufpbag}, improving the 65-approximation should not be very difficult. Again, considering $\frac{1}{2}$-small and $\frac{1}{2}$-large demands can be useful here. Moreover, we are using the approximation algorithm for \emph{throughout maximization for real-time scheduling} as a black box. If we can directly attack the problem, a much better approximation is possible.

For {\ufpround} and {\ufpmax} on trees, if we can use the tree structure more effectively, instead of breaking it into two paths and thereby losing a factor of 2, a better approximation is possible. A possible approach could be to consider the requests based on the depth of the least common ancestor (LCA) of the source and destination of a request.

For the {\oic} problem on paths with arbitrary capacities and arbitrary demands with NBA, designing an algorithm with a small constant approximation factor would be a significant challenge. The best lower bound for this problem with uniform capacities and arbitrary demands is $\frac{24}{7} \approx 3.43$ by Epstein et al. \cite{EpsteinL08}, improving the lower bound of 3 by Kierstead and Trotter for unit capacities and unit demands. Clearly, there is a big gap between the upper and lower bounds which needs to be closed. For trees, closing the gap between the upper bound of $O(\log n)$ and the lower bound of $\Omega\left(\frac{\log n}{\log \log n}\right)$ is a long-standing open problem.

A far more challenging task is to design good approximation algorithms for these problems without NBA. For {\ufpmax} on paths, a breakthrough was achieved when a $(7+\eps)$-approximation was given by \cite{BonsmaSW11}. To do this, they had to introduce new techniques, one of which is a novel geometric dynamic programming algorithm for the maximum weight independent set of rectangles problem. Since the congestion bound $r$ is very bad without NBA, for {\ufpround} we need significantly new ideas. A combination of the congestion bound $r$ and clique bound $\omega$ may do the job. We may also require a completely new and better lower bound.

For the {\PResAll} problem, the main goal is to either come up with a constant factor approximation algorithm, or to show that none exists by establishing a matching lower bound. One way to design the former is to design a constant factor approximation algorithm for the {\PCResAll} problem having the {\em Lagrangian Multiplier Preserving} property. Note that by using the Jain-Vazirani framework, we can immediately obtain a constant factor approximation algorithm for the {\PResAll} problem. It is also not clear whether the factors $O(\log n)$ and 4 for the {\PResAll} and {\PCResAll} problems respectively are the best possible.

Here are some future directions and open questions for these problems.

\begin{itemize}
	\item Is there a 2-approximation algorithm for {\ufpround} with uniform capacities?
	\item Can we improve the approximation factor of {\ufpround}, {\ufpmax} and {\ufpbag} problems on paths and trees?
	\item What is the approximability of these problems without the {\nba}? For {\ufpmax} on paths, a $(7+\eps)$-approximation is known.
	\item Is there a better constant factor competitive algorithm for the {\oic} problem on paths?
	\item For the {\oic} problem on trees, is it possible to close the gap between the upper bound of $O(\log n)$ and the lower bound of $\Omega\left(\frac{\log n}{\log \log n}\right)$?
	\item Is there a constant factor approximation algorithm for the {\PResAll} problem?
	\item Is there a constant factor approximation algorithm for the {\PCResAll} problem having the {\em Lagrangian Multiplier Preserving} property?
	\item What is the hardness of approximation of these problems?
\end{itemize}

\bibliographystyle{plain}
\bibliography{thesis}

\thispagestyle{plain}

\begin{center}
{\large \textbf{Biography of the Author}} 
\end{center}
\ \\
Arindam Pal completed Bachelor of Engineering from Jadavpur University, Kolkata in 2000 and Master of Engineering from Indian Institute of Science, Bangalore in 2002, both in Computer Science and Engineering. He worked as a Software Engineer in Microsoft and Yahoo! from February 2002 to July 2007. From August 2007 to November 2012, he worked for his  Ph.D. degree at the Department of Computer Science and Engineering, IIT Delhi. He is currently working as a Research Scientist at TCS Innovation Labs Kolkata. His research areas are approximation algorithms, combinatorial optimization, graph theory and machine learning.

\end{document}